\documentclass[draft,a4paper,UKenglish,thm-restate,autoref,numberwithinsect]{lipics-v2019}
\nolinenumbers %% SM: for taking out those line numbers
\usepackage{microtype}

\usepackage{enumitem}
\usepackage{xspace}
\usepackage{xcolor}
\usepackage{soul}
\usepackage[utf8]{inputenc}
\usepackage{amsmath,amssymb,mathrsfs}
\usepackage{stmaryrd} %\ll/rrbracket
\usepackage{wasysym}

\usepackage{etex,etoolbox}

\newcommand{\resetCurThmBraces}{%
\gdef\curThmBraceOpen{(}%
\gdef\curThmBraceClose{)}}
\resetCurThmBraces
\newcommand{\removeThmBraces}{%
\gdef\curThmBraceOpen{}%
\gdef\curThmBraceClose{}}
\resetCurThmBraces

\newenvironment{notheorembrackets}{\removeThmBraces}{\resetCurThmBraces}

\usepackage{etoolbox}
\patchcmd{\thmhead}{(#3)}{\curThmBraceOpen #3\curThmBraceClose}{}{}

\usepackage{hyperref}
\hypersetup{hidelinks,final,bookmarks}

\usepackage[notref,notcite,final]{showkeys}
\usepackage[noadjust]{cite}

\usepackage{seqsplit}
\usepackage{xstring}
\newcommand{\defaultshowkeysformat}[1]{%
\StrSubstitute{#1}{ }{\textvisiblespace}[\TEMP]%
\parbox[t]{\marginparwidth}{\raggedright\normalfont\small\ttfamily\(\{\){\color{red!50!black}\expandafter\seqsplit\expandafter{\TEMP}}\(\}\)}%
}

\renewcommand*\showkeyslabelformat[1]{%
\noexpandarg%
\defaultshowkeysformat{#1}%
}

\newcommand{\midmid}{\hspace{0.2ex}{\rule[-0.1ex]{0.6pt}{1.65ex}}\hspace{0.2ex}}
\newcommand{\scriptmidmid}{\hspace{0.2ex}{\rule[-0.1ex]{0.6pt}{1.1ex}}\hspace{0.2ex}}

\newcommand{\newletter}[1]{{\midmid}#1}
\newcommand{\scriptnew}[1]{{\scriptmidmid}#1}

\newcommand{\trans}[1]{\xrightarrow{#1}}
\newcommand{\N}{\mathsf{N}}
\newcommand{\FN}{\mathsf{FN}}
\newcommand{\BN}{\mathsf{BN}}
\newcommand{\id}{\mathit{id}}

\newcommand{\prelang}{L_\mathsf{pre}}

\newcommand{\free}[1]{\mathsf{bs}(#1)}
\newcommand{\formulae}{\mathsf{Bar}}
\newcommand{\negepsilon}{\neg\epsilon}
\newcommand{\degree}{\mathsf{deg}}

\newcommand{\barstrings}{\barA^*}
\newcommand{\barA}{\overline{\names}}
\newcommand{\barname}{\sigma}
\newcommand{\ub}{\mathsf{ub}}

\numberwithin{equation}{section}

\usepackage[author=anonymous,marginclue,footnote]{fixme}
\FXRegisterAuthor{dh}{adh}{DH}
\FXRegisterAuthor{ls}{als}{LS}
\FXRegisterAuthor{sm}{asm}{SM}
\FXRegisterAuthor{tw}{atw}{TW}

\usepackage{algorithmicx}
    \usepackage[ruled]{algorithm}
\usepackage[noend]{algpseudocode}

\usepackage{tikz}
 \usetikzlibrary{trees}
 \usetikzlibrary{shapes}
 \usetikzlibrary{fit}
 \usetikzlibrary{shadows}
 \usetikzlibrary{backgrounds}
 \usetikzlibrary{arrows,automata}

\tikzset{
   n/.style= {circle,fill,inner sep=1.5pt,node distance=2cm}
  ,acc/.style={circle,draw,inner sep=3pt,node distance=2cm}
  ,phantom/.style={circle},
  ,arr/.style={->, >=stealth, semithick, shorten <= 3pt, shorten >= 3pt}
}

\newcommand{\CO}{\mathcal{O}}

\newcommand{\muBar}{\textsf{Bar-$\mu$TL}\xspace}
\newcommand{\states}{\mathsf{mstates}}
\newcommand{\prestates}{\mathsf{pstates}}
\newcommand{\quasistates}{\mathsf{qstates}}
\renewcommand{\Box}{\square}
\renewcommand{\Diamond}{\lozenge}

\newcommand{\hearts}{\heartsuit}

\newcommand{\sem}[1]{\llbracket #1 \rrbracket}

\makeatletter
\def\moverlay{\mathpalette\mov@rlay}
\def\mov@rlay#1#2{\leavevmode\vtop{%
   \baselineskip\z@skip \lineskiplimit-\maxdimen
   \ialign{\hfil$\m@th#1##$\hfil\cr#2\crcr}}}
\newcommand{\charfusion}[3][\mathord]{
    #1{\ifx#1\mathop\vphantom{#2}\fi
        \mathpalette\mov@rlay{#2\cr#3}
      }
    \ifx#1\mathop\expandafter\displaylimits\fi}
\makeatother

\newcommand\NExpTime{$\textsc{NExpTime}$\xspace}
\newcommand\NExpSpace{$\textsc{NExpSpace}$\xspace}

\newcommand{\takeout}[1]{\empty}
\newcommand{\names}{\mathbb{A}}
\newcommand{\fix}{\mathop{\mathsf{fix}}}
\newcommand{\Fix}{\mathop{\mathsf{Fix}}}
\newcommand{\supp}{\mathop{\mathsf{supp}}}
\newcommand{\clos}{{\mathsf{cl}}}
\newcommand{\clo}{{\mathsf{Cl}}}

\makeatletter
\newcommand{\mysubsec}[1]{%
  \par\medskip\noindent{\bfseries\sffamily #1}\hspace{3mm}%
  \@ifnextchar\par{\@gobble}{}% this eats a following \par if any
}
\makeatother
\theoremstyle{definition}
\newtheorem{defn}[theorem]{Definition}
\newtheorem{expl}[theorem]{Example}

\newtheorem*{fact}{Fact}

\newcommand\ExpSpace{$\textsc{ExpSpace}$\xspace}
\newcommand\PSpace{$\textsc{PSpace}$\xspace}

\title{A Linear-Time Nominal $\mu$-Calculus\\ with Name Allocation}
\titlerunning{A Linear-Time Nominal $\mu$-Calculus with Name Allocation}%optional, please use if title is longer than one line
\author{Daniel Hausmann}%
{Gothenburg University, Sweden}%
{daniel.hausmann@fau.de}%
{https://orcid.org/0000-0002-0935-8602}
{Supported by Deutsche Forschungsgemeinschaft (DFG) under project MI~717/7-1
and by the European Research Council (ERC) under project 772459}
\author{Stefan Milius}%
{Friedrich-Alexander-Universit\"at Erlangen-N\"urnberg, Germany}%
{mail@stefan-milius.eu}%
{https://orcid.org/0000-0002-2021-1644}
{Supported by Deutsche Forschungsgemeinschaft (DFG) under project MI~717/7-1}% and as part of the Research and Training Group 2475 ``Cybercrime and Forensic Computing'' (393541319/GRK2475/1-2019)}
\author{Lutz Schr\"oder}{Friedrich-Alexander-Universit\"at Erlangen-N\"urnberg, Germany}%
{lutz.schroeder@fau.de}%
{https://orcid.org/0000-0002-3146-5906}
{Supported by Deutsche Forschungsgemeinschaft (DFG) under project \mbox{SCHR~1118/15-1}}
\authorrunning{D.~Hausmann, S.~Milius, L.~Schr\"oder}%TODO mandatory. First: Use abbreviated first/middle names. Second (only in severe cases): Use first author plus 'et al.'
\Copyright{Daniel Hausmann, Stefan Milius, Lutz Schr\"oder}%TODO mandatory, please use full first names. LIPIcs license is "CC-BY";  http://creativecommons.org/licenses/by/3.0/
%TODO mandatory; please add comma-separated list of keywords
\hideLIPIcs  %uncomment to remove references to LIPIcs series (logo, DOI, ...), e.g. when preparing a pre-final version to be uploaded to arXiv or another public repository

\EventEditors{Filippo Bonchi and Simon J. Puglisi}
\EventNoEds{2}
\EventLongTitle{46th International Symposium on Mathematical Foundations of Computer Science (MFCS 2021)}
\EventShortTitle{MFCS 2021}
\EventAcronym{MFCS}
\EventYear{2021}
\EventDate{August 23--27, 2021}
\EventLocation{Tallinn, Estonia}
\EventLogo{}
\SeriesVolume{202}
\ArticleNo{81}

\relatedversion{}
\begin{document}
\maketitle
\begin{abstract}
  Logics and automata models for languages over infinite alphabets,
  such as Freeze LTL and register automata, serve the verification of
  processes or documents with data. They relate tightly to formalisms
  over nominal sets, such as nondetermininistic orbit-finite automata
  (NOFAs), where names play the role of data. Reasoning problems in
  such formalisms tend to be computationally
  hard. % For instance, inclusion checking of
  % unrestricted register automata (equivalently NOFAs) is
  % undecidable. Similarly, satisfiability in full Freeze LTL is
  % undecidable, and decidable but not primitive recursive if the
  % number
  % of registers is limited to at most one.
  \emph{Name-binding} nominal automata models such as \emph{regular
    nondeterministic nominal automata (\mbox{RNNAs})} have been shown to be
  computationally more tractable. In the present paper, we introduce a
  linear-time fixpoint logic $\muBar$ for finite words over an infinite
  alphabet, which features full negation and freeze quantification via
  name binding. We show by a nontrivial reduction to \emph{extended
    regular nondeterministic nominal automata} that even though
  $\muBar$ allows unrestricted nondeterminism and unboundedly many
  registers, model checking $\muBar$ over RNNAs and satisfiability
  checking both have elementary complexity. For example, model checking is in
  $2\ExpSpace$, more precisely in parametrized $\ExpSpace$, effectively
  with the number of registers as the parameter.
\end{abstract}
\section{Introduction}
There has been longstanding interest in logics and automata models
over infinite alphabets, such as the classical register automaton
model~\cite{KaminskiFrancez94} and Freeze LTL
(e.g.~\cite{Lazic06,DemriLazic09}), or automata models over nominal
sets~\cite{Pitts13} such as nondeterministic orbit-finite
automata~\cite{BojanczykEA14}, in which \emph{names} play the role of
letters. Infinite alphabets may be seen as representing data. For
example, nonces in cryptographic protocols~\cite{KurtzEA07}, data
values in XML documents~\cite{NevenEA04}, object
identities~\cite{GrigoreEA13}, or parameters of method
calls~\cite{HowarEA19} can all be usefully understood as letters in
infinite alphabets. A central challenge in dealing with infinite
alphabets is that key decision problems in many logics and automata
models are either undecidable or of prohibitively high complexity
unless drastic restrictions are imposed (see the related work
section). In a nutshell, the contribution of the present work is the
identification of a linear-time fixpoint logic $\muBar$ for \emph{finite}
words over infinite alphabets that
\begin{itemize}
\item allows both safety and liveness constraints (via fixpoints, thus
  in particular going beyond the expressiveness of
  LTL~\cite{DeGiacomoVardi13}) as well as full nondeterminism, and
  e.g.~expresses the language `some letter occurs twice' (which
  cannot be accepted by deterministic or unambigous register
  automata~\cite{BojanczykEA14,SchroderEA17-full});
\item imposes no restriction on the number of registers; and
\item is closed under complement;
\end{itemize}
and nevertheless allows model checking and satisfiability checking in
elementary complexity: $\muBar$ model checking over regular
nondeterministic nominal automata~\cite{SchroderEA17} (in the sense of
checking that all words accepted by a given automaton satisfy a given
formula) is in $2\ExpSpace$ and more precisely in parametrized
$\ExpSpace$, with the maximal size of the support of states as the
parameter (in the translation of nominal automata to register
automata~\cite{BojanczykEA14,SchroderEA17}, this corresponds to the
number of registers); and satisfiability checking is in $\ExpSpace$
(and in parametrized $\PSpace$).

The tradeoff that buys this comparatively low complexity is a mild
recalibration of the notion of freshness, which we base on explicit
binding of names in strings in a nominal language model; depending on
the exact discipline of $\alpha$-renaming of bound names, one obtains
either global freshness (w.r.t.\ all previous letters, as in session
automata~\cite{BolligEA14}) or local freshness (w.r.t.\ currently
stored letters, as in register automata). This principle has been
previously employed in the semantics of nominal Kleene
algebra~\cite{GabbayCiancia11,KozenEA15} and in regular non-deterministic
nominal automata (RNNAs)~\cite{SchroderEA17}. It carries the
limitation that in the local freshness variant, letters can be
required to be distinct only from such previous letters that are expected to be
seen again -- a restriction that seems reasonable in applications;
that is, in many situations, one presumably would not need to insist
on an object identifier, nonce, or process name~$b$ to be distinct
from a previous name~$a$ if~$a$ is never going to be used again.

We introduce a dedicated finite-word automaton model, \emph{extended
  regular nondeterministic nominal automata (ERNNAs)}, which extend
RNNAs by \emph{$\top$-states}, i.e.~deadlocked universal states. We
then base our results on mutual translations between~$\muBar$ and
\mbox{ERNNAs}; the tableau-style logic-to-automata translation turns out to
be quite nontrivial as it needs to avoid the accumulation of renamed
copies of subformulae in automata states.
 % , and showing
% that language inclusion of RNNA in ERNNA is in parametrized
% $\PSpace$. The satisfiability problem of $\muBar$ reduces to
% non-emptiness checking for ERNNA, which we show to be in $\PSpace$,
% while the validity problem of $\muBar$ is in $\ExpSpace$ by reduction
% to the model checking problem. Like RNNAs, $\muBar$ in fact leads a
% second life as a formalism for global freshness (indeed, the
% difference between global and local freshness is determined simply by
% whether one does or does not insist on clean $\alpha$-renaming of bar
% strings). The above results on model checking and satisfiability
% checking (although not the one for validity checking) hold in
% essentially the same way also for this global freshness
% interpretation, so that $\muBar$ may also be seen as a logical
% language for session automata
Since RNNAs are, under global freshness, essentially equivalent to
session automata~\cite{BolligEA14}, one consequence of our results is
that session automata (which as such are only closed under
\emph{resource-bounded} complementation~\cite{BolligEA14}) can be made
closed under complement simply by allowing $\top$-states.

Proofs are mostly omitted or only sketched; full proofs are in the
extended version~\cite{HausmannEA20}.%

\mysubsec{Related work}
% The
% situation is better for deterministic or unambiguous
% models~\cite{,MottetQuaas19}, which however have markedly
% lower expressiveness; e.g.~the language `some letter appears twice'
% can be accepted by a non-deterministic register automaton but not by
% any deterministic or unambiguous
% one~\cite{BojanczykEA14,SchroderEA17-full}.
  % ), and as
% such is, e.g., closed under relative complement within universal
% languages allowing only a fixed finite set of non-fresh letters
% (i.e.~of free names in bar strings)

\emph{Automata:} Over infinite alphabets, the expressive power of
automata models generally increases with the power of control
(deterministic/nondeterministic/alternating)~\cite{KlinEA21}. In
deterministic models, language inclusion can often be decided in
reasonable complexity; this remains true for unambiguous register
automata~\cite{MottetQuaas19,Colcombet15}.  For nondeterministic
register automata and the equivalent nondeterministic orbit-finite
automata~\cite{BojanczykEA14}, emptiness is decidable but inclusion is
undecidable unless one restricts to at most two
registers~\cite% [Appendix~A]
{KaminskiFrancez94}. Similarly, language inclusion (equivalently
nonemptiness) of alternating register automata is undecidable unless
one restricts to at most one register, and even then is not primitive
recursive~\cite{DemriLazic09}.  Automata models for infinite words
outside the register paradigm include data walking
automata~\cite{ManuelEA16}, whose inclusion problem is decidable even
under nondeterminism but at least as hard as reachability in Petri
nets (equivalently vector addition
systems)~\cite{ColcolmbetManuel14arXiv}, as well as the highly
expressive data automata~\cite{BojanczykEA11}, whose nonemptiness
problem is decidable but, again, at least as hard as Petri net
reachability. Note that by recent results~\cite{CzerwinskiEA21}, Petri
net reachability is not elementary, and in fact
Ackermann-complete~\cite{Leroux21,CzerwinskiOrlikowski21}.
% The decidability of the nonemptiness problem of alternating
% one-register automata~\cite{DemriLazic09} (in the mentioned high
% complexity) is preserved under extensions with nondeterministic update
% and a form of quantification over data
% values~\cite{Figueira12}. 
 % The
% principles of name binding automata and name dropping that we employ
% in the present work go back to work on regular nondetermistic nominal
% automata (RNNAs)~\cite{SchroderEA17}. RNNAs have an alternative global
% freshness semantics, and under this semantics are essentially
% equivalent to session automata~\cite{BolligEA14}. ERNNAs with global
% semantics are slightly more expressive than session automata, as they
% can express the universal language. (E)RNNA natively accept strings
% with name binding, so-called bar strings, which are related both to
% dynamic sequences~\cite{GabbayEA15} and to linear expressions in
% nominal Kleene algebra~\cite{GabbayCiancia11}.

\emph{Logics:} $\muBar$ is incomparable to Freeze LTL. Satisfiability
checking of Freeze LTL is by reduction to alternating register
automata, and has the same high complexity even for the one-register
case~\cite{DemriLazic09}. Satisfiability in the \emph{safety} fragment
of Freeze LTL over infinite words~\cite{Lazic06} is
$\ExpSpace$-complete if the number of registers is restricted to at
most one, while the refinement and model checking problems are
decidable but not primitive recursive (all three problems become
undecidable in presence of more than one register). The
\emph{$\mu$-calculus with
  atoms}~\cite{KlinLelyk19} % has an infinitary but orbit-finite syntax
% and
is interpreted over Kripke models with atoms; its satisfiability
problem is undecidable while its model checking problem is decidable,
with the complexity analysis currently remaining open. It is related
to the very expressive \emph{first-order
  $\mu$-calculus}~\cite{GrooteMateescu98,GrooteWillemse05}, for which
model checking is only known to be semidecidable. The
\emph{$\mu$-calculus over data
  words}~\cite{ColcombetManuel14,ColcombetManuel15} works in the data
walking paradigm. The satisfiability problem of the full calculus is
undecidable; that of its $\nu$-fragment, which translates into data
automata, is decidable but elementarily equivalent to Petri net
reachability. \emph{Variable LTL}~\cite{GrumbergEA12} extends LTL with
a form of first-order quantification over data domains. The full
language is very expressive and in particular contains full Freeze
LTL~\cite{SongWu16}. Some fragments have decidable satisfiability or
model checking problems (typically not
both)~\cite{GrumbergEA12,SongWu16}, occasionally in elementary
complexity; these impose prenex normal form (reducing expressiveness)
and restrict to quantifier prefixes that fail to be stable under
negation. Decidable fragments will, of course, no longer contain full
Freeze LTL; how their expressiveness compares to $\muBar$ needs to be
left open at present. \emph{Flat} Freeze LTL~\cite{DemriSagnier10}
interdicts usage of the freeze quantifier in safety positions (e.g.\
the freeze quantifier can occur under~$\mathsf{F}$ but not
under~$\mathsf{G}$); its existential model checking problem over
(infinite runs of) one-counter automata is $\NExpTime$-complete, while
the universal model checking problem is undecidable~\cite{BolligEA19}.

% \cite{ParrowEA15} Eher nicht so relevant, da nicht temporal

\section{Preliminaries: Nominal Sets}
\label{sec:prelim}
\noindent Nominal sets offer a convenient formalism for dealing with
names and freshness; for our present purposes, names play the role of
data. We briefly recall basic notions and facts (see~\cite{Pitts13}
for more details).

Fix a countably infinite set $\names$ of \emph{names}, and let $G$
denote the group of finite permutations on $\names$, which is
generated by the \emph{transpositions} $(a\,b)$
for $a\neq b\in\names$ (recall that $(a\,b)$ just swaps~$a$ and~$b$);
we write~$1$ for the neutral element of~$G$. A (left) \emph{action}
of~$G$ on a set~$X$ is a map $(-)\cdot(-)\colon G\times X\to X$ such
that $1\cdot x=x$ and $\pi\cdot(\pi'\cdot x)=(\pi\pi')\cdot x$ for all
$x\in X$, $\pi,\pi'\in G$. For instance,~$G$ acts on~$\names$ by
$\pi\cdot a = \pi(a)$. A set $S\subseteq\names$ is a \emph{support}
of~$x\in X$ if $\pi(x)=x$ for all $\pi\in G$ such that $\pi(a)=a$ for
all $a\in S$. We say that~$x$ is \emph{finitely supported} if~$x$ has
some finite support, and \emph{equivariant} if the empty set is a
support of~$x$, i.e.~$\pi\cdot x=x$ for all~$\pi\in G$. The
\emph{orbit} of $x \in X$ is the set $\{\pi \cdot x \mid \pi \in
G\}$. The set of all orbits forms a partition of $X$, and $X$ is
\emph{orbit-finite} if there are finitely many orbits.

A \emph{nominal set} is a set~$X$ equipped with an action of~$G$ such
that every element of~$X$ is finitely supported; e.g.~$\names$ itself
is a nominal set. An early motivating example of a nominal set is the
set of $\lambda$-terms with variable names taken from~$\names$, and
with the action of~$\pi\in G$ given by replacing every variable
name~$a$ with~$\pi(a)$ as expected, e.g.\
$(a\,b)\cdot \lambda a.ba=\lambda b.ab$. Every element~$x$ of a
nominal set has a least finite support, denoted $\supp(x)$, which one
may roughly think of as the set of names occurring (`freely', see
below) in~$x$. A name $a\in\names$ is \emph{fresh} for~$x$ if
$a\notin\supp(x)$. On subsets $A\subseteq X$ of a nominal
set~$X$, $G$ acts by $\pi\cdot A = \{\pi\cdot x \mid x \in A\}$.
Thus, $A\subseteq X$ is equivariant iff $\pi\cdot A\subseteq A$ for
all $\pi\in G$.  Similarly,~$A$ has support~$S$ iff
$\pi\cdot A\subseteq A$ whenever~$\pi(a)=a$ for all $a\in S$. We say
that~$A$ is \emph{uniformly finitely supported} if
$\bigcup_{x \in A} \supp(x)$ is finite~\cite{TurnerWinskel09}, in
which case~$A$ is also finitely
supported~\cite[Thm.~2.29]{gabbay2011}. (The converse does not hold,
e.g.~the set $\names$ is finitely supported but not uniformly
finitely supported.) Uniformly finitely supported subsets of
orbit-finite sets are always finite but in general, uniformly finitely
supported sets can be infinite; e.g.~for finite $B\subseteq\names$,
the set $B^*\subseteq\names^*$ is uniformly finitely supported.

% \begin{defn}[Support-relative complement]\label{def:rel-comp}
%   Let $S\subseteq\names$. The \emph{$S$-relative complement} (or,
%   omitting mention of~$S$, \emph{support-relative complement}) of a
%   uniformly finitely supported subset~$Y$ of a nominal set~$X$ is the
%   uniformly finitely supported set
%   $\{x\in X\setminus Y\mid\supp(x)\subseteq S\}$.
% \end{defn}

% Recall that an \emph{action} of a group $G$ on a
% set $X$ is a map $G \times X \to X$, denoted by juxtaposition, such
% that $\pi(\rho x) = (\pi\rho)x$ and $1x = x$ for $\pi, \rho \in G$,
% $x \in X$. A \emph{$G$-set} is a set $X$ equipped with an action of
% $G$.
  % For $A\subseteq X$ and $x\in X$, we
% put
% \[
% \fix x = \{ \pi \in G \mid \pi x = x\} 
% \qquad
% \text{and}
% \qquad
% \textstyle\Fix A = \bigcap_{x \in A} \fix x.
% \]
% Note that elements of $\fix A$ and $\Fix A$ fix $A$ setwise and
% pointwise, respectively.

% We say that $a\in\names$ is \emph{fresh for $x$}, and
% write $a\fresh x$, if $a\notin\supp(x)$.
%Any ufs set is fs but not
%conversely; e.g.~the set $\names$ is fs but not ufs. Moreover, any
%finite subset of $X$ is ufs but not conversely; e.g.~the set of words
%$a^n$ for fixed $a\in\names$ is ufs but not finite.

The Cartesian product $X \times Y$ of nominal sets~$X,Y$ is a
nominal set under the compo\-nentwise group action; then,
$\supp(x,y) = \supp(x) \cup \supp(y)$. Given a nominal set $X$
equipped with an equivalence relation~$\sim$ that is equivariant as a
subset of $X \times X$, the quotient $X/\mathord{\sim}$ is a nominal
set under the group action $\pi \cdot [x]_\sim = [\pi \cdot x]_\sim$.
A key role in the technical development is played by \emph{abstraction
  sets}, which provide a semantics for binding
mechanisms~\cite{GabbayPitts99}:%
\begin{defn}[Abstraction set]\label{def:abstraction}
  Given a nominal set $X$, an equivariant equivalence relation $\sim$
  on $\names \times X$ is defined by
  \begin{equation*}
    (a,x)\sim (b,y)\qquad\text{iff}\qquad
    \text{$(a\, c)\cdot x=(b\, c)\cdot y$ for some~$c\in\names$ that is fresh
  for $(a,x,b,y)$}
\end{equation*}
(equivalently for all such~$c$). The
  \emph{abstraction set} $[\names]X$ is the quotient set
  $(\names\times X)/\mathord{\sim}$. The $\sim$-equivalence class of
  $(a,x)\in\names\times X$ is denoted by
  $\langle a\rangle x\in [\names]X$.
\end{defn}
\noindent
We may think of~$\sim$ as an abstract notion of $\alpha$-equivalence,
and of~$\langle a\rangle$ as binding the name~$a$. Indeed we have
$\supp(\langle a\rangle x)= \supp(x)\setminus\{a\}$, as expected in
binding constructs.

% In the semantics of~$\muBar$, we will need to take fixed points of
% monotone functions on nominal sets~$\powfs(X)$. We write $\LFP f$ and
% $\GFP f$ for the least and greatest fixed point of a function~$f$, if
% these exist. Although $\powfs(X)$ is not (for topos theorists:
% externally) a complete lattice, we do have a nominal variant of the
% Knaster-Tarski theorem:
% \begin{theorem}\label{thm:fix} Let~$X$ be a
%   nominal set, and let $f\colon\powfs(X)\to\powfs(X)$ be finitely
%   supported and monotone. Then $\LFP f$ and $\GFP f$ exist and are
%   given by
%   \begin{equation*}
%     \LFP f = \bigcap\{Z\in\powfs(X)\mid f(Z)\subseteq Z\},
%     %\text{ and }
%     \quad
%     \GFP f = \bigcup\{Z\in\powfs(X)\mid Z\subseteq f(Z)\}.
%   \end{equation*}
% \end{theorem}

\section{Data Languages and Bar Languages}
\label{sec:databarl}
\noindent As indicated, we use the set~$\names$ of names as the data
domain, and capture freshness of data values via $\alpha$-equivalence,
roughly as follows. We work with words over~$\names$ where names may
be preceded by the bar symbol~`$\newletter$', which indicates that the
next letter is bound until the end of the word
(cf.~\autoref{rem:nka}); such words are called \emph{bar
  strings}~\cite{SchroderEA17}. Bar strings may be seen as patterns
that govern how letters are read from the input word; broadly
speaking, an occurrence of $\newletter a$ corresponds to reading a
letter from the input word, and binding this letter to the name~$a$,
while an undecorated occurrence of~$a$ means that the letter referred
to by~$a$ occurs in the input word. (We loosely speak of the input
word as consisting of letters, and of bar strings as consisting of
names; formally, however, letters and names are the same, viz.,
elements of~$\names$.) Bound names can be renamed, giving rise to a
notion of $\alpha$-equivalence; as usual, the new name needs to be
sufficiently \emph{fresh}, i.e.~cannot already occur freely in the
scope of the binding. For instance, in $a\newletter bab$, the
$\newletter$ binds the letter~$b$ in $\newletter bab$. The bar string
$a\newletter bab$ is $\alpha$-equivalent to $a\newletter cac$ but not,
of course, to $a\newletter aaa$. That is, we can rename the bound
name~$b$ into~$c$ but not into~$a$, as~$a$ already occurs freely in
$\newletter bab$; we say that renaming~$b$ into~$a$ is \emph{blocked}.

We will see that bar strings modulo $\alpha$-equivalence relate to
formalisms for \emph{global} freshness (a name is globally fresh if it
has never been seen before), such as session
automata~\cite{BolligEA14}. Contrastingly, if we regard a bar string
as representing all words over~$\names$ that arise by performing some
$\alpha$-equivalent renaming and then removing the bars, we arrive at
a notion of \emph{local} freshness, similar to freshness w.r.t.\
currently stored names as in register automata; precise definitions
are given later in this section. For instance, the bar string
$\newletter a\newletter b a$ represents the set of words
\begin{math}
  \{cdc\mid c,d\in\names, c\neq d\}\subseteq\names^*
\end{math}
under both local and global freshness semantics --
in~$\newletter a\newletter ba$, $a$ and $b$ cannot be renamed into the
same letter, since~$a$ occurs freely in the scope of~$\newletter
b$. Contrastingly, $\newletter a\newletter b$ represents the set
$\{cd\mid c\neq d\}\subseteq\names^*$ under global freshness
semantics, but under local freshness semantics it just represents the
set~$\names^2$ of all two-letter words, since
$\newletter a\newletter b$ is $\alpha$-equivalent to
$\newletter a\newletter a$. The impossibility of expressing the
language $\{cd\mid c\neq d\}$ under local freshness is thus hardwired
into our language model. We emphasize again that this restriction
seems reasonable in practice, since one may expect that freshness of
new letters is often relevant only w.r.t.\ letters that are intended
to be seen again later, e.g.~in deallocation statements or message
acknowledgements. Formal definitions are as follows.

\begin{defn}[Bar strings]
  We put $\barA=\names\cup\{\newletter a\mid a\in\names\}$; we refer
  to elements $\newletter a\in\barA$ as \emph{bar names}, and to
  elements $a\in\barA$ as \emph{plain names}.  A \emph{bar string}
  is a word $w=\sigma_1\sigma_2\cdots \sigma_n\in\barstrings$, with
  \emph{length} $|w|=n$; we denote the empty string by~$\epsilon$.
  %We sometimes treat~$w$ notationally as a (partial)
  %function $w\colon\mathbb{N}\rightharpoonup\barA$, writing
  %$w(i)=a_i$.  For a bar string $w$ and $i\leq |w|$, we define $w_i$
  %to be the postfix of $w$ starting at the $i$-th position, i.e.\
  %$w_i(j)=w(i+j-1)$ for $1\le j\le|w|-i+1$. 
  We turn~$\barA$ into a nominal set by putting $\pi\cdot a=\pi(a)$
  and $\pi\cdot\newletter a=\newletter\pi(a)$; then,~$\barstrings$ is
  a nominal set under the pointwise action of~$G$. We
  define~\emph{$\alpha$-equivalence} on bar strings to be the least
  equivalence $\equiv_\alpha$ such that
  \begin{equation*}
    w\newletter a v\equiv_\alpha w\newletter b u\quad\text{whenever}\quad
    \langle a\rangle v= \langle b\rangle u\text{ in $[\names]\,\barstrings$}
  \end{equation*}
  (\autoref{def:abstraction}) for $w,v,u\in\barstrings$,
  $a\in\names$. Thus, $\newletter a$ binds~$a$, with scope extending
  to the end of the word. Correspondingly, a name~$a$ is \emph{free}
  in a bar string~$w$ if there is an occurrence of~$a$ in~$w$ that is
  to the left of any occurrence of~$\newletter a$. We write
  $[w]_\alpha$ for the $\alpha$-equivalence class of $w\in\barstrings$
  and $\FN(w) = \{a\in\names\mid a\text{ is free in }w\}$
  ($=\supp([w]_\alpha)$) % \twnote{make this a lemma}
  for the set of free names of~$w$. If~$\FN(w)=\emptyset$, then~$w$ is
  \emph{closed}.  A bar string~$w$ is \emph{clean} if all bar
  names~$\newletter a$ in~$w$ are pairwise distinct and
  have~$a\notin\FN(w)$.  % A set of bar
  % strings is \emph{$\alpha$-closed} if it is closed under
  % $\alpha$-equivalence.
  For a set $S\subseteq \names$ of names, we
  write $\free{S}=\{w\in\barstrings\mid \FN(w)\subseteq S\}$.
\end{defn}
%
% \begin{remark}
%   Note that $\alpha$-equivalence is not a congruence w.r.t.\
%   concatenation from the right; e.g.\
%   $\newletter a\equiv_\alpha\newletter b$ but
%   $\newletter aa\not\equiv_\alpha\newletter ba$.
% \end{remark}
\begin{remark}\label{rem:nka}
  Closed bar strings are essentially the same as the \emph{well-formed
    symbolic words} that appear in the analysis of session
  automata~\cite{BolligEA14}. Indeed, symbolic words consist of
  operations that read a letter into a register, corresponding to bar
  names, and operations that require seeing the content of some
  register in the input, corresponding to plain names. Symbolic words
  are normalized by a register allocation procedure similar
  to $\alpha$-renaming. Well-formedness of symbolic words corresponds
  to closedness of bar strings.

  Moreover, modulo the respective equational laws, bar strings
  coincide with the \emph{$\nu$-strings}~\cite{KozenEA15b,KozenEA15}
  that appear in the semantics of \emph{Nominal Kleene Algebra
    (NKA)}~\cite{GabbayCiancia11}; cf.~\cite{SchroderEA17}. These are
  constructed from names in~$\names$, sequential composition, and a
  binding construct $\nu a.\,w$, which binds the name~$a$ in the
  word~$w$. In particular, the equational laws of $\nu$-strings allow
  extruding the scope of every~$\nu$ to the end of the word after
  suitable $\alpha$-renaming. We note that $\muBar$ and its associated
  automata models are more expressive than NKA as they express
  languages with unbounded nesting of binders~\cite{SchroderEA17}.
\end{remark}

\noindent We will work with three different types of languages:
\begin{defn}
\begin{enumerate}%[wide]
\item \emph{Data languages} are subsets of $\names^*$.
\item \emph{Literal languages} are subsets of $\barstrings$, i.e.\
  sets of bar strings.
\item \emph{Bar languages} are subsets of
  $\barstrings/{\equiv_\alpha}$, i.e.~sets of $\alpha$-equivalence
  classes of bar strings.\smallskip
\end{enumerate}
\noindent A bar language~$L$ is \emph{closed} if $\supp(L)=\emptyset$.
\end{defn}
Bar languages are the natural semantic domain of our formalims, and
relate tightly to data languages as discussed next. A key factor in
the good computational properties of regular nominal nondeterministic
automata (RNNA)~\cite{SchroderEA17} is that the bar languages they
accept (cf.~\autoref{sec:eauto}) are uniformly finitely supported, and
we will design~$\muBar$ to ensure the same property. Note that a
uniformly supported bar language is closed iff it consists of
(equivalence classes of) closed bar strings. For brevity, we will
focus the exposition on target formulae (in model checking) and
automata that denote or accept, respectively, closed bar languages,
with free names appearing only in languages accepted by non-initial
states or denoted by proper subformulae of the target formula. (The
treatment is easily extended to bar languages with free names; indeed,
such globally free names are best seen as a separate finite alphabet
of constant symbols~\cite{KozenEA15b}.)  We will occasionally describe
example bar languages as regular expressions over $\barA$ (i.e.~as
\emph{regular bar expressions}~\cite{SchroderEA17}), meaning the set
of all $\alpha$-equivalence classes of instances of the expression.

To convert bar strings into data words, we define
$\ub(a)=\ub(\newletter a)=a$ and extend $\ub$ to bar strings
letterwise; i.e.~$\ub(w)$ is the data word obtained by erasing all
bars `$\newletter$' from~$w$. We then define two ways to convert a bar
language~$L$ into a data language:
\begin{equation*}
  N(L)=\{\ub(w)\mid [w]_\alpha\in L, w\text{ clean}\}
  \qquad\text{and}\qquad
  D(L)=\{\ub(w)\mid [w]_\alpha\in L\}.
\end{equation*}
That is,~$N$ is a global freshness interpretation of~$\newletter$,
while~$D$ provides a local freshness interpretation as exemplified
above; e.g.~as indicated above we have
$D(\newletter a\newletter ba)=N(\newletter a\newletter ba)=\{aba\mid
a,b\in\names,a\neq b\}$, % (where
% we use regular expressions over~$\barA$ to describe bar languages as
% indicated above)
while $N(\newletter a\newletter b)=\{ab\mid a,b\in\names,a\neq b\}$
but $D(\newletter a\newletter b)=\{ab\mid a,b\in\names\}$.
\begin{remark}\label{rem:local-global}
  % We will equip $\muBar$ with a primary bar language semantics, from
  % which we will derive a global freshness semantics and a local
  % freshness semantics using~$N$ and~$D$, respectively.
  In fact, the operator~$N$ is injective on closed bar languages,
  because $\ub$ is injective on closed clean bar
  strings~\cite{SchroderEA17,SchroderEA17-full}. This means that bar
  language semantics and global freshness semantics are essentially
  the same, while local freshness semantics is a quotient of the other
  semantics. It is immediate from~\cite[Lemma~A.4]{SchroderEA17-full}
  % the fact that $\ub$ is injective on clean closed bar strings
  that~$N$ preserves intersection and complement of closed bar
  languages, the latter in the sense that
  $N(\free{\emptyset}\setminus L)=\names^*\setminus N(L)$ for closed
  bar languages~$L$. Both properties fail for the local freshness
  interpretation~$D$; the semantics of formulae should therefore be
  understood first in terms of bar languages, with~$D$ subsequently
  applied globally.  % ; to give a simple example of this
  % mechanism at the level of bar languages, the local freshness
  % semantics of the bar language $\{\newletter a\newletter a\}$
  % coincides with the global freshness semantics of the bar language
  % $\{\newletter a\newletter a,\newletter a a\}$.
  % Summing up, bar languages are structurally the main objects of
  % interest: They capture both global freshness and, as a quotient,
  % local freshness; and they support formalisms with nice closure
  % properties, in particular closure under complement.
  % Re N: Intersection: closed languages:
  % N(L1) cap N(L2) subset N(L1 cap L2):
  % w = ub(v) = ub(u), v in L1 ==> v = u by Lemma A.4 in bars-arXiv
  % ==> w in N(L1 cap L2).
  % N(L1 cap L2) subset N(L1) cap N(L2):
  % Monotonicity.
  %
  % empty-relative complement:
  % N (free(empty)-L) subset A^* - N(L):
  % w = ub(v), v closed, not in L: assume w in N(L),
  % i.e. w = ub(u), u in L ==> u = v by Lemma A.4,
  % contradiction
\end{remark}

\section{Syntax and Semantics of $\muBar$}
\label{sec:logic}
\noindent We proceed to introduce a variant $\muBar$ of linear
temporal logic whose formulae define bar languages. This logic
relates, via its local freshness semantics, to Freeze LTL. It replaces
freeze quantification with name binding modalities, and features
fixpoints, for increased expressiveness in comparison to the temporal
connectives of LTL~\cite{DeGiacomoVardi13}. Via global freshness
semantics, $\muBar$ may moreover be seen as a logic for session
automata~\cite{BolligEA14}.

\mysubsec{Syntax}
We fix a countably infinite set $\mathsf{V}$ of \emph{(fixpoint)
  variables}.  The set $\formulae$ of \emph{bar formulae}
$\phi,\psi,\ldots$ (in negation normal form) is generated by the grammar
\[
  \phi,\psi := \epsilon\mid \negepsilon% \mid \top \mid \bot
  \mid \phi\wedge\psi \mid \phi\vee\psi \mid\,
  \hearts_\barname \phi\mid 
  X\mid \mu X.\,\phi,
\]
where $\hearts\in\{\Diamond,\Box\}$, $\barname\in\barA$ and
$X\in\mathsf{V}$. We define $\top=\epsilon\lor\neg\epsilon$ and
$\bot=\epsilon\land\neg\epsilon$. We refer to $\Diamond_{\barname}$
and $\Box_{\barname}$ as \emph{$\barname$-modalities}.  The meaning of
the Boolean operators is standard; the fixpoint construct~$\mu$
denotes unique fixpoints, with uniqueness guaranteed by a guardedness
restriction to be made precise in a moment. The other constructs are
informally described as follows. The constant~$\epsilon$ states that
the input word is empty, and $\negepsilon$ that the input word is
nonempty. % (both $\top$ and $\negepsilon$ additionally impose that the
% free names of the remaining word come from a given context; this will
% be made precise in the formal definition of the semantics)
A formula $\Diamond_a\phi$ is read `the first letter is~$a$, and the
remaining word satisfies~$\phi$', and $\Box_a\phi$ is read dually as
`if the first letter is~$a$, then the remaining word
satisfies~$\phi$'. The reading of $\newletter a$-modalities is similar
but involves $\alpha$-renaming as detailed later in this section; as
indicated in \autoref{sec:databarl}, this means that
$\newletter a$-modalities effectively read fresh letters. They thus
replace the freeze quantifier; one important difference with the
latter is that $\Diamond_{\scriptnew a}$ consumes the letter it reads,
i.e.~advances by one step in the input word. A name~$a$ is \emph{free}
in a formula~$\phi$ if $\phi$ contains an $a$-modality at a position
that is not in the scope of any $\newletter a$-modality; that is,
$\newletter a$-modalities bind the name~$a$. We write $\FN(\phi)$ for
the set of free names in~$\phi$, and $\BN(\phi)$ for the set of
\emph{bound names} in~$\phi$, i.e.~those names~$a$ such that~$\phi$
mentions~$\newletter a$; we put $\N(\phi)=\FN(\phi)\cup\BN(\phi)$, and
(slightly generously) define the \emph{degree} of~$\phi$ to be
$\mathsf{deg}(\phi)=|N(\phi)|$.  We write~$\mathsf{FV}(\phi)$ for the
set of \emph{free} fixpoint variables in~$\phi$, defined in the
standard way by letting~$\mu X$ bind~$X$; a formula~$\phi$ is
\emph{closed} if $\mathsf{FV}(\phi)=\emptyset$. (We refrain from
introducing terminology for formulae without free \emph{names}.)
% We define a notion of free (fixpoint) variable, determined as usual
% (i.e.~an occurrence of~$X$ in a formula is \emph{bound} if it is in
% scope of fixpoint operator~$\mu X$, and \emph{free} otherwise;
% moreover,~$X$ is \emph{free in~$\phi$} if~$X$ has a free occurrence
% in~$\phi$).
As indicated above we require
that all fixpoints $\mu X.\,\phi$ are \emph{guarded}, that is, all
free occurrences of~$X$ lie within the scope of some
$\barname$-modality in~$\phi$. % The \emph{degree} $\mathsf{deg}(\phi)$
% of~$\phi$ is the least number~$k$ such that for all subformulae $\psi$
% of $\phi$, we have $|\FN(\phi)\cup \mathsf{scope}(\psi)|\leq k$, where
% $\mathsf{scope}(\psi)=\{a\in \names\mid \text{$\psi$ occurs in the
%   scope of a $\newletter a$-modality within $\phi$}\}$.
We denote by $\clos(\phi)$ the \emph{closure} of~$\phi$ in the
standard sense~\cite{Kozen83}, i.e.~the least set of formulae that
contains~$\phi$ and is closed under taking immediate subformulae and
unfolding top-level fixpoints; this set is finite. We define the
\emph{size} of~$\phi$ as $|\phi|=|\clos(\phi)|$.

For purposes of making $\formulae$ a nominal set, we regard every
fixpoint variable~$X$ with enclosing fixpoint expression
$\mu X.\,\phi$ as being annotated with the set $A=\FN(\mu X.\,\phi)$;
that is, we identify~$X$ with the pair $(X,A)$. We then
let~$G$ act by replacing names in the obvious way; i.e.\
$\pi \cdot \phi$ is obtained from $\phi$ by replacing~$a$ with
$\pi(a)$, $\newletter a$ with $\newletter\pi(a)$, and $(X,A)$ with
$(X,\pi\cdot A)$ everywhere. Otherwise, the definition is as expected:
\begin{defn}
  \emph{$\alpha$-Equivalence} $\equiv_\alpha$ on formulae is the
  congruence relation generated by
  \begin{equation*}
    \Diamond_{\scriptnew a} \phi\equiv_\alpha \Diamond_{\scriptnew
      b}\psi\text{ and }
    \Box_{\scriptnew a} \phi\equiv_\alpha \Box_{\scriptnew b}\psi\quad
    \text{whenever }
    \langle a\rangle \phi = \langle b\rangle\psi
  \end{equation*}
  (cf.\
  \autoref{def:abstraction}).
\end{defn}
% As usual, this is equivalent to the
% expected recursive definition, e.g.\
% $\Diamond_{\scriptnew a}\phi\equiv_\alpha\Diamond_{\scriptnew b}\psi$
% iff there is $\psi'$ such that
% $\langle a\rangle\phi=\langle b\rangle\psi'$ and
% $\psi'\equiv_\alpha\psi$.
\begin{remark}\label{rem:alpha}
  The point of implicitly annotating fixpoint variables with the free
  names of the enclosing $\mu$-expression is to block unsound
  $\alpha$-renamings: It ensures that, e.g.,
  $\Diamond_{\scriptnew a}(\mu
  X.\,(\Diamond_a\epsilon\lor\Diamond_{\scriptnew b}X))$ is \emph{not}
  $\alpha$-equivalent to
  $\Diamond_{\scriptnew a}(\mu
  X.\,(\Diamond_a\epsilon\lor\Diamond_{\scriptnew a}X))$ (as~$X$ is
  actually $(X,\{a\})$), and is required to ensure stability of
  $\alpha$-equivalence under fixpoint expansion, recorded next. We
  note that fixpoint expansion does \emph{not} avoid capture of names;
  e.g.~the expansion of
  $\mu X.\,(\Diamond_a\epsilon\lor\Diamond_{\scriptnew a}X)$ is
  $\Diamond_a\epsilon\lor\Diamond_{\scriptnew a}(\mu
  X.\,(\Diamond_a\epsilon\lor\Diamond_{\scriptnew a}X))$.
\end{remark}
\begin{lemma}\label{lem:alpha-unfolding}
  Let $\mu X.\,\phi\equiv_\alpha\mu X.\,\phi'$. Then
  $\phi[\mu X.\,\phi/X]\equiv_\alpha\phi'[\mu X.\,\phi'/X]$.
\end{lemma}
\begin{proof}
Immediate from the fact that by the convention that fixpoint variables
are annotated with the free names of their defining formulae,~$X$,
$\mu X.\,\phi$, and $\mu X.\,\phi'$ have the same free names and hence
allow the same $\alpha$-renamings in the outer contexts $\phi$
and~$\phi'$, respectively. 
\end{proof}
% mu X. <>_a eps \/ <>_|a (X,a)
% unfolds to
% <>_a eps \/ <>_|a mu X. <>_a eps \/ <>_|a (X,a)

\mysubsec{Semantics} We interpret each bar formula~$\phi$ as denoting
a uniformly finitely supported bar language, depending on a
\emph{context}, i.e.~a finite set $S\subseteq\names$ such that
$\FN(\phi)\subseteq S$, which specifies names that are allowed to
occur freely; at the outermost level,~$S$ will be empty
(cf.~\autoref{sec:databarl}). The context grows when we traverse
modalities $\Diamond_{\scriptnew a}$ or $\Box_{\scriptnew a}$.
 % represented as an $\alpha$-closed literal language
 % $\semS{\phi}_\sigma\subseteq\barstrings$ that depends on a
 % \emph{context} $\FN(\phi)\subseteq S\subseteq\names$, which
 % specifies names that are allowed to occur freely, and a
 % \emph{valuation}, a finitely supported function
 % $\sigma\colon \mathsf{V} \to \powfs(\barstrings)$ of fixed point
 % variables (regarded as a discrete nominal set) as uniformly
 % finitely supported bar languages, also given as $\alpha$-closed
 % literal languages.
 Correspondingly, we define satisfaction $S,w\models\phi$ of a
 formula~$\phi$ by a bar string~$w\in\free{S}$ recursively by the
 usual clauses for the Boolean connectives, and
 \allowdisplaybreaks
 \[ % SM: align* ist hierfür nicht das Richtige, und wenn man hier
    % Abstände verändern will, dann mit @{...} im
    % Header des Arrays und nicht in jeder Zeile \qquad schreiben
   \begin{array}{r@{\ }lll}
     % S,w&\models \bot  &&\qquad\text{never}\\
     % S,w&\models \top  & &\qquad \text{always}\\
     S,w&\models \negepsilon & {\Leftrightarrow} & w\neq\epsilon\\
     S,w&\models \epsilon & {\Leftrightarrow} & w=\epsilon\\
     % S,w&\models \phi\wedge\psi & {\Leftrightarrow} & S,w\models\phi\text{ and } S,w\models\psi\\
     % S,w&\models \phi\vee\psi & {\Leftrightarrow} &  S,w\models\phi\text{ or } S,w\models\psi\\
     S,w&\models \mu X.\,\phi &{\Leftrightarrow} &  S,w\models\phi[\mu X.\,\phi/X]\\
     S,w&\models \Diamond_a\phi & {\Leftrightarrow} &\exists v.\,w=av\text{ and }
       S,v\models\phi\\
     S,w&\models \Box_a\phi & {\Leftrightarrow}  & \forall v.\,\text{if }w=av\text{ then }
       S,v\models\phi\\
     S,w&\models \Diamond_{\scriptnew a}\phi & {\Leftrightarrow} &\exists \psi\in\formulae,
     v\in\barstrings, b\in\names.\,\\
     &&&
     w\equiv_\alpha\newletter bv\text{ and } \langle a\rangle\phi
     =
     \langle b \rangle \psi\text{ and } S\cup\{b\},v\models\psi
     \\
     S,w&\models \Box_{\scriptnew a}\phi & {\Leftrightarrow} &\forall \psi\in\formulae,
     v\in\barstrings, b\in\names.\,\\
     &&&
     \text{if }w\equiv_\alpha\newletter bv\text{ and } \langle a\rangle\phi
     =
     \langle b \rangle \psi\text{ then }  S\cup\{b\},v\models\psi.
   \end{array}
 \]
% \twnote{In the semantics of $\Box_{\scriptnew a}$ and $\Diamond_{\scriptnew a}$,
% instead of quantifying over $\psi$, isn't it enough to demand that $b$ is fresh
% for $\phi$ and then simply put $\psi := (a\, b)\phi$?}%
% where $\semS{\phi}^X_\sigma(A)=\semS{\phi}_{\sigma[X\mapsto A]}$,
% $(\sigma[X\mapsto A])(X)=A$ and $(\sigma[X\mapsto A])(Y)=\sigma(Y)$
% for $Y\neq X$, $A\subseteq \barstrings$. Moreover, the operator $\FP$
% takes unique fixed points; indeed, $\semS{\phi}^X_{\sigma}$ has a
% unique fixpoint due do the guardedness assumption on fixpoint
% variables, since we work over finite words. 
%That is, the semantics
% of $\mu X.\,\phi$ is just defined by recursion over the word length:
% $w\in\semS{\mu X.\,\phi}_\sigma$ iff
% $w\in\semS{\phi[\mu X.\,\phi/X]}_\sigma$, where
Guardedness of fixpoint variables guarantees that on the right hand
side of the fixpoint clause, $\mu X.\,\phi$ is evaluated only on words
that are strictly shorter than~$w$, so the given clause uniquely
defines the semantics. % ; that is, $\mu$ takes unique guarded
% fixpoints. 
Notice that $\Diamond_{\scriptnew a}$ and $\Box_{\scriptnew a}$ allow
$\alpha$-renaming of both the input word and the formula; we comment
on this point in
\autoref{rem:alpha-eq}. %For closed
% formulae $\psi$, $\semS{\phi}_\sigma$ does not depend on $\sigma$, so
% we write just
For a formula~$\phi$ such that $\FN(\phi)=\emptyset$, we briefly write
\[
  \sem{\phi}_0=\{w\in\free{\emptyset}\mid\emptyset,w\models \phi\}
  \qquad\text{and}\qquad
  \sem{\phi}=\sem{\phi}_0/{\equiv_\alpha},
\]
referring to $\sem{\phi}_0$ as the \emph{literal language} and
to~$\sem{\phi}$ as the \emph{bar language} of~$\phi$ (variants with
non-empty context and $\FN(\phi)\neq\emptyset$ are technically
unproblematic but require more notation). In particular,~$\sem{\phi}$
is closed by construction.  The \emph{global} and \emph{local
  freshness semantics} of~$\phi$ are $N(\sem{\phi})$ and
$D(\sem{\phi})$, respectively, where $N$ and~$D$ are the operations
converting bar languages into data languages described in
Section~\ref{sec:databarl}.

\begin{remark}\label{rem:ltl}
  In $\muBar$, fixpoints take on the role played by the temporal
  operators in Freeze LTL. In bar language semantics, the overall mode
  of expression in $\muBar$, illustrated in \autoref{expl:formulae},
  is slightly different from that of Freeze LTL, as in $\muBar$ the
  input is traversed using modalities tied to specific letters rather
  than using a \emph{next} operator~$\fullmoon$. %
  In local freshness semantics, the effect of~$\fullmoon$ is included
  in the name binding modality~$\Diamond_{\scriptnew a}$. For
  instance, in local freshness semantics we can express LTL-style
  formulae $\phi\ U\, \psi$ (`$\phi$ until~$\psi$') as
  $\mu X.\,\psi\lor(\phi\land\Diamond_{\newletter a}X)$. In
  particular, $\mu X.\,\epsilon\lor\Diamond_{\newletter a}X$ defines
  the universal data language, so~$\top$ is not actually needed in
  local freshness semantics. Overall, Freeze LTL and $\muBar$ (with
  local freshness semantics) intersect as indicated but are
  incomparable: On the one hand, Freeze LTL can express the language
  `the first two letters are different', which as indicated in
  \autoref{sec:databarl} is not induced by a bar language. On the
  other hand, $\muBar$ features fixpoints, which capture properties
  that generally fail to be expressible using LTL operators, e.g.~the
  language of all even-length words. The latter point relates to the
  fact that even over finite alphabets, LTL on finite words is only as
  expressive as first-order logic, equivalently star-free regular
  expressions (cf.~\cite{DeGiacomoVardi13}). Constrastingly, thanks to
  the fixpoint operators,~$\muBar$ is as expressive as its
  corresponding automata model (\autoref{lem:autformeq}).
\end{remark}
\begin{remark}\label{rem:negation}
  As indicated previously, \muBar is closed under complement: By
  taking negation normal forms, we can define $\neg\phi$ so that
  $S,w\models\neg\phi$ iff $S,w\not\models\phi$.
\end{remark}
We note next that literal languages of formulae are closed under
$\alpha$-equivalence, and that $\alpha$-equivalent renaming of
formulae indeed does not affect the semantics (cf.\
Remark~\ref{rem:alpha}):%
\begin{lemma}\label{lem:alpheq}
  For $\phi,\psi\in\formulae$, $a\in\names$, $S\subseteq \names$,
  and $w,w'\in \barstrings$, we have:%
%  \smnote{Enumerate please for substatements so that we can use
%    ref-macros.}
  \begin{enumerate}
  \item\label{lem:alpheq:1} If $S,w\models\psi$ and $w\equiv_\alpha w'$, then
    $S,w'\models \psi$. 
  \item\label{lem:alpheq:2} If $S,w\models\psi$ and
    $\phi\equiv_\alpha\phi'$, then $S,w'\models \phi'$.
  \end{enumerate}
\end{lemma}
\noindent The proof is by induction along the recursive definition of
the semantics; the case for fixpoints in Claim~\ref{lem:alpheq:2}
is by \autoref{lem:alpha-unfolding}.
\begin{remark}\label{rem:alpha-eq}
  We have noted above that the semantics allows $\alpha$-renaming of
  both words and formulae. Let us refer to an alternative semantics
  where the definition of $S,w\models\Diamond_{\scriptnew a}\phi$ is
  modified to require that there exists $w\equiv_\alpha\newletter a v$
  such that $S\cup\{a\},v\models\phi$ (without allowing
  $\alpha$-renaming of $\Diamond_{\scriptnew a}\phi$), similarly for
  $\Box_{\scriptnew a}$, as the \emph{rigid} semantics, and to the
  semantics defined above as the \emph{actual} semantics. The rigid
  semantics is not equivalent to the actual semantics, and has several
  flaws. First off, claim~\ref{lem:alpheq:2} of the above
  \autoref{lem:alpheq} fails under the rigid semantics, in which,
  for example,
  \[
    \emptyset,\newletter b\newletter a b\models\Diamond_{\scriptnew
      b}\Diamond_{\scriptnew a}\top
    \qquad\text{but}\qquad
    \emptyset,\newletter b\newletter a b\not\models\Diamond_{\scriptnew
      b}\Diamond_{\scriptnew b}\top
  \]
  (the latter because
  $\{b\},\newletter ab\not\models\Diamond_{\scriptnew b}\top$, as
  $\alpha$-renaming of~$\newletter a$ into~$\newletter b$ is blocked
  in $\newletter ab$). More importantly, the rigid semantics has
  undesirable effects in connection with fixpoints. For instance, in
  the actual semantics, the formula
  $\phi=\mu X.\,((\neg\epsilon\land\Box_{\scriptnew
    a}\bot)\lor\Diamond_{\scriptnew a}X)$ has the intuitively intended
  meaning: A bar string satisfies~$\phi$ iff it contains some plain
  name. In the rigid semantics, however, we unexpectedly have
  $\emptyset,\newletter a\newletter bab\not\models\phi$; to see this,
  note that $\{a\},\newletter bab\not\models\phi$ in the rigid
  semantics, since $\alpha$-renaming of~$\newletter b$
  into~$\newletter a$ is blocked in $\newletter bab$.
\end{remark}

\begin{expl}\label{expl:formulae}
  We consider some $\muBar$ formulae and their respective semantics
  under local and global freshness. (The local freshness versions are
  expressible in Freeze LTL in each case; recall however
  \autoref{rem:ltl}.)
  \begin{enumerate}%[wide,labelindent=0pt,topsep=3pt]
  \item The bar language $\sem{\top}$ is the set of all
    closed bar strings (modulo $\alpha$-equivalence, a qualification
    that we omit henceforth). Under both global and local freshness
    semantics, this becomes the set of all data words.
  \item The bar language $\sem{\Diamond_{\scriptnew a}\Box_a\epsilon}$
    is the language of all closed bar strings that start with a bar
    name $\newletter a$, and stop after the second letter if that
    letter exists and is the plain name~$a$
    (e.g.~$\sem{\Diamond_{\scriptnew a}\Box_a\epsilon}$ contains
    $\newletter a$, $\newletter aa$, $\newletter a\newletter b ab$ but
    not $\newletter a a a$). In both global and local
    freshness semantics, this becomes the language of all words that
    stop after the second letter if that letter exists and coincides
    with the first letter.
  \item In context~$\{a\}$, a bar string satisfies
    $\mu Y.\,((\Diamond_{\scriptnew b} Y)\lor\Diamond_a\top)$ iff it
    contains a free occurrence of~$a$ preceded only by bar names
    distinct from~$\newletter a$. Thus, the bar language of
    \begin{equation*}
      \mu X.\,(\Diamond_{\scriptnew a} (X \lor
      \mu Y.\,((\Diamond_{\scriptnew b} Y)\lor\Diamond_a\top)))
    \end{equation*}
    consists of all closed bar strings that start with a prefix of bar
    names and eventually mention a plain name corresponding to one of
    these bar names. Under both local and global freshness semantics,
    this becomes the data language of all words mentioning some letter
    twice (which is not acceptable by deterministic or even unambiguous
    register automata~\cite{BojanczykEA14,SchroderEA17-full}). Notice
    that during the evaluation of the formula, the context can become
    unboundedly large, as it grows every time a bar name is read.
  \item The bar language of the similar formula
    \begin{equation*}
      \mu X.\,((\Diamond_{\scriptnew a} X) \lor(\Diamond_{\scriptnew a}
      \mu Y.\,((\Diamond_{\scriptnew b} Y)\lor\Diamond_a\epsilon)))
    \end{equation*}
    consists of all closed bar strings where all names except the last
    one are bound names. Under global freshness semantics, this becomes
    the data language where the last letter occurs \emph{precisely}
    twice in the word, and all other names only once. Under local
    freshness semantics, the induced data language is that of all
    words where the last letter occurs \emph{at least} twice, with no
    restrictions on the other letters.
  \item To illustrate both the mechanism of local freshness via
    $\alpha$-equivalence and, once again, the use of~$\top$, we
    consider the bar language of
    \begin{equation*}
      \Diamond_{\scriptnew a}\Diamond_{\scriptnew b}\mu X.\,
      ((\Diamond_{\scriptnew b} X)\lor\Diamond_a\Diamond_b\top),
    \end{equation*}
    which consists of all closed bar strings that start
    with a bar name $\newletter a$, at some later point contain a
    substring $\newletter bab$, and have only bar names distinct
    from~$\newletter a$ in between. Under global freshness semantics,
    this becomes the data language of all words where the first
    name~$a$ occurs a second time at the third position or later, all
    letters are mutually distinct until that second occurrence, and
    the letter preceding that occurrence is repeated immediately
    after. The local freshness semantics is similar but only requires
    the letters between the first and second occurrence of~$a$ to be
    distinct from~$a$ (rather than mutually distinct), that is, the
    substring $bab$ is required to contain precisely the second
    occurrence of~$a$.
  \end{enumerate}
\end{expl}

\section{Extended Regular Nondeterministic Nominal Automata}
\label{sec:eauto}

\noindent We proceed to introduce the nominal automaton model we use
in model checking, \emph{extended regular nondeterministic nominal
  automata} (ERNNAs), a generalized version of RNNAs~\cite{SchroderEA17}
that allow for limited alternation in the form of deadlocked
universal states. % We extend two key constructions from RNNAs to
% ERNNAs: The \emph{name dropping} construction which ensures that
% literal languages (Section~\ref{sec:databarl}) are closed under
% $\alpha$-equivalence; and the finite representation in terms of a form
% of nondeterministic finite automata (NFAs) called \emph{(extended) bar
%   NFAs}.

Nominal automata models~\cite{BojanczykEA14} generally feature nominal
sets of states; these are infinite as sets but typically required to
be orbit-finite. RNNAs are distinguished from other nominal automata
models (such as \emph{nondeterministic orbit-finite
  automata}~\cite{BojanczykEA14}) in that they impose finite branching
but feature \emph{name-binding} transitions; that is, they have
\emph{free} transitions $q\trans{a}q'$ for $a\in\names$ as well as
\emph{bound} transitions $q\trans{\scriptnew a}q'$, both consuming the
respective type of letter in the input bar string~$w$. Bound
transitions may be understood as reading fresh letters. RNNAs are a
nondeterministic model, i.e.~accept~$w$ if there \emph{exists} a run
on~$w$ ending in an accepting state. ERNNAs additionally feature
\emph{$\top$-states} that accept the current word even if it has not
been read completely, and thus behave like the formula~$\top$; these
states may be seen as universal states without outgoing
transitions. % For uniformity, we correspondingly treat accepting states
% in the standard sense as existential deadlocks (which is w.l.o.g.\
% compared to unrestricted use of accepting states).\lsnote{These
%   comments do not match the current definition}
Formal definitions are as follows.

% ERNNAs have a notion of
% acceptance that, like satisfaction of $\muBar$ formulae, depends on a
% context~$S$ of allowed free names. As indicated above,  As soon
% as a run reaches such a state, the input bar string~$w$ is accepted as
% the universal condition holds vacuously, provided that the free names
% of~$w$ are in the context~$S$; that is, universal deadlocks correspond
% to the formula~$\top$, which we therefore use to label them. 

\begin{defn}
  %\sloppypar
  An \emph{extended regular nondeterministic nominal automaton (ERNNA)}
  is a four-tuple $A=(Q,\mathord{\rightarrow},s,f)$ that consists of
\begin{itemize}
\item an orbit-finite nominal set $Q$ of \emph{states} (whose orbits
  we also refer to as the \emph{orbits of~$A$});
\item an \emph{initial state} $s\in Q$ such that $\supp(s)=\emptyset$;
%  (in keeping with our restriction to closed bar languages);
\item an equivariant \emph{transition} relation
  $\mathord{\rightarrow}\subseteq Q\times(\barA\cup\{\epsilon\})\times Q$,
  with $(q,\sigma,q')\in\mathord{\to}$ denoted by $q\trans{\sigma}q'$; and
\item an equivariant \emph{acceptance} function
  $f\colon Q\to \{0,1,\top\}$%% SM: Zeilenumbruch durch s.th. gespart
\end{itemize}
such that $\rightarrow$ is \emph{$\alpha$-invariant} (that is,
$q\trans{\scriptnew a}q'$ and
$\langle a\rangle q'=\langle b\rangle q''$ imply
$q\trans{\scriptnew b}q''$) and finitely branching up to
$\alpha$-equivalence (i.e.~for each~$q$, the sets
$\{(a,q')\mid q\stackrel{a}{\rightarrow}q'\}$,
$\{(\epsilon,q')\mid q\stackrel{\epsilon}{\rightarrow}q'\}$, and
$\{\langle a\rangle q'\mid q\stackrel{\scriptnew a}{\rightarrow}q'\}$
are finite).  % If there are no $\alpha\in \barA\cup\{\epsilon\}$ and
% $q'\in Q$ such that $q\stackrel{\alpha}{\rightarrow}q'$, then we write
% $q\not\rightarrow$.
Whenever $f(q)=\top$, we require $\supp(q)=\emptyset$ and moreover
that~$q$ is a deadlock, i.e.~there are no transitions of the form
$q\trans{\sigma}q'$.  The \emph{degree} $\degree(A)$ of~$A$ is the maximal 
size of the support of a state in~$Q$ (in the translation of nominal
automata into register automata, the degree corresponds to the number of
registers~\cite{BojanczykEA14,SchroderEA17}). A state~$q$ is
\emph{accepting} if $f(q)=1$, \emph{non-accepting} if $f(q)=0$, and a
\emph{$\top$-state} if $f(q)=\top$.

We extend the transition relation to words~$w$ over~$\barA$, i.e.~to
bar strings, as usual; that is, $q\trans{w}q'$ iff there exist states
$q=q_0,q_1,\dots,q_k=q'$ and transitions
$q_i\trans{\barname_{i+1}}q_{i+1}$ for $i=0,\dots,k-1$ such that~$w$
is the concatenation $\barname_1\cdots\barname_k$, where~$\barname_i$
is regarded as a one-letter word if $\barname_i\in\barA$, and as the
empty word if $\barname_i=\epsilon$. We define
$\prelang(A)\subseteq\barstrings\times\{1,\top\}$ (for
\emph{prelanguage}) as
\begin{equation*}
  \prelang(A)=\{(w,f(q))\mid s\trans{w}q, f(q)\in\{1,\top\}\}.
\end{equation*}
The \emph{literal language}
$L_0(A)\subseteq\barstrings$ accepted by an ERNNA~$A$ then
is defined by
\begin{equation*}
  L_0(A)=\free{\emptyset}\cap\big(\{w\mid (w,1)\in\prelang(A)\}\;\cup
  \{vu\mid (v,\top)\in\prelang(A),u\in\barstrings\}\big);
\end{equation*}
that is, a closed bar string~$w$ is literally accepted if
either~$w$ has a run ending in an accepting state or a prefix of~$w$
has a run ending in a $\top$-state. (Again, extending the treatment to
bar strings with free names is technically unproblematic but heavier
on notation.)
% A \emph{run-in-context} of $A$ on a bar string
% $w=\alpha_0\alpha_1\ldots\alpha_m$ and a context $S_0\subseteq \names$
% is a sequence $\tau=(S_0,q_0),(S_1,q_1),\ldots,(S_n,q_n)$ such that
% for all $0\leq i<n$, we have
% $q_i\stackrel{\beta_i}{\rightarrow}q_{i+1}$ for some
% $\beta_i\in\barA\cup\{\epsilon\}$; we also require that $S_{i+1}=S_i$
% if $\beta_i\in\names$ and $S_{i+1}= S_{i}\cup\{a\}$ if
% $\beta_i=\newletter a$, and that
% $\beta_0\beta_1\ldots\beta_n|_{\barA}$ is a prefix of $w$, where
% $v|_{\barA}$ denotes bar string that is obtained from the word
% $v\in(\barA\cup\{\epsilon\})^*$ by dropping all occurrences of
% $\epsilon$.  In this case we refer to
% $\beta_0\beta_1\ldots\beta_n|_{\barA}$ as the \emph{run word}
% of~$\tau$.  A run-in-context $\tau$ of length $n$ on a bar string~$w$
% of length $m$ and a context~$S$ is \emph{accepting} if either the run
% word is $w$ and $f(q_n)=\epsilon$, or $f(q_n)=\top$ and
% $v\in\free{S_n}$, where $v$ is obtained by taking the postfix of $w$
% that remains after the run word has been read.\lsnote[inline]{On
%   second thoughts, does this need to be so complicated? Compare to
%   definition for extended bar NFAs} Given a context
% $\supp(s)\subseteq S\subseteq\names$, the automaton~$A$
% \emph{literally accepts} the language
% \begin{align*}
% L_0(S\vdash A)=\{w\in\barstrings\mid& \text{ there is an accepting run of $A$ on $w$}\\
% &\text{
% in context $S$}\}
% \end{align*}
% \emph{in context $S$}. 
The \emph{bar language accepted by~$A$} is the quotient
\[
  L_\alpha(A) = L_0(A)/\mathord{\equiv_\alpha}.
\]
We say that~$A$ is \emph{$\epsilon$-free} if~$A$ contains no
$\epsilon$-transitions. If~$A$ is $\epsilon$-free and contains no
$\top$-states, then $A$ is a \emph{regular nondeterministic nominal
  automaton (RNNA)}. % For RNNAs~$A$, the literal language and the bar
% language do not depend on the context~$S$, and hence are denoted
% just as~$L_0(A)$ and $L_\alpha(A)$, respectively.\smnote{This
% sentence makes no sense; no context $S$ was mentioned
% above.}\lsnote{Check whether this last part is needed.}
\end{defn}
\begin{remark}
  The presence of $\top$-states makes ERNNAs strictly more expressive
  than \mbox{RNNAs} under bar language semantics (equivalently, under
  global freshness semantics). Indeed, the ERNNA consisting of a
  single $\top$-state accepts the universal bar language, which is not
  acceptable by an RNNA~\cite{SchroderEA17}. On the other hand, under
  local freshness semantics, an accepting state with a
  $\newletter a$-self-loop accepts the universal data language (in
  analogy to the expressibility of~$\top$ by
  $\mu X.\,\epsilon\lor\Diamond_{\scriptnew a}X$ in \muBar under local
  freshness semantics, cf.~\autoref{rem:ltl}), so RNNAs are as
  expressive as ERNNAs under local freshness semantics.
\end{remark}
\mysubsec{Name dropping} Like for RNNAs, the literal language accepted
by an ERNNA is not in general closed under
$\alpha$-equivalence. However, one can adapt the notion of name
dropping~\cite{SchroderEA17} to ERNNA: Roughly speaking, an ERNNA is
\emph{name-dropping} if all its transitions may nondeterministically
lose any number of names from the support of states (which corresponds
to losing register contents in a register automaton). The literal
language of a name-dropping ERNNA is closed under
$\alpha$-equivalence, and every ERNNA~$A$ can be transformed into a
name-dropping ERNNA~$\mathsf{nd}(A)$, preserving the bar
language. This transformation is central to the inclusion checking
algorithm (see additional remarks in \autoref{sec:mc}).

\mysubsec{Representing ERNNAs}

\noindent ERNNAs are, prima facie, infinite objects; we next discuss a
finite representation of ERNNAs as \emph{extended bar NFAs},
generalizing the representation of RNNAs as bar
NFAs~\cite{SchroderEA17}. The intuition behind extended bar NFAs is
similar to that of ERNNAs, except that extended bar NFAs are not
closed under name permutation. In particular, extended bar NFAs
feature deadlocked universal states:
\begin{defn}
  An \emph{extended bar NFA} $A=(Q,\to,s,f)$ consists of
  \begin{itemize}
  \item a finite set~$Q$ of \emph{states};
  \item a \emph{transition relation}
    $\mathord{\to}\subseteq(Q\times\barA\times Q)$, with
    $(q,\barname,q')\in\mathord{\to}$ denoted by
    $q\trans{\barname}q'$;
  \item an \emph{initial} state $s\in Q$; and
  \item an \emph{acceptance} function $f\colon Q\to \{0,1,\top\}$
  \end{itemize}
  such that whenever $f(q)=\top$, then~$q$ is a deadlock, i.e.~has no
  outgoing transitions.
  We % write $q\trans{\alpha}q$ for $(q,\alpha,q')\in\:\to$, and
  extend the transition relation to words over~$\barA$ (including the
  empty word) as usual. Similarly as for ERNNAs, we define
  $\prelang(A)\subseteq\barstrings\times\{1,\top\}$ by
  \begin{equation*}
    \prelang(A)=\{(w,f(q))\mid s\trans{w}q, f(q)\in\{1,\top\}\}.
  \end{equation*}
  The \emph{literal language}
  $L_0(A)$ accepted by~$A$ is
  \begin{equation*}
    L_0(A)=  \{w\mid (w,1)\in\prelang(A)\}
    \cup\{vu\mid (v,\top)\in\prelang(A),u\in\barstrings\}.
  \end{equation*}
  The \emph{bar language} of~$A$ is then defined as the quotient
  $L_\alpha(A)=L_0(A)/\mathord{\equiv_\alpha}$. For ease of
  presentation, we only consider the case where $L_\alpha(A)$ is
  closed, which is easily checked syntactically (no~$a$ may be reached
  in~$A$ without passing~$\newletter a$). We generally write $(A,q)$
  for the extended bar NFA that arises by making $q\in Q$ the initial
  state of~$A$, dropping however the requirement that the bar language
  accepted by $(A,q)$ is
  closed. % , e.g.~in $\prelang(A,q)$ and $L_0(A,q)$
  The set $\FN(A,q)$ of \emph{free names} of $q\in Q$ is then
  $\FN(A,q)=\bigcup_{(w,b)\in\prelang(A,q)}\FN(w)$. Slightly
  sharpening the original definition~\cite{SchroderEA17}, we take the
  \emph{degree} of~$A$ to be $\degree(A) := \max_{q\in Q}|\FN(A,q)|$.
\end{defn}
\begin{theorem}\label{thm:barnfas}
  ERNNAs and extended bar NFAs accept the same bar languages; that is:
  \begin{enumerate}
  \item For a given extended bar NFA with~$n$ states and degree~$k$,
    there exists a name-dropping ERNNA of degree~$k$ with~$n\cdot 2^k$
    orbits that accepts the same bar language.
  \item For a given ERNNA with~$n$ orbits and degree~$k$, there exists
    an extended bar NFA of degree~$k$ with $n\cdot k!$ states that
    accepts the same bar language.
  \end{enumerate}
\end{theorem}
\noindent The key algorithmic task on ERNNAs is inclusion checking; we
generalize the inclusion algorithm for RNNAs~\cite{SchroderEA17} to
obtain
\begin{theorem}\label{thm:incl}
  % The inclusion problem between RNNAs and ERNNAs, represented as
  % (extended) bar NFAs, is in \ExpSpace, more precisely in
  % para-$\PSpace$: 
  Given a bar NFA~$A_1$ %with initial state~$s_1$
  and an extended bar NFA~$A_2$, %with initial state~$s_2$,
  the inclusion $L_\alpha (A_1)\subseteq L_\alpha(A_2)$
  %, where $S=\supp(s_1)\cup\supp(s_2)$, 
  can be checked using space polynomial in the number of orbits of
  $A_1$ and $A_2$, and exponential in $\degree(A_1)$ and
  $\degree(A_2)$.  The same holds under local freshness
  semantics. % , i.e.~for checking whether
  % $D(L_\alpha (A_1))\subseteq D(L_\alpha(A_2))$.
\end{theorem}
% The proof is via an $\NExpSpace$-algorithm that checks for
% \emph{non-inclusion} of~$A_1$ % . In
% % fact, % order to restrict the search space in $A_2$,
% % the algorithm checks for non-inclusion
% in the name-dropping ERNNA~$A_2'$ obtained from~$A_2$
% (\autoref{sec:eauto}), whose literal language is closed under
% $\alpha$-equivalence so that the algorithm only has to explore states
% that can be reached by literally following transitions in~$A_1$, thus
% ensuring the space bound. The algorithm for local freshness semantics
% differs only by allowing a plain name~$a$ read by~$A_1$ to be matched
% by either an $a$-transition or a $\newletter a$-transition
% in~$A'_2$.
\section{Equivalence of $\muBar$ and ERNNA}
\label{sec:muBar-ERNNA}

\noindent Our model checking algorithm will be based on translation of
closed formulae into~\mbox{ERNNAs}, in what amounts to a tableau construction
that follows a similar spirit as the standard automata-theoretic
translation of LTL, but requires a special treatment of $\Box$-formulae
and~$\negepsilon$, and moreover uses nondeterminism to bound the
number of free names in automata states (which may be thought of as
the number of registers) by guessing certain names, as explained in
the following example.
\begin{example}\label{expl:de-alt}
  Consider the formulae
  $\phi(b)=\mu Y.\,(\Box_b\bot\wedge\Box_{\scriptnew{c}} Y)$ and
  $\psi=\mu X.\,(\Box_{\scriptnew{a}} X \wedge \Box_{\scriptnew
    b}\phi(b))$. The formula~$\phi(b)$ states that the first plain
  name that occurs is not a free occurrence of~$b$, and~$\psi$ thus
  states that none of the bar names have a free occurrence later on.
  When evaluating~$\psi$ over a bar string
  $w=\newletter a_1 \newletter a_2\ldots \newletter a_n a_i v$
  consisting of~$n$ bar names~$\newletter a_i$ followed by the plain
  name $a_i$ ($1\leq i\leq n$) and a remaining bar string~$v$ (so~$w$
  does not satisfy~$\psi$), one eventually has to evaluate all the
  formulae $\phi(a_1),\dots,\phi(a_n)$ over the bar string~$a_i v$.
  Thus, the number of copies of formulae that a naively constructed
  ERRNA for~$\psi$ needs to keep track of can in principle grow
  indefinitely. At a first glance, this seems to prohibit a
  translation of formulae into orbit-finite automata.  However, we
  observe that when the letter $a_i$ is read, all $\phi(a_j)$ for
  $i\neq j$ immediately evaluate to~$\top$ (since the conjuncts
  $\Box_{a_j}\bot$ and $\Box_{\scriptnew c}\phi(a_j)$ of their
  fixpoint unfolding both hold vacuously), and only the evaluation of
  $\phi(a_i)$ becomes relevant (since the argument of $\Box_{a_i}\bot$
  is actually evaluated).  In fact, it is possible to let the ERRNA
  for $\psi$ nondeterministically guess the first plain name $a_i$
  that occurs in the input bar string. Then it suffices to let the
  ERNNA keep track of $\phi(a_i)$ since as discussed, all other copies
  of $\phi(b)$ become irrelevant. % Our
  % construction below is partly based on the observation that this
  % guessing mechanism works in general.
\end{example}
The idea from the previous example can be generalized to work
for all formulae. % Specifically: % Alternatively, one may morally view
% $A(\phi)$ as a nondeterminization of an alternating automaton whose
% states are subformulae of~$\phi$. In either view, 
% States in $A(\phi)$ are pairs~$(\Gamma,b)$, where $\Gamma$ is, roughly
% speaking, a set of $\alpha$-renamings of subformulae of~$\phi$, read
% conjunctively;~$b$ is either $0$ or one of the bound names of~$\phi$,
% and plays a central role in the treatment of~$\Box$.
To this end we introduce a recursive manipulation of formulae that uses annotations
to explicitly restrict the support of formulae and to guess and
enforce so-called distinguishing letters.  When constructing an
ERNNA from a formula, we
will use such manipulated formulae to avoid the problem described in
\autoref{expl:de-alt} by bounding the number of formulae that the
constructed ERNNA has to track.

\begin{defn}\label{def:restriction}
  Fix a marker element $*\notin\names$ (indicating absence of a
  name). Let $\phi$ be a formula, let $B$ and $C$ be sets of letters
  such that $B\subseteq C$, and let $a\in(\names
  \setminus B)\cup\{*\}$.
  For $n\in\mathbb{N}$, we define $\phi^{B}_C(a)_n$ recursively (as
  termination measure of the recursive definition we use tuples
  $(|n|,\mathsf{u}(\phi),|\phi|)$, ordered by lexicographic ordering,
  where $\mathsf{u}(\phi)$ denotes the number of unguarded fixpoint
  operators in~$\phi$) by putting $\phi^{B}_C(a)_0=\top$ and, for
  $n>0$,
  \begin{align*}
%\top^{B}_C(a)_n &= \top & \bot^{B}_C(a)_n &= \bot \\
\epsilon^{B}_C(a)_n &= \epsilon & \negepsilon^{B}_C(a)_n &= \negepsilon \\
(\psi\wedge\chi)^{B}_C(a)_n &= \psi^{B}_C(a)_n\wedge \chi^{B}_C(a)_n & 
(\psi\vee\chi)^{B}_C(a)_n &= \psi^{B}_C(a)_n\vee \chi^{B}_C(a)_n\\
(\Diamond_{\scriptnew b}\psi)^{B}_C(a)_n &= \Diamond_{\scriptnew b}
(\psi^{B\cup\{b\}}_{C\cup\{b\}}(a)_{n-1}) &
(\Box_{\scriptnew b}\psi)^{B}_C(a)_n &= \Box_{\scriptnew b}
(\psi^{B\cup\{b\}}_{C\cup\{b\}}(a)_{n-1})\\
(\mu X.\,\psi)^{B}_C(a)_n &= 
(\psi[\mu X.\,\psi/X])^{B}_C(a)_n
\end{align*}
and
\begin{align*}
(\Diamond_b\psi)^{B}_C(a)_n &= \begin{cases}
\bot & b\notin C\\
\Diamond_b (\psi^{B}_C(a)_{n-1}) & b\in C, b\in B\\
\Diamond_b (\psi^{\emptyset}_{\FN(\psi)}(*)_{n-1}) & b\in C, b\notin B 
\end{cases} \\
(\Box_b\psi)^{B}_C(a)_n &= \begin{cases}
\epsilon\vee\Diamond_{\scriptnew c}\top\vee\bigvee_{d\in B\cup\{a\}}\Diamond_d\top & b\notin C, a\neq *\\
\Box_b(\psi^\emptyset_{\FN(\psi)}(*)_{n-1}) & b\notin C, a=*\\
\Box_b (\psi^{B}_C(a)_{n-1}) & b\in C, b\in B\\
\chi(b)^{B}_C(a)_{n-1} & b\in C,b\notin B 
\end{cases}
\end{align*}
where
\begin{align*}
\chi(b)^{B}_C(a)_{n-1}=
\begin{cases}
\Box_b(\psi^\emptyset_{\FN(\psi)}(*)_{n-1}) & a=*\\
\epsilon\vee\Diamond_{\scriptnew c}\top\vee \bigvee_{d\in B}\Diamond_d\top\vee\Diamond_{a} (\psi^\emptyset_{\FN(\psi)}(*)_{n-1}) & a=b\\
\epsilon\vee\Diamond_{\scriptnew c}\top\vee\bigvee_{d\in B\cup\{a\}}\Diamond_d\top & *\neq a\neq b.
\end{cases}
\end{align*}
\end{defn}
During this process, fixpoint formulae are unfolded before replacing
any free modalities within their arguments; by guardedness of fixpoint
variables, this happens at most~$n$ times. (We intend the number~$n$
as as a strict upper bound on the length of bar strings over which
$\phi^{B}_C(a)_n$ is meant to be evaluated;
cf.~\autoref{lemma:guess-free}.)  The process replaces just the first
freely occurring boxes whose index is in $C$ but not in $B$; hence,
modal operators whose index comes from~$B$ are left unchanged.
Intuitively, $\phi^{B}_C(a)_n$ is a formula that behaves like the
restriction of $\phi$ to the support $C$ on bar strings $w$ such that
$|w|< n$ and in which $a$ is the first free name that is not contained
in $B$ (if any such name occurs in $w$; in this case we refer to the
letter $a$ as \emph{distinguishing letter}). Hence $\phi^{B}_C(a)_n$
is like the restriction of $\phi$ to support~$C$, assuming that the
distinguishing letter is~$a$. Formally:

\begin{lemma}\label{lemma:guess-free}
  Let $B$ and $C$ be sets of names such that $B\subseteq C$, let
  $a\in (\names\setminus B)\cup\{*\}$, let~$\phi$ be a formula, and let $S$
  be a context. Let $v$ be a bar string such that if $a\neq *$, then
  each letter that has a free occurrence in $v$ before the first free
  occurrence of $a$ in $v$ is contained in~$B$.  Also, suppose that if
  there is some freely occurring letter in~$v$ that is not contained
  in~$B$ and the first such letter~$d$ is in~$\mathsf{FN}(\phi)$, then
  $d\in C$.  Under these assumptions, we have
  \begin{align*}
    S,v\models\phi\iff
    S,v\models \phi^{B}_C(a)_n
    \qquad
    \text{for all $n$ such that $|v|< n$}.
  \end{align*}
\end{lemma}
\begin{proof}[Proof sketch]
  Induction along the recursive definition of $\phi^{B}_C(a)_n$.
\end{proof}
Let $\phi$ be a formula, let $B$, $C$, and $D$ be sets of names such
that $B\subseteq C$, and let $a\in(\names\setminus B)\cup D\cup\{*\}$. We
put
\begin{align*}
 \phi^{B}_C(a)_n^D=\begin{cases}\bot & \text{if }a\in D\\
							\phi^{B}_C(a)_n & \text{if }a\notin D, 
 \end{cases}
\end{align*}
the intuition being that $\phi^{B}_C(a)_n^D$ encodes $\phi$
(restricted to support $C$) together with the guess that $a$ is the
distinguishing letter (and that all names freely occurring before $a$
are from the set $B$), where guessing a letter from the set $D$ is not
allowed. This rules out the situation that a letter from $D$ is used
to satisfy a distinguishing box formula in $\phi$. Using
\autoref{lemma:guess-free}, we obtain:

\begin{lemma}\label{lemma:name-restriction}
  Let $\pi,\pi'$ be permutations, let $\phi$ be a formula, let $w$ be
  a bar string, let $S$ be a context, and put
  $A:=\mathsf{FN}(\pi\cdot\phi)$, $A':=\mathsf{FN}(\pi'\cdot\phi)$ and
  $B:=A\cap A'$.  Furthermore, let $a$ be either the first freely
  occuring letter in $w$ that is not in $B$ if that letter exists and
  is in $A\cup A'$ (in which case we say that $a$ is the
  \emph{distinguishing letter} w.r.t $B$ and $w$), or $a=*$
  otherwise. Then $a\notin B$ and we have
  \begin{multline*}
    S,w\models 
    (\pi\cdot\phi)\wedge (\pi'\cdot\phi)\iff %\Leftrightarrow
    \,
    S,w\models ((\pi\cdot\phi)^{B}_{A}(a)^{A'}_{|w|+1}\wedge(\pi'\cdot \phi)^\emptyset_B(a)^\emptyset_{|w|+1})\vee\\
     ((\pi'\cdot\phi)^{B}_{A'}(a)^A_{|w|+1}\wedge
    (\pi\cdot\phi)^\emptyset_B(a)^\emptyset_{|w|+1}).
  \end{multline*}
\end{lemma}

\noindent Relying on name restriction as in~\autoref{lemma:name-restriction},
we are able to translate formulae to ERRNAs.

\begin{theorem}\label{thm:mubar-ernna}
  For every closed formula~$\phi$ of size~$m$ and degree~$k$, there is
  an ERNNA~$A(\phi)$ with degree bounded exponentially in~$k$ and
  polynomially in~$m$, and with number of orbits bounded doubly
  exponentially in~$k$ and singly exponentially in~$m$,
  % $\min(\BN(\phi),(3|\clos(\phi)|+1)\cdot\degree(\phi))$
  that accepts the bar language of~$\phi$,
  i.e.~$L_\alpha(A(\phi))=\sem{\phi}$.
\end{theorem}
% \noindent In combination with the name-dropping
% construction~$\mathsf{nd}$ (\autoref{sec:eauto}), we thus
% obtain:
% \begin{corollary}\label{cor:formautlit}
%   For every closed formula $\phi$, the name-dropping
%   ERNNA~$\mathsf{nd}(A(\phi))$ has degree $\degree(\phi)$ and at most
%   $|\BN(\phi)|\cdot 2^{3|\clos(\phi)|+\deg(\phi)}$, and accepts the
%   same literal language as~$\phi$,
%   i.e.~$L_0(\mathsf{nd}(A(\phi)))=\sem{\phi}_0$
% \end{corollary}

\noindent We sketch the construction of~$A(\phi)$. 
Let $\clos(\phi)$ denote the closure of $\phi$, defined in the standard way.
We put
\[
  \clo=\clos(\phi)\cup\{\Diamond_\barname\psi\mid \Box_\barname\psi\in\clos(\phi)\}
  \cup \{\epsilon,\bot\}\cup \{\Diamond_b\top\mid b\in\BN(\phi)\},
\]
noting $|\clo|\leq 4m$ (recall that $\bot$ and $\top$ abbreviate $\epsilon\wedge\negepsilon$
and $\epsilon\vee\negepsilon$, respectively).
Furthermore, we put%
\begin{align*}
\mathsf{formulae}&=\{\psi^{B}_C(a)\mid \psi\in\clo, B\subseteq \N(\phi),C\subseteq\FN(\psi),
a\in C\cup\{*\}\},
\end{align*}
noting that the cardinality of $\mathsf{formulae}$ is linear in~$m$
and exponential in~$k$; specifically,
$|\mathsf{formulae}|\leq |\clo|\cdot (|\N(\phi)|+1)\cdot
2^{2|\N(\phi)|}\leq 2^{2k}\cdot 4m(k+1)\leq
2^{2k}\cdot 5m^2$ since we have $|\FN(\psi)|\leq |\N(\phi)|$ for all
$\psi\in\clo$ and since $k\leq m$.
The set $\mathsf{formulae}$ contains formulae
$\psi\in\mathsf{Cl}$ that are annotated with a single name $a$ (or $*$) and
two sets $B$ and $C$ of names; we thus refer to elements of $\mathsf{formulae}$ as \emph{annotated formulae}. The annotation relates to the transformation of formulae according to \autoref{def:restriction} but is formally just a decoration; note in particular that actually carrying out the transformation in \autoref{def:restriction} additionally requires the word length~$n$. The annotation with~$a$ is used to
encode a guessed name (with $*$ denoting the situation that no name has been
guessed yet) which we call \emph{distinguishing letter}, 
while the set $C$ encodes the restriction of the support
of $\psi$ to $C$ and the set $B$ denotes the names that are allowed to occur freely
before the distinguishing letter does. These data will be used to bound the number
of copies of subformulae that can occur in nodes of the constructed tableau.
We construct an ERRNA $A(\phi)=(Q,\to,s,f)$ with carrier set
\begin{align*}
Q=\big\{\{(\{\pi_1\cdot\phi_1,\ldots,\pi_n\cdot\phi_n\},
a)
\mid &\,\pi_1,\ldots,\pi_{n}\in G,\\
&
\{\phi_1,\ldots,\phi_n\}\subseteq\mathsf{formulae}, a\in\names\cup\{*\}\big\}, 
\end{align*}
where the $\pi_i$ act on names as usual and we define
$\pi_{i}(*)=*$. The bound on the number of orbits follows since
each combination of a subset $\Phi$ of $\mathsf{formulae}$
and an equivalence relation on the set of free names
of $\Phi$ in union with $\{a\}$
gives rise to at most one orbit. We put
$s=(\{\phi^\emptyset_{\FN(\phi)}(*)\},*)$.  Given a state
$(\Gamma,b)\in Q$, formulae $\pi\cdot (\psi_C^B(a))\in \Gamma$ stand
for instances of $\psi$ in which the distinguishing letter is fixed to
be~$\pi\cdot a$, $\pi\cdot B$ is the set of free names that are
allowed to occur before $\pi\cdot a$ does, and the support of $\psi$
is restricted to $\pi\cdot C$.  The shape of a formula
$\pi\cdot (\psi_C^B(a))$ is just the shape of $\psi$; for brevity, we
omit the annotations with $C$, $B$, $a$ when they are not relevant.
To deal with box operators and $\negepsilon$, states also contain a
separate name component $b$ (which may be $*$) that denotes a guess of
the last relevant free name after which the evaluation of formulae
from $\Gamma$ will stop. A state $(\Gamma,b)\in Q$ is
\emph{propositional} if $\Gamma$ contains some formula of the shape
$\psi_1\vee\psi_2$, $\psi_1\wedge\psi_2$, or $\mu X.\,\psi_1$;
\emph{quasimodal} if $(\Gamma,b)$ is not propositional but~$\Gamma$
contains some formula of the form $\neg\epsilon$ or~$\Box_\barname\phi$;
and \emph{modal} if $(\Gamma,b)$ is neither propositional nor
quasimodal, i.e.~if all elements of~$\Gamma$ are either
of the shape  $\epsilon$ or $\Diamond_\barname\phi$. We define the acceptance function
$f\colon Q \to \{0,1,\top\}$ by $f(\emptyset,b)=\top$,
$f(\Gamma,b)=1$ if $\Gamma\neq \emptyset$ and all elements
of $\Gamma$ are of the shape $\epsilon$, and $f(\Gamma,b)=0$ for
all other states.

Transitions work roughly as follows. The component~$\Gamma$ of a state
$(\Gamma,b)$ plays the usual role: It records formulae that the
automaton requires the remaining input word to satisfy. To bound the
number of free names that have to be tracked, formulae are annotated
with guesses for distinguishing letters. The
transitions from propositional and modal states follow the
standard tableau rules; e.g.~given a propositional state
$q=(\Gamma\cup\{\psi\land\chi\},b)$, we have
$q\trans{\epsilon}(\Gamma\cup\{\psi,\chi\},b)$; given a modal state
$q=(\{\Diamond_a\psi_1,\dots,\Diamond_a\psi_n\},b)$, we have
$q\trans{a}(\{\psi_1,\dots,\psi_n\},b)$; and a modal state containing
$\Diamond_a\psi$ and $\Diamond_b\chi$ for $a\neq b$ is a deadlock (the
treatment of $\Diamond_{\scriptnew a}$ is more
involved). $\Box$-formulae and $\neg\epsilon$ are dealt with by
$\epsilon$-transitions from quasimodal states $(\Gamma,b)$ to modal
states. This process is straightforward if~$\Gamma$ contains some
formula $\Diamond_\barname\psi$. The critical case is the remaining one:
A formula $\Box_\barname\psi$ in~$\Gamma$ can be satisfied by in fact
satisfying~$\Diamond_\barname\psi$, by ending the word, or by reading either
a plain name~$c$ other than~$\barname$ or a bar name (if $\barname\in\names$),
or by reading any plain name (if $\barname\in\overline{\names}$).
This is where the second component~$b$ of states comes
in: The letter~$c$ must have previously appeared as a bar name;
$A(\phi)$ guesses when this happens (in bound transitions from modal
states), and records its guess in~$b$, with~$*$ representing the
situation that no guess has yet been made.

When constructing transitions in $A(\phi)$ as explained above, the set
of formulae to which a naively constructed transition leads may
contain two instances $\pi\cdot(\psi_C^B(a))$ and
$\pi'\cdot(\psi_C^B(a))$ of an annotated formula~$\psi$. In this
situation, the rest of the word has to satisfy both formulae, but the
set $Q$ can only contain one instance of $\psi_C^B(a)$.  Here we use
the above restriction technique and repeated application
of~\autoref{lemma:name-restriction} to ensure that only a single copy
needs to be kept.%

We also have a converse translation which goes via the equivalence of ERNNAs
and extended bar NFAs, and associates a fixpoint variable to each
state.
\begin{theorem}\label{lem:autformeq}
  For every ERNNA~$A$ there exists a formula~$\phi$ such that
  $\sem{\phi}=L_\alpha(A)$.
\end{theorem}
\noindent Notably, by \autoref{rem:negation}, the results of this
section imply
\begin{corollary}
  The class of closed bar languages definable by ERNNAs is closed
  under complement: for each ERNNA~$A$ there exists an ERNNA~$\bar A$
  such that $L_\alpha(\bar A)=\free{\emptyset}\setminus L_\alpha(A)$.
\end{corollary}
\noindent As mentioned previously, RNNAs are essentially equivalent to
session automata~\cite{SchroderEA17}, which are only closed under
\emph{resource-bounded} complement~\cite{BolligEA14} for some resource
bound~$k$; in our present terminology, this corresponds to complement
within the set of closed bar strings, modulo $\alpha$-equivalence,
that can be written with at most~$k$ different bound names. It is
maybe surprising that full complementation (of session automata or
RNNAs) is enabled by simply allowing $\top$-states.
\section{Reasoning for \muBar}
\label{sec:mc}

Using \autoref{thm:mubar-ernna}, we reduce reasoning problems for
$\muBar$ to ERNNAs:
\begin{defn}[Reasoning problems]
  A closed formula~$\phi$ is \emph{satisfiable} if
  $\sem{\phi}\neq\emptyset$, and \emph{valid} if
  $\sem{\phi}=\free{\emptyset}$. A formula~$\psi$
  \emph{refines}~$\phi$ if $\sem{\psi}\subseteq\sem{\phi}$. An
  RNNA~$A$ \emph{satisfies~$\phi$} (written $A\models\phi$) if
  $L_\alpha(A)\subseteq\sem{\phi}$. The \emph{model checking problem}
  is to decide whether $A\models\phi$. We sometimes emphasize that
  these problems refer to \emph{bar languages} or, equivalently,
  \emph{global freshness}. The corresponding problems for \emph{local
    freshness} arise by applying the~$D$ operator
  (\autoref{sec:databarl}) to all bar languages involved; e.g.~$A$
  % where $S=\supp(s)\cup\FN(\phi)$ and for given~$A,\phi$. Similarly,
  % \emph{model checking under local
  \emph{satisfies~$\phi$ under local freshness} if
  $D(L_\alpha(A))\subseteq D(\sem{\phi})$.
\end{defn}
%\noindent (We know that the RNNA~$A$ accepts only words whose free
%names are contained in $\supp(s)$~\cite{SchroderEA17}, so it suffices
%to consider the case where the context of~$\phi$
%is~$\supp(s)\cup\FN(\phi)$.) 
 The complexity of \muBar reasoning problems, obtained using
in particular \autoref{thm:mubar-ernna} and \autoref{thm:incl}, is
summed up as follows. 

\begin{theorem}\label{cor:mc}
  \begin{enumerate}
  \item Model checking $\muBar$ over RNNAs is in $\textsc{2ExpSpace}$,
    more precisely in para-$\textsc{ExpSpace}$ with the degree of the
    formula as the parameter, under both bar language semantics
    (equivalently global freshness) and under local freshness. The
    same holds for validity under local
    freshness.  % Yet more precisely, $A\models \phi$ can be checked in
    % space polynomial in the number of orbits of $A$, and doubly
    % exponential in $|\N(\phi)|$ and $\degree(A)$.\lsnote{@Daniel:
    % More
    % precision here?}
  \item The satisfiability problem for $\muBar$ is in
    $\textsc{ExpSpace}$, more precisely in para-$\textsc{PSpace}$ with
    the degree of the formula as the parameter, under both bar
    language semantics / global freshness and under local
    freshness. The same holds for validity and refinement under bar
    language semantics / global freshness.
  \end{enumerate}
\end{theorem}
We leave the complexity (and in fact the decidability) of refinement
checking under local freshness semantics as an open problem.

% Since $A(\phi)$ is, by the above, of degree
% $\degree(\phi)$ and has at most $(|\BN(\phi)|+1)\cdot 2^{3|\clos(\phi)|}$
% orbits, we obtain a model checking algorithm for $\muBar$ witnessing
% an upper bound $\ExpSpace$, or again more precisely para-$\PSpace$:

% Li closed, N(L1) subset N(L2) =?=> L1 subset L2:
% Let w in L1 => ub(w) in N(L2) => w = ub(v), v in L2 => w=v.

\section{Conclusions}
\noindent We have defined a specification logic~$\muBar$ for finite
words over infinite alphabets, modelled in the framework of nominal
sets, which covers both local and global freshness.  $\muBar$ features
freeze quantification in the shape of name-binding modalities, and as
such relates to Freeze LTL. It combines comparatively low complexity
of the main reasoning problems with reasonable expressiveness, in
particular unboundedly many registers, full nondeterminism, and
closure under complement.  Freshness is based on $\alpha$-equivalence
in nominal words; this entails certain expressive limitations on
local freshness, which however seem acceptable in relation to the
mentioned good computational properties. 

An important issue for future work is the behaviour of $\muBar$ over
infinite words; also, we will investigate whether our methods extend
to languages of data trees~(e.g.~\cite{FigueiraSegoufin09}).

%\clearpage
% \bibliographystyle{my}
\bibliographystyle{plainurl}% the mandatory bibstyle for LIPIcs
\bibliography{coalgml}

\clearpage
\appendix

\noindent{\bfseries\sffamily\LARGE Appendix}
\par
\bigskip

%\section{Appendix}

\section{Details for \autoref{sec:prelim}}
Given a nominal set~$X$, we generally write
\begin{equation*}
  \fix x = \{ \pi \in G \mid \pi x = x\} 
  \quad
  \text{and}
  \quad
  \textstyle\Fix A = \bigcap_{x \in A} \fix x
\end{equation*}
for $A\subseteq X$, $x\in X$ (so~$A$ is a support of~$x$ iff
$\Fix(A)\subseteq\fix(x)$). On functions $f\colon X \to Y$ between
nominal sets by $(\pi \star f) (x) = \pi\cdot (f(\pi^{-1}\cdot x))$
(we denote the action on functions by~$\star$ to avoid misreadings).
Thus, $f\colon X\to Y$ is equivariant if
$f(\pi\cdot x) = \pi\cdot f(x)$ for all $x \in X$ and $\pi \in G$;
similarly~$f$ has support~$S$ iff this condition holds
whenever~$\pi\in\Fix(S)$. For equivariant~$f$, we have
$\supp(f(x)) \subseteq \supp(x)$, that is, equivariant functions never
introduce new names.  The function $\supp$ itself is equivariant,
i.e.~$\supp(\pi\cdot x)=\pi\cdot (\supp(x))$ for $\pi\in G$. Hence
$|\supp(x_1)| = |\supp(x_2)|$ whenever $x_1$, $x_2$ are in the same
orbit of a nominal set.

\section{Details for \autoref{sec:databarl}}

\subsection*{Details for \autoref{rem:local-global}}

\noindent \emph{Preservation of intersection by~$N$:} Let $L,K$ be
closed bar languages. Since~$N$ is clearly monotone, we have
$N(L\cap K)\subseteq N(L)\cap N(K)$. Conversely, let
$v\in N(L)\cap N(K)$. Then we have clean $w\in L$ and $u\in L$ such
that $v=\ub(w)=\ub(u)$; since~$w,u$ are moreover closed, injectivity
of~$\ub$ on clean closed bar
strings~\cite[Lemma~A.4]{SchroderEA17-full} implies that $w=u$, so
$v\in N(L\cap K)$.

\emph{Preservation of complement by~$N$:} Recall for this claim and
the next that we restrict to closed bar strings; in particular,
complements are understood within the set of all closed bar
strings. The argument is then analogous to the one for intersection
once we show that $N(\free{\emptyset})=\names^*$. To this end, let
$w\in\names^*$. Take $\bar w$ to be the bar string obtained by
replacing the first occurrence of every letter~$a$ in~$w$ with
$\newletter a$. Then $\bar w\in\free{\emptyset}$, and $w=\ub(\bar w)$.

\emph{Failure of preservation of intersection and complement for~$D$}:
The bar languages $\newletter aa$ and $\newletter a\newletter a$ have
empty intersection, but
$D(\newletter aa)\cap D(\newletter a\newletter a)=\{aa\mid
a\in\names\}$. Similarly, the complement of the bar language
$\newletter a\newletter b$ contains $\newletter aa$, while
$D(\newletter a\newletter b)$ contains all words of length~$2$, so the
complement of $D(\newletter a\newletter b)$ does not contain any words
of length~$2$.

\section{Details for \autoref{sec:logic}}

\subsection*{Details for \autoref{rem:negation}}

We define $\neg\phi$ recursively by 
\begin{align*}
%\neg\top&=\bot & \neg\bot &= \top\\
\neg(\epsilon) &= \negepsilon & \neg(\negepsilon)&=\epsilon\\
\neg(\psi_1\land\psi_2) &= \neg\psi_1\lor\neg\psi_2 & 
\neg(\psi_1\lor\psi_2) &= \neg\psi_1\land\neg\psi_2\\
\neg\Diamond_\barname\psi &= \Box_\barname\neg\psi & 
\neg\Box_\barname\psi &= \Diamond_\barname\neg\psi\\
\neg \mu X.\,\psi &= \mu X.\, \neg\psi[\neg X/X],
\end{align*}
where we understand
double negation $\neg\neg X$ as implicitly eliminated.
Since~$\mu$ takes unique fixpoints, it is self-dual.  For instance,
$\neg\mu X.\epsilon\lor\Diamond_a X=\mu X.\neg \epsilon\land \Box_a X$.
It remains to show that for all contexts $S\subseteq\names$, all formulae 
$\phi\in\formulae$
and all bar strings $w\in\barstrings$, we have
$S,w\models \neg\phi$ if and only if $S,w\not\models\phi$.
The proof of this is by induction over the tuple $(|w|,\mathsf{u}(\phi),|\phi|)$, in 
lexicographic order, where $\mathsf{u}(\phi)$ denotes the number
of unguarded fixpoint operators in $\phi$. We show the proof just for
selected illustrative cases.

E.g.~for $\phi=\psi_1\lor\psi_2$, we have
\begin{align*}
S,w\models\neg(\psi_1\lor\psi_2) & \Leftrightarrow
S,w\models\neg\psi_1\land\neg\psi_2\\
&\Leftrightarrow  S,w\models\neg\psi_1\text{ and }S,w\models\neg\psi_2\\
&\Leftrightarrow 
S,w\not\models\psi_1\text{ and }S,w\not\models\psi_2\\
&\Leftrightarrow  S,w\not\models \psi_1\lor\psi_2,
\end{align*}
where the third equivalence follows from the inductive hypothesis
since $|\psi_i|<|\phi|$ for $i\in\{1,2\}$.

For $\phi=\Diamond_a\psi$, we have by definition that 
$S,w\models\neg\Diamond_a\psi$ if and only if
$S,w\models\Box_a\neg\psi$. If~$w$ has the form either
$w=bv$ for $b\in\names$ such that $b\neq a$ or $w=|cv$ for $c\in\names$,
 then  $S,w\models\Box_a\neg\psi$ and
$S,w\not\models\Diamond_a\psi$ by the semantics of modal operators. 
If~$w$ has the form $w=av$,
then  $S,w\models\Box_a\neg\psi$ if and only if
$S,v\models\neg\psi$, which by the inductive hypothesis (which applies because $|v|<|w|$)  holds
if and only if $S,v\not\models\psi$, which in turn is the case
if and only if $S,w\not\models\Diamond_a\psi$.

For $\phi=\Box_{\scriptnew a}\psi$, we have 
\begin{align*}
S,w\models\neg\Box_{\scriptnew a}\psi & \Leftrightarrow
S,w\models\Diamond_{\scriptnew a}\neg\psi\\
&\Leftrightarrow\exists v\in\barstrings, c\in\names, \chi\in\formulae.\,
w\equiv_\alpha \newletter cv,
\langle a\rangle \neg\psi = \langle c\rangle \neg\chi,
S\cup\{c\},v\models\neg\chi\\
&\Leftrightarrow\exists v\in\barstrings, c\in\names, \chi\in\formulae.\,
w\equiv_\alpha \newletter cv,
\langle a\rangle \neg\psi = \langle c\rangle \neg\chi,
S\cup\{c\},v\not\models\chi\\
&\Leftrightarrow\neg(\forall v\in\barstrings, c\in\names, \chi\in\formulae.\,
w\equiv_\alpha \newletter cv\land
\langle a\rangle \neg\psi = \langle c\rangle \neg\chi\Rightarrow
S\cup\{c\},v\models\chi)\\
&\Leftrightarrow S,w\not\models\Box_{\scriptnew a}\psi,
\end{align*}
where the third equivalence follows from the inductive hypothesis since
$|v|<|w|$.

For $\phi=\mu X.\,\psi$, we have 
\begin{align*}
S,w\models\neg\mu X.\,\psi & \Leftrightarrow  S,w\models\mu X.\,\neg\psi[\neg X/X]\\
&\Leftrightarrow  S,w\models(\neg\psi[\neg X/X])[\mu X.\,\neg\psi[\neg X/X]/X]\\
&\Leftrightarrow  S,w\models\neg(\psi[\mu X.\,\psi/X])\\
&\Leftrightarrow  S,w\not\models \psi[\mu X.\,\psi/X]\\
&\Leftrightarrow  S,w\not\models \mu X.\,\psi,
\end{align*} 
where the fourth equivalence holds by the inductive
hypothesis since the fixpoint variable~$X$ is guarded by modal operators in $\psi$
so that $\mu X.$ is unguarded in $\mu X.\,\psi$
but guarded in $\psi[\mu X.\,\psi/X]$, hence 
$\mathsf{u}(\psi[\mu X.\,\psi/X])<\mathsf{u}(\mu X.\,\psi)$.

\subsection*{Proof of \autoref{lem:alpheq}}

\begin{proof}%[\autoref{lem:alpheq}]
  In both claims, the proof is by induction along the recursive
  definition $S,w\models\phi$, exploiting that this
  definition terminates as explained just after the definition of the
  semantics.
  \begin{enumerate}
  \item 
%    \begin{enumerate}[label=(1\alph*)]
    For % $\top$, $\bot$,
    $\epsilon$ and $\negepsilon$ the statement
    holds trivially. % , and for variables $X$ by our
      % assumption that each $\sigma(X)$ is $\alpha$-closed (see the
      % definition of the semantics on page~\pageref{sec:sem}).
    For conjunction, disjunction, $a$-modalities, and fixpoints, one
    simply uses the induction hypothesis. For
    $\newletter a$-modalities, the claim is immediate from the
    definition of the semantics.

  \item Again, the claim is trivial for %$\top$, $\bot$, and
    $\epsilon$ and $\neg\epsilon$, and the inductive steps for
    $\land$ and~$\lor$ are straightforward. The case for fixpoints is
    immediate from the inductive hypothesis and
    \autoref{lem:alpha-unfolding}.

    For $\phi=\Diamond_a\chi$, we have $S,w\models\Diamond_a\chi$ so
    that there is $w'$ such that $w=aw'$ and $S,w'\models\chi$. As
    $\phi\equiv_\alpha \phi'$, we have $\phi'=\Diamond_a\chi'$ for
    some $\chi'$ such that $\chi\equiv_\alpha\chi'$.  The induction
    hypothesis finishes the case.

    For $\phi=\Diamond_{\scriptnew a}\chi$, we have
    $S,w\models\Diamond_{\scriptnew a}\chi$ so that there are
    $b\in\names$, $\psi\in\formulae$, and $v\in\barstrings$ such that
    $w=\newletter bv$,
    $\Diamond_{\scriptnew a}\chi \equiv_\alpha \Diamond_{\scriptnew
      b}\psi$ and $S\cup\{b\},v\models\psi$.  Since
    $\phi\equiv_\alpha\phi'$, there are $c\in\names$,
    $\chi'\in\formulae$ such that $\phi'=\Diamond_{\scriptnew c}\chi'$
    and $\chi= (bc)\cdot\chi'$.  Hence, it suffices to show that there
    are $\psi'\in\formulae$, $d\in\names$ and $v'\in\barstrings$ such
    that $w=\newletter dv'$,
    $\Diamond_{\scriptnew c}\chi' \equiv_\alpha \Diamond_{\scriptnew
      d}\psi'$ and $S\cup\{d\},v\models\psi'$.  Pick $\psi'=\psi$,
    $v'=v$ and $d=b$ so that $w=\newletter bv=\newletter dv'$ and
    $S\cup\{d\},v'\models\psi'$. We also have
    $\Diamond_{\scriptnew c}\chi'\equiv_\alpha \Diamond_{\scriptnew
      a}\chi \equiv_\alpha\Diamond_{\scriptnew
      b}\psi=\Diamond_{\scriptnew d}\psi'$, as required. \qedhere
  \end{enumerate}
\end{proof}
%
% \begin{remark}\label{rem:fixbar}
%   Note that it does not matter whether we compute fixed points of
%   functions $\sem{\phi}_\sigma^X$ in $\powfs(\barstrings)$ or its
%   finitely supported complete subposet of all bar languages
%   (represented by $\alpha$-closed subsets of $\barstrings$). Indeed,
%   $\FP \sem{\phi}_\sigma^X$ is the same in both cases because every
%   prefixed point $A \in \powfs(\barstrings)$ contains an
%   $\alpha$-closed one. To see this, let $A_0$ be the largest
%   $\alpha$-closed subset of $A$. This exists since $\emptyset$ is
%   $\alpha$-closed and $\alpha$-closed subsets of $A$ are closed under
%   union. Then $A_0$ is a prefixed point of $\sem{\phi}_\sigma^X$:
%   \[
%     \sem{\phi}_\sigma^X(A_0)
%     \subseteq
%     \sem{\phi}_{\sigma}^X(A)
%     \subseteq A,
%   \]
%   where the first inclusion uses monotonicity of the semantics in
%   $\sigma$ (\autoref{lem:sem}\ref{lem:sem:2}). By
%   \autoref{lem:alpheq}\ref{lem:alpheq:1},
%   $\sem{\phi}_\sigma^X(A_0) = \sem{\phi}_{\sigma[X \mapsto A_0]}$ is
%   $\alpha$-closed. Thus, this set is contained in $A_0$ as desired.
  
%   For $\GFP$ one argues dually. 
% \end{remark}
% 

\takeout{%% SM: creates dangling refs
\subsection*{Proof of \autoref{lem:boxfree}}

\begin{proof}
We have $S,w\models\negepsilon$ if and only if $w$ is
a non-empty bar string
whose free names are in $S$. This is the
case if and only if $w$ is non-empty and
the first
letter in $w$ is either a plain name $b$ 
such that $b\in S$ or a bar name $\newletter c$ for some $c\in\names$.
This in turn is the case if and only if
 $S,w\models\bigvee\nolimits_{b\in S}\Diamond_b\top\vee
 \Diamond_{\scriptnew c}\top$.

We have $a\in S$ so that $S,w\models \Box_a\psi$ if and only if
$w$ is the empty bar string $\epsilon$, or a non-empty bar
string that starts with some plain name $b\in S$ such that $a\neq b$ 
or with some bar name $\newletter c$, or 
we have $w=aw'$ and $S,w'\models\psi$.
This in turn is the case if and only if
$S,w\models\epsilon$ or
$S,w\models\bigvee\nolimits_{b\in S\setminus \{a\}}\Diamond_b\top$ or
$S,w\models\Diamond_{\scriptnew c}\top$ or $S,w\models
\Diamond_a\psi$.

Lastly, we have $S,w\models\Box_{\scriptnew a}\psi$ if
and only if $w=\epsilon$ or $w$ starts with some plain name $b\in S$,
or we have $w=|cw'$ for some $c\in\names$ and
 $S\cup\{c\},w'\models\chi$
for some $\chi\in\formulae$ such that
$\langle a\rangle\psi=\langle b\rangle\chi$.
This in turn is the case if and only if
$S,w\models \epsilon$ or
$S,w\models \bigvee\nolimits_{b\in S}\Diamond_b\top$
or $S,w\models \Diamond_{\scriptnew c}\psi$.

\end{proof}}% end takeout

\section{Details for \autoref{sec:eauto}}

\begin{lemma}\label{lem:altauto}
Let $A$ be an ERNNA with set $Q$ of states, let $q,q'\in Q$, and let $a\in\names$. Then the following hold.
\begin{enumerate}
\item If $q\stackrel{\epsilon}{\rightarrow} q'$, then $\mathsf{supp}(q')\subseteq\mathsf{supp}(q)$.
\item If $q\stackrel{a}{\rightarrow} q'$, then $\mathsf{supp}(q')\cup\{a\}\subseteq\mathsf{supp}(q)$.
\item If $q\stackrel{\scriptnew a}{\rightarrow} q'$, then $\mathsf{supp}(q')\subseteq\mathsf{supp}(q)\cup \{a\}$.
\end{enumerate}
\end{lemma}

\begin{proof} 
\begin{enumerate}
\item Let $Z=\{(\epsilon,q')\mid q\stackrel{\epsilon}{\rightarrow} q'\}$. Then we have 
\begin{align*}
\mathsf{supp}(q')=\mathsf{supp}(\epsilon,q')\subseteq\mathsf{supp}(Z)\subseteq \mathsf{supp}(q),
\end{align*}
where the first inclusion holds since $Z$ is ufs
and the second inclusion holds by equivariance of $\rightarrow$.
\item Let $Z=\{(a,q')\mid q\stackrel{a}{\rightarrow} q'\}$. Then we have 
\begin{align*}
\mathsf{supp}(q')\cup\{a\}=\mathsf{supp}(a,q')\subseteq\mathsf{supp}(Z)\subseteq \mathsf{supp}(q),
\end{align*}
where the first inclusion holds since $Z$ is ufs and the second inclusion holds by equivariance of $\rightarrow$.
\item
Let $Z=\{\langle a\rangle q' \mid q\stackrel{\scriptnew a}{\rightarrow} q'\}$. Then we have 
\begin{align*}
\mathsf{supp}(q')\subseteq \mathsf{supp}(\langle a\rangle q')\cup\{a\}\subseteq\mathsf{supp}(Z)\cup\{a\}\subseteq \mathsf{supp}(q)\cup\{a\},
\end{align*}
where the first inclusion holds since $\mathsf{supp}(\langle a\rangle q')=\mathsf{supp}(q')\setminus\{a\}$, the second inclusion holds since
$Z$ is ufs and the third inclusion holds by equivariance of the relation
$\rightarrow$.\qedhere
\end{enumerate}
\end{proof}

\begin{remark}
  It is easy to see that $\epsilon$-transitions can be eliminated from
  an ERNNA in the usual fashion; the key point to note is that by the
  above lemma, the support of states cannot increase along
  $\epsilon$-transitions, so that $\epsilon$-elimination preserves the
  conditions on finite branching.
\end{remark}

\subsection{Name Dropping}
\begin{expl}
  The literal languages of ERNNAs are in general not closed under
  $\alpha$-equivalence, as can be seen on, e.g., the ERNNA (similar to
  the corresponding example RNNA in~\cite{SchroderEA17}) denoted by
\begin{align*}
s\stackrel{|a}{\rightarrow}t(a)\stackrel{|b}{\rightarrow}u(a,b)\stackrel{b}{\rightarrow} v,
\end{align*}
where~$v$ is an $\epsilon$-state; by $t(a)$, $u(a,b)$ we mean to
denote orbits of the state space arising by renaming the indicated
names~$a,b$ in accordance with the equivariance and
$\alpha$-invariance conditions in the definition of ERNNA. The
automaton accepts the bar string $|a |b b$ since we have transitions
$s\stackrel{|a}{\rightarrow}t(a)$,
$t(a)\stackrel{|b}{\rightarrow}u(a,b)$ and
$u(a,b)\stackrel{b}{\rightarrow} v$. On the other hand, the automaton
does not accept the $\alpha$-equivalent bar string $|b |b b$ since
after the transition $s\stackrel{|b}{\rightarrow}t(b)$, we do have a
transition $t(b)\stackrel{|c}{\rightarrow}u(b,c)$, in which
however~$c$ cannot be renamed into~$b$ as~$b$ occurs freely in the
target state.
\end{expl}
\noindent We adapt the name-dropping construction for
RNNA~\cite{SchroderEA17}, which is aimed at regaining closure of the
literal language under $\alpha$-equivalence, to ERNNA. E.g.~in terms
of the above example, the name dropping construction will just offer
the option of forgetting the name~$b$ in $u(b,c)$, thus unblocking the
desired $\alpha$-renaming.
\begin{defn}
Given an $\epsilon$-free ERNNA $A=(Q,\rightarrow,s,f)$ of degree $k$ with $n$ orbits, we define
the \emph{name-dropping} ERNNA $\mathsf{nd(A)}=(Q',\rightarrow',s',f')$ 
by 
\begin{align*}
Q'=\{q|_N\mid q\in Q, N\subseteq\supp(q)\},
\end{align*}
where $q|_N:=\mathsf{Fix}(N)q$;
here, $\mathsf{Fix}(N)q$ denotes the orbit of $q$ under $\mathsf{Fix}(N)$.
%\dhnote[inline]{FoSSaCS-Reviewer 1 thinks it is established to
%call $\mathsf{Fix}(N)q$ the $N$-orbit of $q$.}
We put $s'=s|_{\supp(s)}$ and $f(q|_N)=f(q)$. As transitions,
we take 
\begin{itemize}
\item $q|_N\stackrel{a}{\rightarrow'} q'|_{N'}$, whenever 
$q\stackrel{a}{\rightarrow} q'$, $N'\subseteq N$, and $a\in N$, and
\item $q|_N\stackrel{\scriptnew a}{\rightarrow'} q'|_{N'}$, whenever 
$q\stackrel{\scriptnew b}{\rightarrow} q''$, $N''\subseteq \supp(q'')\cap(N\cup\{b\})$, and $\langle a\rangle (q'|_{N'})=\langle b\rangle (q''|_{N''})$.
\end{itemize}
\end{defn}
As in~\cite{SchroderEA17},
the name-dropping construction does not change the bar languages of
automata, but it closes the accepted
literal language under $\alpha$-equivalence.

\begin{lemma}\label{lem:namedrop}
Given an $\epsilon$-free ERNNA $A$ of degree $k$ and with $n$ orbits, 
a bar string $w\in\barstrings$,
we have 
\begin{enumerate}
\item $L_\alpha(A)=L_\alpha(\mathsf{nd}(A))$,
\item $L_0(\mathsf{nd}(A))$ is closed under $\alpha$-equivalence.
\end{enumerate}
Furthermore, $\mathsf{nd}(A)$ is of degree $k$ and has at most 
$n\cdot 2^k$ orbits.
\end{lemma}

\begin{proof}
  The proof works just as the proof of the corresponding results in
  previous work~\cite[Lem.~5.7 and~5.8]{SchroderEA17}, with the only difference being that, in
  contrast to RNNA, ERNNA may contain states $q$ such that
  $f(q)=\top$. By definition of ERNNA, such states have empty support
  and no outgoing transitions so that the name-dropping construction
  does not change anything for these states.
\end{proof}

\subsection*{Proof of \autoref{thm:barnfas}}

\begin{proof}
  An extended bar NFA~$A$ generates an ERNNA~$\bar A$ as
  follows. Roughly speaking, we annotate each state~$q$ of~$A$ with
  the free names of $\prelang(A,q)$ in the evident generalized sense
  (i.e.~no longer insisting that bar strings in $\prelang(A,q)$ are
  closed), and then close under renaming. Formally, we take the state
  set~$\bar Q$ of~$\bar A$ to consist of pairs $(q,\pi \Fix(N_q))$
  where~$q\in Q$, $\pi\in G$, $N_q=\FN(A,q)$, and $\pi\Fix(N_q)$
  denotes a left coset (i.e.\
  $\pi\Fix(N_q)=\{\pi\rho\mid\rho\in\Fix(N_{q})\}$).  The acceptance
  map~$\bar f$ of~$\bar A$ is given by $\bar f(q,\pi\Fix(N_q))=f(q)$,
  and the initial state of~$\bar A$ is~$(s,\id\Fix(N_s))$. The set of
  left cosets of $\Fix(N_q)$ is in bijection with the set of injective
  maps $N_q\to\names$, so $N_q$ behaves like a set of registers. We
  let~$G$ act on left cosets by
  $\pi_1\cdot(\pi_2\Fix(N_q))=\pi_1\pi_2\Fix(N_q)$, and on states by
  $\pi_1\cdot(q,\pi_2\Fix(N_q))=(q,\pi_1\pi_2\Fix(N_q))$. Moreover,
  $\bar A$ has a free transition
  $(q,\pi \Fix(N_q))\trans{\pi(a)} (q',\pi \Fix(N_{q'}))$ whenever
  $q\trans{a}q'$ in~$A$ and $\pi\in G$, and a bound transition
  $(q,\pi\Fix(N_q))\trans{\scriptnew a}(q',\pi'\Fix(N_{q'}))$ whenever
  $q\trans{\scriptnew b}q'$ in~$A$, $\pi\in G$, and
  $\langle \pi(b)\rangle \pi\Fix(N_q)=\langle
  a\rangle\pi'\Fix(N_{q'})$. Applying the name dropping construction
  to $\bar A$ produces states of the more general form
  $(q,\pi\Fix(N))$ for $N\subseteq N_q$; such a state restricts
  $(q,\pi\Fix(N_q))$ to~$N$ (and thus represents the situation where
  the content of registers in $N_q\setminus N$ has been lost).
  Transitions are inherited from~$\bar A$ as prescribed by the
  definition of name dropping.

  In the converse direction, we extract an extended bar NFA~$A_0$ from
  an ERNNA~$A=(Q,\to,s,f)$ as follows. We first eliminate
  $\epsilon$-transitions from~$A$ in the standard manner, obtaining an
  $\epsilon$-free ERNNA of the same size and degree that we continue
  to designate as~$A$. Let~$A$ have degree~$k$, and fix a $k$-element
  set $\names_0\subseteq\names$ such that
  $\supp(s)\subseteq\names_0$. Pick an additional name
  $\ast\in\names\setminus\names_0$. We define the state set~$Q_0$
  of~$A_0$ as $Q_0=\{q\in Q\mid\supp(q)\subseteq\names_0\}$, with
  initial state~$s$. The acceptance function is obtained by
  restriction. Between two given states of~$A_0$, we include the same
  free transitions as in~$A$, which are necessarily labelled with
  names from~$\names_0$, and all bound transitions in~$A$ with
  labels~$\newletter a$ such that $a\in\names_0\cup\{\ast\}$.

  Both constructions preserve the accepted bar
  language. For a more detailed explanation of the constructions and 
  for the associated proofs, we refer to~\cite{SchroderEA17}, noting
  that the proofs for ERRNAs work in the same way as for RNNAs,
  with the exception   of the (unproblematic) treatment of $\top$-states.
\end{proof}

\subsection*{Proof of \autoref{thm:incl}}
\begin{proof}
Let $A$ be an RNNA %with initial state $s_1$ and 
and $A'$ an $\epsilon$-free ERNNA. %with initial state $s_2$
%and let $S=\supp(s_1)\cup\supp(s_2)$.
By \autoref{thm:RNNAbarNFA}, we can transform
$A$ to a bar NFA $A_1$,
such that $L_\alpha(A)=L_\alpha(A_1)$ and
the number of states in $A_1$ is linear 
in the number of orbits of $A$ and exponential in $\degree(A)$, and
$\degree(A_1)\leq \degree(A)+1$.
By \autoref{lem:namedrop}, 
we have $L_\alpha(A')=L_\alpha(\mathsf{nd}(A'))$ and
$\mathsf{nd}(A')$ is of degree $\degree(A')$ and has at most 
$n\cdot 2^{\degree(A')}$ orbits, where $n$ is the number of orbits in $A'$;
furthermore, the literally accepted language of $\mathsf{nd}(A')$ is closed
under $\alpha$-equivalence.

We exhibit a $\NExpSpace$ procedure to check that $L_\alpha(A_1)$
is \emph{not} a subset of $L_\alpha(\mathsf{nd}(A'))$.
The claimed bound then follows by Savitch's theorem.
Our algorithm guesses a closed bar string $w$ and an accepting run of $A_1$ on
$w$ for which $\mathsf{nd}(A')$ has no accepting run.
To this end, our construction maintains a pair
$(q,\Gamma)$, where $q$ is a state of $A_1$ that is reachable
by the bar string read so far, and
$\Gamma$ is a set of states of $\mathsf{nd}(A')$, with the intuition
that all states from $\mathsf{nd}(A')$ are reachable by the bar string that
has been read so far. The algorithm
initializes $q$ to the initial state of $A_1$, $S$ to $\emptyset$ 
 and $\Gamma$ to $\{s'\}$ where $s'$ is the initial state of
 $\mathsf{nd}(A')$, and then iterates the following:
\begin{enumerate}
\item Guess a transition $q\stackrel{\barname}{\rightarrow}q'$ in $A_1$
and update $q$ to $q'$.
\item 
Compute the set $\Gamma'$ of all states of $\mathsf{nd}(A)$
that are reachable from some state in $\Gamma$ by an $\barname$-transition.
If $\barname=\newletter a$ for some $a\in\names$, then update $S$
to $S\cup\{a\}$.
If $\barname=\newletter a$ or $\barname\in S$, then put
$\Gamma''=\{q'\in\Gamma\mid f(q')=\top\}$, otherwise put 
$\Gamma''=\emptyset$.
Update $\Gamma$ to 
$\Gamma'\cup\Gamma''$.
\end{enumerate}
The algorithm terminates successfully and reports $L_\alpha(A)\not\subseteq
L_\alpha(A')$ if it reaches $(q,S,\Gamma)$ such that
$q$ is a final state of $A_1$ and there is
no state $q'\in\Gamma$ such that $f(q')=1$ or $f(q')=\top$.

Correctness follows from the proof of Theorem 7.1 in~\cite{SchroderEA17}
together with the observation that whenever the algorithm reaches
$(q,S,\Gamma)$ such that there is $q'\in\Gamma$ such that $f(q')=\top$, 
the state $q'$ accepts all bar strings from $\free{S}$, 
in particular the empty bar string.
Since the algorithm has to ensure that there is no accepting run
of $\mathsf{nd}(A')$ for the guessed word,
any further $\barname$-transitions do not remove $q'$ 
from $\Gamma$ unless $\barname\in\names\setminus S$. This prevents the algorithm from
 terminating successfully until
some free occurrence of a name from $\barname\in\names\setminus S$ is read 
without having previously been bound.

For space usage, we always have $|S|\leq \degree(A_1)$.
Regarding the size of $\Gamma$, we first note
that for each state in $q'\in\Gamma$, we have
$\supp(q')\subseteq \supp(s')\cup\mathsf{names}(A_1)$, where 
$s'$ is the initial state in $\mathsf{nd}(A')$ and
$\mathsf{names}(A_1)$ denotes the set of all names $a$ that occur as 
label of a transition in $A_1$,
either in the form $\newletter a$ or $a$. This is the case
since we start with the initial state $s'$ of
$\mathsf{nd}(A')$ which has support $\supp(s')$ 
and since every $\barname$-transition in
$A_1$ that the algorithm tracks in $\mathsf{nd}(A')$ adds either
nothing (if $\barname\in\names$), or at most $ a$
 (if $\barname=\newletter a$)
to the support of states. Hence we always have $|\supp(q')|\leq
  \degree(A_1)+ \degree(A')$ for all $q'\in\Gamma$.
  Accounting for the order in which the names from $\supp(q')$
  may have been encountered, we hence have that $\Gamma$ contains
  at most $(\degree(A_1)+ \degree(A'))!p$ states, 
  where $p$ is the number of orbits of 
  $\mathsf{nd}(A')$. Recall that $\mathsf{nd}(A')$ has
  $n\cdot 2^{\degree(A')}$ orbits, so we can make do with space
\begin{align*}
  \mathcal{O}(m!\cdot n\cdot 2^{\degree(A')})\subseteq
  \mathcal{O}(n\cdot 2^{m\log m+\degree(A')}),
\end{align*}
  where $n$ is the number of orbits of $A'$
  and $m=\degree(A)+ \degree(A') +1$.
\end{proof}

\section{Details for \autoref{sec:muBar-ERNNA}}

\subsection*{Proof of~\autoref{lemma:guess-free}}

\begin{proof}
 The proof is by induction on $(|v|,\mathsf{u}(\phi),|\phi|)$,
 ordered by lexicographic ordering. Since $0\leq |v|<n$, we have $n\geq 1$.
 \begin{itemize}
 \item If we have $\phi=\phi^{B}_C(a)_n$
 then we are done. This covers the cases %$\phi=\bot$, $\phi=\top$, 
 $\phi=\epsilon$ and $\phi=\negepsilon$.
 \item If $\phi=\psi\wedge\chi$, then we have
 \begin{align*}
    S,v\models\psi\wedge\chi&\Leftrightarrow
   S,v\models\psi\text{ and }S,v,\models\chi\\
   &\Leftrightarrow S,v\models \psi^{B}_C(a)_n\text{ and }
   S,v\models \chi^{B}_C(a)_n\\
   &\Leftrightarrow S,v\models \psi^{B}_C(a)_n\wedge\chi^{B}_C(a)_n
   \\
   &\Leftrightarrow S,v\models (\psi\wedge\chi)^{B}_C(a)_n
 \end{align*}
 where the second equivalence is by the inductive hypothesis.
 The case $\phi=\psi\vee\chi$ is analogous.
 \item If $\phi=\Diamond_{\scriptnew c}\psi$, then we have
 \begin{align*}
    S,v\models\Diamond_{\scriptnew c}\psi&\Leftrightarrow
   \text{ there are }b, v', \chi\text{ such that }\\
   &\quad \quad v\equiv_\alpha 
   \newletter b v',
   \Diamond_{\scriptnew b}\chi\equiv_\alpha 
   \Diamond_{\scriptnew c}\psi, S\cup\{b\},v'\models\chi\\
&\Leftrightarrow
   \text{ there are }b, v', \chi'\text{ such that }\\
   &\quad\quad v\equiv_\alpha 
   \newletter b v',
   \Diamond_{\scriptnew b}\chi'\equiv_\alpha 
   \Diamond_{\scriptnew c}(\psi^{B\cup\{c\}}_{C\cup\{c\}}(a)_{n-1}), S\cup\{b\},v'\models\chi'\\
      &\Leftrightarrow S,v\models \Diamond_{\scriptnew c}(\psi^{B\cup\{c\}}_{C\cup\{c\}}(a)_{n-1})
   \\
   &\Leftrightarrow S,v\models (\Diamond_{\scriptnew c}\psi)^{B}_C(a)_n,
 \end{align*}
 where the second equivalence is by the inductive hypothesis since
 $|v|< n$ implies $|v'|< n-1$ and by picking $\chi'$ and $\chi$ such
 that $\chi'=\chi^{B\cup\{b\}}_{C\cup\{b\}}(a)_{n-1}$.  To see that
 the inductive hypothesis can be applied, note first that if
 $a\neq *$, then every plain name~$d$ that occurs before the first
 free occurrence of $a$ in $v'$ is contained in $B\cup\{b\}$ (namely,
 either~$d$ is free in $\newletter bv'\equiv_\alpha v$, and
 then~$d\in B$ by assumption, or $d=b$). Second, suppose that~$d$ is
 the first freely occurring letter in~$v'$ that is not contained in
 $B\cup\{b\}$ (so $d\neq b$), and that $d\in \FN(\phi)$.  Then
 the first occurrence of~$d$ is also free
 in~$\newletter bv'\equiv_\alpha v$ and we have
 $d\in\FN(\Diamond_{\scriptnew b}\phi)$ so that
 $d\notin C\subseteq C\cup\{b\}$.  The case
 $\phi=\Box_{\scriptnew c}\psi$ is analogous.
\item If $\phi=\mu X.\,\psi$, then we have
 \begin{align*}
    S,v\models\mu X.\,\psi&\Leftrightarrow
   S,v\models\psi[\mu X.\,\psi/X]\\
   &\Leftrightarrow S,v\models (\psi[\mu X.\,\psi/X])^{B}_C(a)_n\\
   &\Leftrightarrow S,v\models (\mu X.\,\psi)^{B}_C(a)_n,
 \end{align*}
where the second equivalence is by the inductive hypothesis
since $\mathsf{u}(\psi[\mu X.\,\psi/X])< \mathsf{u}(\mu X.\,\psi)$
by guardedness of fixpoint variables.
\item Case $\phi=\Diamond_b \psi$:
\begin{itemize}
\item If $b\notin C$, then we have $(\Diamond_b\psi)^{B}_C(a)_n=\bot$.
By assumption, $B\subseteq C$, so $b\notin B$.
Suppose that $d$ is the first freely occuring letter in $v$ that is not contained in $B$
and that $d\in \FN(\Diamond_b\psi)$. Then $d\in C$ by assumption, so 
$d\neq b$. Hence $v\neq bw'$ so that $S,v\not\models\Diamond_b\psi$,
as required.
\item
If $b\in C$ and $b\notin B$, then we have
 \begin{align*}
    S,v\models\Diamond_b \psi&\Leftrightarrow
   v=bv'\text{ and }S,v'\models\psi\\
   &\Leftrightarrow v=bv'\text{ and }S,v'\models\psi^{\emptyset}_{\FN(\psi)}(*)_{n-1}\\
   &\Leftrightarrow S,v\models (\Diamond_b\psi)^{B}_C(a)_n,
 \end{align*}
 where the second equivalence is by the inductive hypothesis since
 $|v|< n$ implies $|v'|<n-1$ and since the assumptions of the inductive
 hypothesis on the name sets hold trivially  for $\emptyset$ and $\mathsf{FN}(\psi)$.
\item
If $b\in C$ and $b\in B$, then we have
 \begin{align*}
    S,v\models\Diamond_b \psi&\Leftrightarrow
   v=bv'\text{ and }S,v'\models\psi\\
   &\Leftrightarrow v=bv'\text{ and }S,v'\models\psi^{B}_C(a)_{n-1}\\
   &\Leftrightarrow S,v\models (\Diamond_b\psi)^{B}_C(a)_n,
 \end{align*}
 where the second equivalence is by the inductive hypothesis; we check
 that the assumptions of the inductive hypothesis hold: Since $|v|< n$
 we have $|v'|< n-1$. In case $a\neq *$, all names that occur freely
 in $v'$ before the first free occurrence of $a$ remain free in
 $v=bv'$, and hence are contained in~$B$ by assumption. Lastly,
 suppose that $d$ is the first freely occurring letter in $v'$ that is
 not contained in $B$ but in $\FN(\psi)$.  Since $b\in B$,~$d$ is also
 the first freely occuring letter in $v=bv'$ that is not contained in
 $B$; moreover, $d\in\FN(\Diamond_b\psi)$ since $d\in\FN(\psi)$. Hence
 $d\in C$ by assumption, as required.
\end{itemize}

\item Case $\phi=\Box_b\psi$:
\begin{itemize}
\item

  If $b\notin C$ and $a\neq *$, then $(\Box_b\psi)^{B}_C(a)_n=
  \epsilon\vee\Diamond_{\newletter c}\vee
  \bigvee_{d\in B\cup\{a\}}\Diamond_{d}\top$.  
By assumption,
  $B\subseteq C$, so $b\notin B$.
   Suppose that $d$ is the first
  freely occuring letter in $v$ that is not contained in $B$ and that
  $d\in \FN(\Box_b\psi)$. Then we have $d\in C$ by assumption, so
  $d\neq b$. Hence, $b$ is not the first letter of $v$.
  Let $v=\sigma v'$. If $\sigma\in\names$, then
  $\sigma\in B\cup\{a\}$ by the assumptions on $B$ and $a$.
  Hence $S,v\models\Box_b\psi$ if and only if
  $S,v\models\epsilon\vee\Diamond_{\newletter c}\vee
  \bigvee_{d\in B\cup\{a\}}\Diamond_{d}\top$, as required.

\item

  If $b\notin C$ and $a=*$, then we have 
 \begin{align*}
    S,v\models\Box_b \psi&\Leftrightarrow
   v=bv'\text{ implies }S,v'\models\psi\\
   &\Leftrightarrow    v=bv'\text{ implies }S,v'
   \models\psi^\emptyset_{\mathsf{FN}(\psi)}(*)_{n-1}\\
   &\Leftrightarrow S,v\models \Box_b(\psi^\emptyset_{\mathsf{FN}(\psi)}(*)_{n-1})
   \end{align*}
 where the second equivalence holds by the inductive hypothesis since
 $|v'|<n-1$ and since the assumptions of the inductive hypothesis on
 the name sets hold trivially for $\emptyset$ and
 $\mathsf{FN}(\psi)$.
  
\item If $b\in C$ and $b\in B$, then the induction continues as in the
  last item of the previous~case.
  
\item If $b\in C$ and $b\notin B$, then 
$(\Box_b\psi)^{B}_C(a)_n=\chi(b)^{B}_C(a)_{n-1}$.

\begin{itemize}

\item[--] If $a=*$, then we have 
 \begin{align*}
    S,v\models\Box_b \psi&\Leftrightarrow
   v=bv'\text{ implies }S,v'\models\psi\\
   &\Leftrightarrow    v=bv'\text{ implies }S,v'
   \models\psi^\emptyset_{\mathsf{FN}(\psi)}(*)_{n-1}\\
   &\Leftrightarrow S,v\models \Box_b(\psi^\emptyset_{\mathsf{FN}(\psi)}(*)_{n-1})\\
   &\Leftrightarrow S,v\models \chi(b)^{B}_C(a)_{n-1},
 \end{align*}
 where the second equivalence holds by the inductive hypothesis since
 $|v'|<n-1$ and since the assumptions of the inductive hypothesis on
 the name sets hold trivially for $\emptyset$ and
 $\mathsf{FN}(\psi)$.

\item[--] If $a=b$, then we have
 \begin{align*}
    S,v\models\Box_a \psi&\Leftrightarrow
   v=av'\text{ implies }S,v'\models\psi\\
   &\Leftrightarrow S,v\models \epsilon\vee\Diamond_{\scriptnew c}\top
   \vee\bigvee\nolimits_{d\in B}\Diamond_d\top\vee\Diamond_a(\psi^{\emptyset}_{\FN(\psi)}(*)_{n-1})\\
   &\Leftrightarrow S,v\models \chi(b)^{B}_C(a)_{n-1},
 \end{align*}
 where the second equivalence is shown as follows. We have
 $v=\epsilon$ or $v=\newletter c v'$ for some letter $c$ and bar
 string $v'$, or $v=dv'$ for some letter $d$ such that $d=a$ or
 $d\neq a$ and $d\in B$ by assumption. All cases are trivial, except
 the case where $v=dv'$ and $d=a$. In this case we have
 $S,v'\models \psi$ if and only if
 $S,v'\models \psi^{\emptyset}_{\FN(\psi)}(*)_{n-1}$ by the inductive
 hypothesis since $|v'|< |n-1|$ and since the assumptions of the
 inductive hypothesis on the name sets hold trivially for $\emptyset$
 and $\mathsf{FN}(\psi)$.

\item[--] If $*\neq a\neq b$, then we have
 \begin{align*}
    S,v\models\Box_b \psi&\Leftrightarrow
   v=bv'\text{ implies }S,v'\models\psi\\
   &\Leftrightarrow S,v\models \epsilon\vee\Diamond_{\scriptnew c}\top
   \vee\bigvee\nolimits_{d\in B\cup\{a\}}\Diamond_d\top\\
   &\Leftrightarrow S,v\models \chi(b)^{B}_C(a)_{n-1},
 \end{align*}
 where the second equivalence holds since we have $v=\epsilon$ or
 $v=\newletter c v'$ for some letter $c$ and bar string $v'$, or
 $v=dv'$ for some letter $d$ such that $d=a$ or $d\neq a$ and $d\in B$
 by assumption. Again, all cases are trivial. Note that in cases where
 $d \neq a$ and $d \in B$ we have $d \neq b$ since we are in the case
 where $b\notin B$. Thus
 $\Diamond_b(\psi^{\emptyset}_{\FN(\psi)}(*)_{n-1})$ need not occur in
 the second line.\qedhere
\end{itemize}
 \end{itemize}
\end{itemize}
\end{proof}

\subsection*{Proof of~\autoref{lemma:name-restriction}}
\begin{proof}
  We first note that the formulae in the claim are actually defined,
  and that
\begin{align*}
(\pi\cdot\phi)^\emptyset_B(a)^\emptyset_{|w|+1}&=
(\pi\cdot\phi)^\emptyset_B(a)_{|w|+1} &\text{and}\qquad 
(\pi'\cdot\phi)^\emptyset_B(a)^\emptyset_{|w|+1}&=
(\pi'\cdot\phi)^\emptyset_B(a)_{|w|+1}.
\end{align*}
We distinguish cases.
\begin{enumerate}
\item If the first freely occurring letter in $w$ that is not
  contained in $B$ is also not contained in $A\cup A'$ (or if no such
  name exists), then we have $a=*$. Since $*\notin\emptyset$ and $\emptyset\subseteq B$,
  and since the
  first freely occurring letter in $w$ that is not contained in $B$ is
  also not contained in $A\cup A'$ and hence neither contained in $A=\FN(\pi\cdot \phi)$
  nor in $A'=\FN(\pi'\cdot \phi)$, we have the following equivalences
  by \autoref{lemma:guess-free}:
\begin{align*}
  S,w\models\pi\cdot\phi
  &\iff
  S,w\models(\pi\cdot\phi)^\emptyset_B(a)_{|w|+1}, \\
  S,w\models\pi'\cdot\phi
  &\iff
  S,w\models (\pi'\cdot\phi)^\emptyset_B(a)_{|w|+1}. 
\end{align*}
As $a\notin A'\cup A$, we also have 
\[
  (\pi\cdot \phi)^{B}_A(a)^{A'}_{|w|+1}
  =
  (\pi\cdot \phi)^{B}_A(a)_{|w|+1}
  \qquad\text{and}\qquad
  (\pi'\cdot \phi)^{B}_{A'}(a)_{|w|+1}^A
  =
  (\pi'\cdot \phi)^{B}_{A'}(a)_{|w|+1}.
\]
Since $a=*$, the assumptions of~\autoref{lemma:guess-free} on the
argument sets $B$ and $A$ (and $A'$, respectively), hold trivially: We
have % $*\notin B$,
$B\subseteq A$, $B\subseteq A'$ and since $\FN(\pi\cdot\phi)=A$ and
$\FN(\pi\cdot\phi)=A'$, the assumptions on the first freely occuring
letter that is not in $B$ hold trivially. Hence we also have the
following equivalences by \autoref{lemma:guess-free}:
\begin{align*}
  S,w\models\pi\cdot\phi
  &\iff
  S,w\models (\pi\cdot \phi)^{B}_A(a)_{|w|+1},
  \\
  S,w\models\pi'\cdot\phi
  &\iff
  S,w\models (\pi'\cdot \phi)^{B}_{A'}(a)_{|w|+1}.
\end{align*}
Thus, we are done.
\item Otherwise, we have by the assumptions of the lemma that $a$ is
  the first freely occurring letter in $w$ such that $a\notin B$, and
  $a\in A\cup A'$.  We assume w.l.o.g.~that $a\in A$; then
  $a\notin A'$ since $a\notin B$. Towards
  applying~\autoref{lemma:guess-free} to the formulae $\pi'\cdot\phi$
  and $(\pi'\cdot\phi)^\emptyset_B(a)_{|w|+1}$, suppose that $d$ is
  the first letter that occurs freely in~$w$ (and which is not
  contained in $\emptyset$) and that
  $d\in\FN(\pi'\cdot\phi)=A'\subseteq A\cup A'$; we then have to show
  $d\in B$.  But since $a\notin A'$, we have $d\neq a$, so~$d$ occurs
  freely in~$w$ strictly before~$a$, whence $d\in B$ by the
  assumptions on~$a$. Hence,  by~\autoref{lemma:guess-free} we have
  \[
    S,w\models\pi'\cdot\phi
    \iff
    S,w\models(\pi'\cdot\phi)^\emptyset_B(a)_{|w|+1}.
  \]
  Since $a\notin A'$, we furthermore have
  \[
    (\pi\cdot \phi)^{B}_A(a)_{|w|+1}^{A'}
    =
    (\pi\cdot \phi)^{B}_A(a)_{|w|+1}.
  \]
  Towards applying~\autoref{lemma:guess-free} to the formulae
  $\pi\cdot \phi$ and $(\pi\cdot \phi)^{B}_A(a)_{|w|+1}$, suppose
  that~$d$ is the first letter that occurs freely in $w$ and which is
  not contained in $B$ and that $d\in\FN(\pi\cdot\phi)=A$.  Then
  $d\in A$, trivially.   Also, every letter that occurs freely in
  $w$ before the first free occurrence of $a$ in $w$ is contained in
  $B$ by the assumptions of this lemma.  Hence,
  by~\autoref{lemma:guess-free} we have
  \[
    S,w\models\pi\cdot \phi \iff S,w\models(\pi\cdot
    \phi)^{B}_A(a)_{|w|+1}.
  \]
  This finishes the ``$\Rightarrow$''-direction of the proof.

  For ``$\Leftarrow$'', it follows from what we have just shown above
  that 
  \begin{align*}
    S,w\models (\pi\cdot\phi)^{B}_A(a)_{|w|+1}^{A'}\wedge(\pi'\cdot\phi)^\emptyset_B(a)^\emptyset_{|w|+1}
    \qquad\text{implies}\qquad
    S,w\models 
    (\pi\cdot\phi)\wedge (\pi'\cdot\phi).
  \end{align*}
  It remains to show that
  \begin{align*}
    S,w\models (\pi'\cdot\phi)^{B}_{A'}(a)_{|w|+1}^A\wedge(\pi\cdot\phi)_B^\emptyset(a)_{|w|+1}^\emptyset
    \qquad\text{implies}\qquad
    S,w\models 
    (\pi\cdot\phi)\wedge (\pi'\cdot\phi).
  \end{align*}
  Since $a\in A$, we have
  $(\pi'\cdot\phi)^{B}_{A'}(a)_{|w|+1}^A=\bot$, and so we are
  done.\qedhere
\end{enumerate}
\end{proof}
We assume w.l.o.g.\ that formulae~$\phi$ are~\emph{clean}, i.e.~bound
fixpoint variables in~$\phi$ are mutually distinct and distinct from
all free variables (this can be ensured by suitable renaming of bound
variables; in particular, this may be necessary when a fixpoint
$\mu X.\,\phi$ is unfolded to $\phi[\mu X.\phi/X]$), and then denote
by $\theta(X)$ \emph{the} subformula $\mu X.\,\psi$ of $\phi$ that
binds a bound variable~$X$. % We generally write
% $\BN(w)$ %= \{a\in\names\mid a\text{ is bound in }w\}$.
% for the set of \emph{bound} names of~$w$, i.e.\
% those~$a$ for which~$\newletter a$ occurs in~$w$.
\begin{defn}[Closure]
We define the \emph{closure} $\clos(\phi)$ of a formula
 $\phi\in\formulae$
recursively by
\begin{align*}
% \clos(\top)&=\{\top\} &
% \clos(\bot)&=\{\bot\}\\
\clos(\epsilon)&=\{\epsilon\} &
\clos(\negepsilon)&=\{\negepsilon\}\\
\clos(\psi_1\vee\psi_2)&=\{\theta^*(\psi_1\vee\psi_2)\}\cup\clos(\psi_1)\cup\clos(\psi_2) &
\clos(X)&=\emptyset\\
\clos(\psi_1\wedge\psi_2)&=\{\theta^*(\psi_1\wedge\psi_2)\}\cup\clos(\psi_1)\cup\clos(\psi_2) &
\clos(\mu X.\,\psi)&=\{\theta^*(\mu X.\,\psi)\}\cup\clos(\psi)\\
\clos(\Diamond_a\psi)&=\{\theta^*(\Diamond_a\psi)\}\cup\clos(\psi) &
\clos(\Box_a\psi)&=\{\theta^*(\Box_a\psi)\}\cup\clos(\psi)\\
\clos(\Diamond_{\scriptnew a}\psi)&=\{\theta^*(\Diamond_{\scriptnew a}\psi)\}\cup\clos(\psi) &
\clos(\Box_{\scriptnew a}\psi)&=\{\theta^*(\Box_{\scriptnew a}\psi)\}\cup\clos(\psi)
\end{align*}
where $\theta^*(\phi)$ denotes the formula that is obtained from $\phi$ by
repeatedly replacing free occurrences of fixpoint variables $X$ in $\phi$ with 
$\theta(X)$.
\end{defn}
\begin{fact}
For all formulae $\phi$, we have $|\clos(\phi)|\leq |\phi|$. 
\end{fact}

\subsection*{Proof of \autoref{thm:mubar-ernna}}
We construct an ERNNA that uses macro-states 
that are built over micro-states $\pi\cdot(\psi^B_C(a))$, that
is, annotated formulae; here $\psi$ is a formula, $\pi$ a permutation, 
$B,C\subseteq \names$ are sets of names and 
$a\in(\names\setminus B)\cup\{*\}$;
we refer to $a$ as the \emph{distinguishing letter (component)} of the
respective annotated formula.
We will encode formulae of
the shape 
$(\pi\cdot\psi)^{\pi\cdot B}_{\pi\cdot C}(\pi\cdot a)_{n}$, 
as used in~\autoref{lemma:name-restriction},
by micro-states $\pi\cdot(\psi^{B}_{C}(a))$.
A macro-state will consist of a set of such micro-states plus one name,
the so-called \emph{reserve name}.

Recall that $\bot$ and $\top$ abbreviate $\epsilon\wedge\negepsilon$
and $\epsilon\vee\negepsilon$, respectively.

\begin{defn}[Name restriction]\label{defn:nameres}
Given a closed formula $\phi$, we put
\[
  \clo=\clos(\phi)\cup\{\Diamond_\barname\psi\mid \Box_\barname\psi\in\clos(\phi)\}
  \cup \{\epsilon,\bot\}\cup \{\Diamond_b\top\mid b\in\BN(\phi)\},
\]
noting $|\clo|\leq 4m$.
Furthermore, we put%
\begin{align*}
\mathsf{formulae}&=\{\psi^{B}_C(a)\mid \psi\in\clo, B\subseteq \N(\phi),C\subseteq\FN(\psi),
a\in C\cup\{*\}\},
\end{align*}

noting that the cardinality of $\mathsf{formulae}$ is linear in~$m$
and exponential in~$k$; specifically,
$|\mathsf{formulae}|\leq |\clo|\cdot (|\N(\phi)|+1)\cdot
2^{2|\N(\phi)|}\leq 2^{2k}\cdot 4m(k+1)\leq
2^{2m}\cdot 5m^2$ since we have $|\FN(\psi)|\leq |\N(\phi)|$ for all
$\psi\in\clo$ and since $k\leq m$.
Given a set $\Delta$ of annotated formulae (of the
general shape $\pi\cdot\psi^{B}_C(a)$) such that there are
$\psi\in\clo$, $\pi,\pi'\in G$, $B\subseteq\N(\phi)$,
$C\subseteq\FN(\psi)$ and $a\in C\cup\{*\}$ such that
$\pi\cdot A\neq \pi'\cdot A$ where $A:=\FN(\psi)\cap C$, and such that
$\Delta$ contains both $\pi\cdot(\psi^B_C(a))$ and
$\pi'\cdot(\psi^B_C(a))$, we say that $\Delta$ contains two
\emph{instances} of an annotated formula, and define a \emph{name
  restriction step}, following~\autoref{lemma:name-restriction} as
follows.  Put $D=(\pi\cdot A)\cap (\pi'\cdot A)$, $E=\pi^{-1}\cdot D$
(so $E\subseteq A$ and $\pi\cdot E = D$), and $E'=(\pi')^{-1}\cdot D$
(so $E'\subseteq A$ and $\pi'\cdot E' = D$).  For each $b\in \names$
such that $\pi(b)\in (\pi\cdot A)\setminus (\pi'\cdot A)$, and for $b=*$, put
\begin{align*}
\Delta(b)&=(\Gamma'\setminus\{\pi\cdot(\psi^B_C(a)),\pi'\cdot(\psi^B_C(a))\})\cup
\{\pi\cdot(\psi^{B\cap E}_{A}(b)),\pi'\cdot(\psi^\emptyset_{E'}(b))\}
\end{align*}
and for each $b\in\names$ such that $\pi'(b)\in (\pi'\cdot A)\setminus (\pi\cdot A)$, put
\begin{align*}
\Delta(b)&=(\Gamma'\setminus\{\pi\cdot(\psi^B_C(a)),\pi'\cdot(\psi^B_C(a))\})\cup
\{\pi'\cdot(\psi^{B\cap E'}_{A}(b)),\pi\cdot(\psi^\emptyset_{E}(b))\}.
\end{align*}
Then we say that $\Delta$ \emph{restricts} (in one step) to the set 
$\{\Delta(b)\mid b\in (((\pi\cdot A)\cup (\pi'\cdot A'))\setminus D)\cup\{*\}\}$.
The sets $\Delta(b)$ may in turn contain two instances of a (differently!) annotated formula.
Let $\Phi$ be a set of sets of annotated formulae.
The set of sets of annotated formulae to which $\Phi$ \emph{restricts} is obtained
by repeatedly applying name restriction steps to all elements in the set that contain
two instances of an annotated formula. 
\end{defn}

\begin{lemma}
Let $\Phi=\{\Delta_1,\ldots,\Delta_m\}$ be a set of sets of annotated formulae.
Then the name restriction process defined above terminates after finitely many name restriction
steps and results in a set $\mathsf{rest}(\Phi)=\{\Delta'_1,\ldots,\Delta'_o\}$ of sets of annotated formulae
such that no $\Delta'_i$ contains two instances of an annotated formula.
\end{lemma}
\begin{proof}
  Each name restriction step removes two instances
  $\pi\cdot(\psi^B_C(a))$ and $\pi'\cdot(\psi^B_C(a))$ of an annotated
  formula from some set $\Delta_i$ and replaces it by
  $\pi\cdot(\psi^{B\cap E}_{A}(b))$ and
  $\pi'\cdot(\psi^\emptyset_{E'}(b))$ or by
  $\pi'\cdot(\psi^{B\cap E'}_{A}(b))$ and
  $\pi\cdot(\psi^\emptyset_{E}(b))$ for suitable $b, A, E$ and
  $E'$. We consider just the first case of restriction, where
  $\pi(b)\in (\pi\cdot A)\setminus (\pi'\cdot A)$; the termination
  argument for the second case (where
  $\pi'(b)\in (\pi\cdot A)\setminus (\pi'\cdot A)$) is
  analogous. Indeed we show that the sum of the sizes of the index
  sets yields a termination measure for the restriction procedure,
  i.e.~that $|B|+|C|+|B|+|C|>|E\cap B|+|A|+|E'|$. Since clearly
  $|B|\ge|E\cap B|$ and $|C|\ge|A|$ (recall that
  $A=\FN(\psi)\cap C\subseteq C$), this follows once we show that
  $E'\subsetneq C$. To this end, note that by the assumptions of the
  restriction step, $\pi\cdot A\neq \pi'\cdot A$ and
  $D=(\pi\cdot A)\cap (\pi'\cdot A)$. Therefore,
  $D\subsetneq \pi'\cdot A$, and hence
  $E'=(\pi')^{-1}\cdot D\subsetneq A\subseteq
  C$.  % , even if $B=\emptyset$ (so
  % that $|B|=|E\cap B|=0$).
\end{proof}

\noindent Next, we define what it means for an annotation of a formula to
be correct with respect to a given barstring. Correctness in this sense
will then be used to ensure applicability of~\autoref{lemma:name-restriction}
and~\autoref{lemma:guess-free}.

\begin{defn}
Given a barstring $w\in\barstrings$ and an annotated
formula $\psi^B_C(a)_n$, we say that the annotation
of $\psi$ with $B$, $C$, $a$ and $n$ is \emph{correct}
(w.r.t. $w$) if the following hold.
\begin{enumerate}
\item $B\subseteq C$, $a\in (\names\setminus B)\cup\{*\}$ and $|w|<n$.
\item If $a\neq *$, then~$a$ has a free occurrence in~$w$, and is the
  first freely occurring letter in $w$ that is not contained in $B$.
\item If $a=*$ and there is some freely occurring letter in $w$ that
  is not contained in $B$, then let $d$ be the first such
  letter. Then $d\in\FN(\psi^B_C(a)_n)$ implies $d\in C$.
\end{enumerate}
We say that a set $\Gamma$ of annotated formulae is \emph{correctly
  annotated} w.r.t.~$w$ if for all annotated formulae
$\pi\cdot((\psi)^B_C(a))\in\Gamma$, the annotation of $\pi\cdot\psi$ with $\pi\cdot B$,
$\pi\cdot C$, $\pi\cdot a$ and $|w|+1$ is correct w.r.t.~$w$. Moreover, $\Gamma$ is
\emph{satisfied} by $w$ (in context $S$)
if $S,w\models (\pi\cdot \psi)^{\pi\cdot B}_{\pi\cdot C}(\pi\cdot
a)_{|w|+1}$ for all $\pi\cdot((\psi)^B_C(a))\in\Gamma$.
\end{defn}

\begin{lemma}{\label{lemm:rest-annot}}
  Let $\Gamma$ be a set of annotated formulae, let $S$ be a context
  and let $w$ be a bar string such that $\Gamma$ is correctly
  annotated w.r.t.~$w$ and satisfied by~$w$ (in context $S$). Then
  there is some $\Delta\in\mathsf{rest}(\{\Gamma\})$ that is correctly
  annotated w.r.t.~$w$ and satisfied by $w$ (in context~$S$).
\end{lemma}
\begin{proof}
It suffices to show that for each single name restriction step,
there is a choice of new distinguishing letter that preserves
correctness of annotations w.r.t. $w$ and satisfaction by $w$.
The latter follows by two applications of~\autoref{lemma:name-restriction},
given that $\Gamma$ is correctly annotated w.r.t. $w$.
To see that for each restriction step, there is a choice
of the distinguishing letter that preserves the correctness of annotations,
let $\pi\cdot(\phi^B_C(a))\in\Gamma$ be a formula that is replaced
by some name restriction step
with $\pi\cdot(\phi^{B\cap E}_A(b))$ or with
$\pi\cdot(\phi^{\emptyset}_E(*))$,
where $\pi'$ is some permutation such that $\pi\cdot A\neq \pi'\cdot A$,
where $A=\FN(\phi)\cap C$, and where
$D=(\pi\cdot A)\cap(\pi'\cdot A)$, $E=\pi^{-1}\cdot D$ and
$b\in ((\pi\cdot A)\setminus(\pi'\cdot A)) \cup\{*\}$.
If $w$ contains a free occurrence of a letter not in $\pi\cdot (B\cap E)$ but in 
$\pi\cdot \FN(\phi)$,
then let $\pi\cdot b$ be the first such letter; otherwise, let $b=*$.
In the case that $\pi\cdot(\phi^B_C(a))$ is replaced with $\pi\cdot(\phi^{B\cap E}_A(b))$,
where $\pi\cdot b$ is the distinguishing letter w.r.t. $w$ (or $\pi\cdot b=*$),
the annotation of $\pi\cdot\phi$ with $\pi\cdot(B\cap E)$, $\pi\cdot A$, $\pi\cdot b$ and $|w|+1$ 
is correct w.r.t. $w$ since
\begin{enumerate}
\item $\pi\cdot (B\cap E)\subseteq \pi\cdot A=\pi\cdot (\FN(\phi)\cap C)$ since
$B\subseteq C$ and $E\subseteq A$; also, if $b=*$, then
$\pi\cdot b\in ((\pi\cdot A\setminus(\pi\cdot (B\cap E)))\cap\FN(\pi\cdot\phi))\cup\{*\}$ trivially, and
if $b\neq *$, then
$\pi\cdot b\in ((\pi\cdot A\setminus(\pi\cdot (B\cap E)))\cap\FN(\pi\cdot\phi))\cup\{*\}$ since
we have $\pi\cdot b\in (\pi\cdot A)$ so that $\pi\cdot b\in(\pi\cdot \FN(\phi))=\FN(\pi\cdot \phi)$, and
furthermore $\pi\cdot b\in (\pi \cdot A)\setminus
(\pi' \cdot A)$, $D\subseteq (\pi'\cdot A)$, $E=\pi^{-1}\cdot D$,
and we have $\pi\cdot b\notin (\pi\cdot B)$ by the assumptions of this case.
Furthermore, $|w|+1>|w|$.
\item If $\pi\cdot b\neq *$, then $\pi\cdot b$ is
the first freely occurring letter in $w$ that is not contained in $\pi\cdot (B\cap E)$
by the assumptions of this case.
\item If $\pi\cdot b= *$ and there is some freely occurring letter in $w$ that is not
contained in $\pi\cdot (B\cap E)$, then let $d$ be the first such letter.
If $d\in\FN(\pi\cdot \phi)$, then we have to show $d\in (\pi\cdot A)$. But
$(\pi\cdot A)=\pi\cdot(\FN(\phi)\cap C)=\FN(\pi\cdot \phi)\cap(\pi\cdot C)$, so that we are done
if $d\in(\pi\cdot C)$.
Let $e$ be the first freely occuring letter in $w$ that is not in $(\pi\cdot B)$
and let $e\in\FN(\pi\cdot \phi)$. Then $e\in (\pi\cdot C)$ since $\Gamma$ is correctly annotated w.r.t. $w$ so that the annotation of $\pi\cdot \phi$ with 
$\pi\cdot B$, $\pi\cdot C$, $\pi\cdot a$ and $|w|+1$ is correct 
w.r.t. $w$.
If $d=e$, then we are done. Otherwise, we have $d\in (\pi\cdot B)$ 
since $d$ has a free occurrence before the first free occurrence of $e$ and $e$ is the
first freely occuring letter that is not contained in $(\pi\cdot B)$. 
We are done since $B\subseteq C$.
\end{enumerate}
In the case that $\pi\cdot(\phi^B_C(a))$ is replaced with $\pi\cdot(\phi^{\emptyset}_E(*))$,
the annotation of $\pi\cdot\phi$ with $(\pi\cdot\emptyset)=\emptyset$, $(\pi\cdot E)$, $(\pi\cdot *)=*$ and $|w|+1$ 
is correct w.r.t. $w$ since
\begin{enumerate}
\item $\emptyset\subseteq (\pi\cdot E)$, $*\in (((\pi\cdot E)\setminus\emptyset)\cap\FN(\pi\cdot \phi))\cup\{*\}$ 
and $|w|+1>|w|$, trivially.
\item The distinguishing letter component is $*$, so that there is nothing to show.
\item Let $d$ be the first freely occuring letter in $w$ that is not in $\emptyset$
(that is, $d$ is just the first freely occuring letter in $w$). Also let $d\in\FN(\pi\cdot \phi)$.
It remains to show that $d\in (\pi\cdot E)$ which follows,
since $E\subseteq A$, if $d\in (\pi\cdot A)=
\pi\cdot (\FN(\phi)\cap C)$, which in turn follows if $d\in(\pi\cdot C)$.
Let $e$ be the first freely occuring letter in $w$ that is not in $(\pi\cdot B)$
and let $e\in\FN(\pi\cdot \phi)$. 
Then $e\in (\pi\cdot C)$ since $\Gamma$ is correctly annotated w.r.t. $w$
so that the annotation of $\pi\cdot \phi$ with $\pi \cdot B$, $\pi\cdot C$, $\pi\cdot a$
and $|w|+1$ is correct w.r.t. $w$. 
If $d=e$, then we are done. 
Otherwise, we have $d\in (\pi\cdot B)$ 
since $d$ has a free occurrence before the first free occurrence of $e$ and $e$ is the
first freely occuring letter that is not contained in $(\pi\cdot B)$. 
We are done since $B\subseteq C$.\qedhere
\end{enumerate}
\end{proof}
\noindent We proceed to give the construction of automata from
formulae in detail:
\begin{defn}[Formula automaton]
Let $\phi$ be a closed formula.
We construct an ERRNA $A(\phi)=(Q,\to,s,f)$ with state set
\begin{align*}
Q=\big\{(\{\pi_1\cdot\phi_1,\ldots,\pi_n\cdot\phi_n\},
a)
\mid &\,\pi_1,\ldots,\pi_{n}\in G,\\
&
\{\phi_1,\ldots,\phi_n\}\subseteq\mathsf{formulae}, a\in \names\cup\{*\}\big\}, 
\end{align*}
Here,
the $\pi_i$ act on names as usual and we define $\pi_i(*)=*$.
Furthermore, we put $s=(\{\mathsf{id}\cdot(\phi^\emptyset_{\emptyset}(*))\},*)$.
The shape of an annotated formula $\pi\cdot(\psi^B_C(a))$ is just the
shape of $\psi$.  We define $\prestates$, the set of
\emph{propositional states}, to be the set of all automaton states
$(\Gamma,b)\in Q$ such that $\Gamma$ contains some annotated formula
of the shape $\psi_1\vee\psi_2$, $\psi_1\wedge\psi_2$, or
$\mu X.\,\psi_1$. On the other hand, a \emph{quasimodal state} is a state
that is not a propositional state but contains some annotated formula of the
shape $\negepsilon$ or $\Box_\sigma\psi$; we write $\quasistates$ for
the set of quasimodal states. Finally, a \emph{modal state} is a state that
is neither a propositional state nor a quasimodal state, and we write $\states$ for the
set of modal states.  That is, states contain neither propositional
formulae nor $\Box$-formulae nor $\negepsilon$. For
$(\Gamma,b)\in\prestates\,\cup\,\quasistates$, we put $f(\Gamma,b)=0$.

The outgoing transitions for
$(\Gamma,b)\in\prestates$ are defined as follows. Let
$\psi_1,\psi_2\in\clo$. For every annotated formula $\pi\cdot (\psi^B_C(a))\in\Gamma$
(where $\pi\in G$, $B\subseteq \N(\phi)$, $C\subseteq \FN(\psi)$ and
$a\in C\cup\{*\}$)
such that $\psi=\psi_1\vee\psi_2$, $\psi=\psi_1\wedge\psi_2$, or $\psi=\mu X.\,\psi_1$, we do
the following.
If $\psi=\psi_1\vee\psi_2$, then put
\begin{align*}
\Gamma'_i=((\Gamma\setminus\{\pi\cdot((\psi_1\vee\psi_2)_C^B(a))\}) \
\cup\{\pi\cdot((\psi_i)^B_C(a))\},b),
\end{align*}
for $i\in\{1,2\}$ and
add transitions $(\Gamma,b)\stackrel{\epsilon}{\rightarrow}(\Delta,b)$ 
for all $\Delta\in\mathsf{rest}(\{\Gamma'_1,\Gamma'_2\})$;
if $\psi=\psi_1\wedge\psi_2$, then put 
\begin{align*}
\Gamma'=((\Gamma\setminus\{\pi\cdot((\psi_1\wedge\psi_2)^B_C(a))\})\cup\{\pi\cdot ((\psi_1)^B_C(a)),\pi\cdot((\psi_2)^B_C(a))\},b)
\end{align*}
and add transitions $(\Gamma,b)\stackrel{\epsilon}{\rightarrow}
(\Delta,b)$ for all $\Delta\in\mathsf{rest}(\{\Gamma'\})$;
and
if $\psi=\mu X.\,\psi_1$ then put
\begin{align*}
\Gamma'=(\Gamma\setminus\{\pi\cdot((\mu
X.\,\psi_1)^B_C(a))\}) \cup\{\pi\cdot((\psi_1[\mu X.\,\psi_1/X])^B_C(a))\}
\end{align*}
and add a transition $(\Gamma,b)\stackrel{\epsilon}{\rightarrow}
(\Delta,b)$ for each $\Delta\in\mathsf{rest}(\{\Gamma'\})$.
In all these cases, we always have $(\Delta,b)\in Q$. 
For $(\Gamma,b)\in\quasistates$ we define the outgoing transitions of
$(\Gamma,b)$ as follows:
\begin{enumerate}
\item If $\Gamma$ contains some annotated $\Diamond$-formula, then let
  $\mathsf{s}_\Gamma$ initially consist of exactly the annotated
  formulae from~$\Gamma$ whose shape is neither $\Box_\sigma\psi$ nor
  $\negepsilon$.  For each
  $\pi\cdot((\Box_{\scriptnew a}\phi)^B_C(b))\in\Gamma$ such that
  there is some
  $\pi'\cdot((\Diamond_{\scriptnew c}\psi)^{B'}_{C'}(b'))\in \Gamma$,
  add $\pi\cdot((\Diamond_{\scriptnew a}\phi)^B_C(b))$ to
  $\mathsf{s}_\Gamma$.  For each
  $\pi\cdot ((\Box_a\phi)_C^B(b))\in\Gamma$ such that
  there is some annotated
  formula $\pi'\cdot((\Diamond_{a'}\psi)_{C'}^{B'}(b'))\in \Gamma$
  such that $\pi(a)=\pi'(a')$, we distinguish cases as follows,
  following the definition of annotated
  $\Box_{\pi(a)}$-formulae.
  \begin{itemize}
  \item If $a\notin C$ and $b\neq *$, then create a new 
  set $s_\Gamma(\theta)$ for each 
  \begin{align*}\theta\in\{\mathsf{id}
  \cdot(\epsilon^\emptyset_\emptyset(*)),
  \mathsf{id}\cdot(\Diamond_{\scriptnew c}
  \top^\emptyset_\emptyset(*)),
      \pi\cdot ((\Diamond_d\top)^\emptyset_{\{d\}}(*)) 
      \mid d\in B\cup\{a\} \}
  \end{align*}
  by adding $\theta$ to $s_\Gamma$ and continue the replacement of 
  boxes in each  $s_\Gamma(\theta)$; 
  \item If $a\notin C$ and $b=*$,  then add $\pi\cdot
   ((\Diamond_a\phi)^\emptyset_{\FN(\phi)}(*))$ to
  $\mathsf{s}_\Gamma$;
  \item If $a\in C$ and $a\in B$, then add $\pi\cdot 
  ((\Diamond_a\phi)_C^B(b))$ to
  $\mathsf{s}_\Gamma$; 
  \item if $a\in C$, $a\notin B$ and $b=*$, then add $\pi\cdot 
  ((\Diamond_a\phi)^\emptyset_{\FN(\phi)}(*))$ to
  $\mathsf{s}_\Gamma$;
  \item if $a\in C$, $a\notin B$ and $b=a$, then create a new set 
  $s_\Gamma(\theta)$
  for each 
  \begin{align*}\theta\in\{\mathsf{id}\cdot(\epsilon^\emptyset_\emptyset(*)),
  \mathsf{id}\cdot(\Diamond_{\scriptnew c}\top^\emptyset_\emptyset(*)),
      \pi\cdot ((\Diamond_d\top)^\emptyset_{\{d\}}(*)),
    \pi\cdot(\Diamond_{b}\psi^{\emptyset}_{\FN(\psi)}(*)) \mid d\in B \}
  \end{align*}
  by adding $\theta$ to $s_\Gamma$ and continue the replacement of boxes
  in each $s_\Gamma(\theta)$;
    \item if $a\in C$, $a\notin B$ and $*\neq b\neq a$, then create a new set $s_\Gamma(\theta)$
  for each 
  \begin{align*}\theta\in\{\mathsf{id}\cdot(\epsilon^\emptyset_\emptyset(*)),
  \mathsf{id}\cdot(\Diamond_{\scriptnew c}\top^\emptyset_\emptyset(*)),
      \pi\cdot ((\Diamond_d\top)^\emptyset_{\{d\}}(*)) \mid d\in B\cup\{a\} \}
  \end{align*}
  by adding $\theta$ to $s_\Gamma$ and continue the replacement of boxes in each
  $s_\Gamma(\theta)$.
  \end{itemize}  
  The above process results in a set $\Phi$ of sets of annotated formulae.
  Add transitions
  $(\Gamma,b)\stackrel{\epsilon}{\to}(\Delta,b)$ for all
  $\Delta\in\mathsf{rest}(\Phi)$. Note that
  $(\mathsf{s}_\Gamma,b)$ is a modal state by construction.
  
\item If $\Gamma$ contains no annotated $\Diamond$-formula, then we define the
  outgoing transitions of $(\Gamma,b)$ as follows: If 
  $\{\pi\cdot(\epsilon_C^B(a)),\pi'\cdot (\negepsilon^{B'}_{C'}(a'))\}\subseteq \Gamma$
  for some $\pi,\pi',B,B',C,C',a$ and $a'$, then we add the
  single transition $(\Gamma,b)\stackrel{\epsilon}{\to}(\{\mathsf{id}\cdot
  (\epsilon_\emptyset^\emptyset(*)),\mathsf{id}\cdot (\negepsilon^{\emptyset}_{\emptyset}(*))\},b)$.
  Otherwise, if there are no $\pi,B,C$ and $a$ such that
  $\pi\cdot(\negepsilon^B_C(a)))\in \Gamma$, then we add a transition
  \begin{align*}
  (\Gamma,b)\stackrel{\epsilon}{\to}(\{\mathsf{id}\cdot (\epsilon^\emptyset_\emptyset(*))\},b).
  \end{align*} 
  If there are no $\pi,B,C$ and $a$ such that
  $\pi\cdot(\epsilon^B_C(a)))\in \Gamma$ (but possibly $\pi\cdot(\negepsilon_C^B(a))\in \Gamma$
  for some $\pi,B,C$ and $a$), then let
  $\{\pi_1\cdot((\Box_{|a_1}\psi_1)^{B_1}_{C_1}(b_1)),\ldots,
  \pi_n\cdot((\Box_{|a_n}\psi_n)^{B_n}_{C_n}(b_n))\}\subseteq\Gamma$ be
  the set of annotated $\Box_\scriptnew$-formulae in $\Gamma$.  If this set is
  non-empty, then we add a transition
  \begin{align*}
  (\Gamma,b)\stackrel{\epsilon}{\to}(\{\pi_1\cdot((\Diamond_{|a_1}\psi_1)^{B_1}_{C_1}(b_1)),\ldots,
  \pi_n\cdot((\Diamond_{|a_n}\psi_n)^{B_n}_{C_n}(b_n))\},b).
  \end{align*}
  If the   above set is empty, then we add a transition
  $(\Gamma,b)\stackrel{\epsilon}{\to} (\{\mathsf{id}\cdot
   (\Diamond_{|c}\top)^\emptyset_\emptyset(*)\},b)$.  For
  each $a\in\names$, let
  $\{\pi_1\cdot((\Box_{a_1}\psi_1)^{B_1}_{C_1}(b_1)),\ldots,\pi_n\cdot((\Box_{a_n}\psi_n)^{B_n}_{C_n}(b_n))\}\subseteq\Gamma$ be the set
  of annotated $\Box$-formulae from $\Gamma$ whose index is the plain
  name~$a$ (that is, we require $\pi_i(a_i)=a$ for $1\leq i\leq n$).
   Whenever this set is non-empty, we add a transition
  \begin{align*}
  (\Gamma,b)\stackrel{\epsilon}{\to}
  (\Delta,b)
  \end{align*}
  for each $\Delta\in \Phi$, where $\Phi$ is the set of sets of annotated formulae
  obtained from
  $\Gamma$  by replacing annotated $\Box_a$-formulae according to the definition
  of annotated formulae, as described in the previous item.
  In the case that there are no $\pi,B,C$ and $a$ such that 
  $\pi\cdot(\epsilon_C^B(a))\in\Gamma$, and we have $b\neq *$, and if $\Gamma$ furthermore
  contains no annotated
  $\Box_b$-formula, then we also add a transition
  $(\Gamma,b)\stackrel{\epsilon}{\to} (\{\mathsf{id}\cdot((\Diamond_b\top)^\emptyset_{\{b\}}(*))\},b)$. 
\end{enumerate}
For $(\Gamma,b)\in\states$, we put $f(\emptyset,b)=\top$, $f(\Gamma,b)=1$ if $\Gamma\neq\emptyset$ and
all annotated formulae in $\Gamma$ are of the shape $\epsilon$,
and $f(\Gamma,b)=0$ otherwise. Let
\begin{align*}
\Gamma=\{\pi_1\cdot((\psi_1)^{B_1}_{C_1}(b_1)),\ldots,\pi_o\cdot((\psi_o)^{B_o}_{C_o}(b_o))\}.
\end{align*}
If there is $a\in\names$ such that $\psi_i=\Diamond_{a_i}\chi_i$,
$(\pi_i\cdot a_i)=a$ 
and $a_i\in C_i$
for all $1\leq i\leq o$, then we distinguish cases.  If
$\chi_i\neq\top$
for some~$i$, then
we add a transition
\begin{align*}
(\Gamma,b)\stackrel{a}{\rightarrow}(\{\pi_1\cdot((\chi_1)^{B_1}_{C_1}(b_1)),
\ldots,\pi_o\cdot((\chi_o)^{B_o}_{C_o}(b_o))\},b);
\end{align*}
otherwise $\chi_i=\top$
for all~$i$, and we add a transition
\begin{align*}
(\Gamma,b)\stackrel{a}{\rightarrow}(\emptyset,b).
\end{align*}
On the other hand, if $\psi_i=\Diamond_{\scriptnew b_i}\chi_i$ for all
$1\leq i\leq o$, then we add, for each $a\in\names$ such that
$a\neq b$ and all annotated formulae $\tau_1\cdot((\theta_1)^{B'_1}_{C'_1}(b'_1)),\ldots,
\tau_o\cdot((\theta_o)^{B'_o}_{C'_o}(b'_o))$ such that
\begin{align*}
\langle b_i\rangle((\pi_i\cdot\chi_i)^{\pi_i\cdot B_i}_{\pi_i\cdot C_i}(\pi_i\cdot b_i)_{n})=\langle a\rangle
((\tau_i\cdot\theta_i)^{\tau_i\cdot B'_i}_{\tau_i\cdot C'_i}(\tau_i\cdot b'_i)_{n})
\end{align*}
for all
$1\leq i\leq o$ and all $n>0$, transitions
$(\Gamma,b)\stackrel{\scriptnew
  a}{\rightarrow}(\Delta,b)$
  for each
  \begin{align*}
  \Delta\in\mathsf{rest}(\{\tau_1\cdot((\theta_1)^{B'_1}_{C'_1}(b'_1)),\ldots,
\tau_o\cdot((\theta_o)^{B'_o}_{C'_o}(b'_o))\}).
\end{align*} 
   If $b=*$, then we also add transitions
$(\Gamma,b)\stackrel{\scriptnew a}{\rightarrow}
(\Delta,a)$; all other modal states have no outgoing
transition.
\end{defn}
\noindent We note first that this construction realizes the claimed
size bounds:
\begin{lemma}
  The number of orbits of the nominal set~$Q$ is at most doubly
  exponential in the degree~$k$ of the formula~$\phi$, and at most
  singly exponential in the formula size~$m$ of~$\phi$. Moreover, the
  degree of~$Q$ (i.e.~the maximal size of the support of elements
  of~$Q$) is at most exponential in~$k$, and polynomial in~$m$.
\end{lemma}
\begin{proof}
  As noted above, the cardinality~$n$ of $\mathsf{formulae}$ is
  exponential in~$k$ and linear in~$m$.  An orbit of~$Q$ is determined
  by the choice of a subset $\Phi\subseteq\mathsf{formulae}$, giving a
  factor that is exponential in~$n$, and an equivalence relation (with
  additional properties) on the disjoint union
  $U:=1+\sum_{\psi\in\Phi}\FN(\psi)$ (the sets~$\FN(\psi)$ are made
  disjoint to reflect the fact that the formulae in~$\Phi$ are renamed
  independently, and the summand~$1$ reflects the reserve name). The
  set~$U$ has size~$\CO(nk)$, so the factor for equivalence relations
  is also exponential in~$n$. This establishes the claimed bound on
  the number of orbits. The support of an element of~$Q$ is at
  most~$|U|=\CO(nk)$, confirming the claimed bound on the degree.
\end{proof}

\begin{lemma}
Let $\phi$ be a closed formula.
The structure~$A(\phi)$ is an ERRNA 
and we have $L_\alpha(A(\phi))=\sem{\phi}$.
\end{lemma}
\begin{proof}
It is easy to see that~$A(\phi)$ is indeed an ERNNA: The set
\begin{align*}
Q=\big\{(\{\pi_1\cdot\phi_1,\ldots,\pi_n\cdot\phi_n\},
a)
\mid &\,\pi_1,\ldots,\pi_{n}\in G,\\
&
\{\phi_1,\ldots,\phi_n\}\subseteq\mathsf{formulae}, a\in \names\cup\{*\}\big\}, 
\end{align*}
is a nominal set by construction.  The sets $\states$, $\prestates$
and $\quasistates$ are equivariant, and since renaming does not affect
the shape of annotated formulae of shape $\epsilon$, $f$ is
equivariant as well. The transition relation~$\rightarrow$ is clearly
equivariant, $\alpha$-invariant and finitely branching up to
$\alpha$-equivalence by construction. % For the claimed size bound,~$Q$
% has at most
% $(2^{|\clo|\cdot (|\N(\phi)|+1)\cdot 2^{2|\N(\phi)|}}+1)^2$ orbits
% since for each subset $\Phi$ of $\mathsf{formulae}$ and for each
% equivalence relation on the set of free names of $\Phi$ in union with
% $\{a\}$ (where the equivalence relation relates all names
% $b\in\FN(\phi_i)$ and $b'\in\FN(\phi_{i'})$ that are identified by the
% respective permutations, that is, such that $\pi_i(b)=\pi_{i'}(b')$, and similarly for names in $\FN(\phi_1)$ and the additional name~$a$),
% there is at most one orbit; there are
% $2^{|\clo|\cdot (|\N(\phi)|+1)\cdot 2^{2|\N(\phi)|}}$ subsets of
% $\Phi$ and at most
% $2^{|\clo|\cdot (|\N(\phi)|+1)\cdot 2^{2|\N(\phi)|}}+1$ equivalence
% relations over $\FN(\Phi)\cup\{a\}$\lsnote{@Ðaniel: I don't quite
%   understand that last estimate; that is, I see that everything works
%   in O-notation but I do not understand the concrete bound. @Lutz: Reformulate things}, yielding
% the claimed bound.  The degree of $A(\phi)$ is the maximal size of the
% support of some state $(\Gamma,b)\in Q$. The size of the support of
% any $(\Gamma,b)\in Q$ is bounded by $|\Gamma|\cdot\degree(\phi)+1$
% since the support of each $\pi\cdot(\psi^B_C(a))\in\Gamma$ is just the
% number of free names of $\psi$, which in turn is bounded by the degree
% $\degree(\phi)$ of $\phi$. Hence the degree of $A(\phi)$ is bounded by
% $|\Gamma|\cdot\degree(\phi)+1= |\clo|\cdot (|\N(\phi)|+1)\cdot
% 2^{2|\N(\phi)|}+1$.

It remains to show that~$L_\alpha(A(\phi))=\sem{\phi}$. Without loss
of generality, we assume that $w$ is clean, that is, that every name
is bound at most once in $w$ (indeed this is  without loss of
generality since both the bar languages of automata and the languages
of formulae are closed under $\alpha$-equivalence).

For the inclusion $L_\alpha(A(\phi))\supseteq\sem{\phi}$, let
$\emptyset,w\models\phi$ and $w\in\free{\emptyset}$.  Recall that the
initial state of $A(\phi)$ is
$s=(\{\mathsf{id}\cdot(\phi^\emptyset_\emptyset(*)),*\})$.  It
suffices to show that either $(w,1)\in \prelang(A(\phi))$ or there are
$v,u\in\barstrings$ such that $w=vu$, $(v,\top)\in \prelang(A(\phi))$.
This amounts to showing that there is a state $q$ such that
$s\stackrel{w}{\to} q$ and $f(q)=1$, or $s\stackrel{v}{\to} q$ and
$f(q)=\top$ where $v$ is a prefix of $w$. We show the following more
general property:
\begin{quote}
  ($\dagger$) Let $v,u\in\barstrings$ such that $w\equiv_\alpha vu$ and let $(\Gamma,b)$ be a
  state such that $s\stackrel{v}{\to}(\Gamma,b)$ and such
  that~$\Gamma$ is correctly annotated w.r.t.~$u$ and satisfied by~$u$
  in context $\BN(v)$.  Then there is some name $b'$ and a state $q$ (possibly,
  $q=(\Gamma,b')$) such that $s\stackrel{v}{\to}(\Gamma,b')$
  and either $(\Gamma,b')\stackrel{u}{\to} q$ and
  $f(q)=1$, or $(\Gamma,b')\stackrel{u'}{\to} q$ and $f(q)=\top$ where
  $u'$ is a prefix of $u$.
\end{quote}
Since $\emptyset,w\models\phi$ and $\phi$ is closed, we have
$\emptyset,w\models(\mathsf{id}\cdot\phi)^{\mathsf{id}\cdot\emptyset}_{\mathsf{id}\cdot\emptyset}(\mathsf{id}(*))_{|w|+1}$
by~\autoref{lemma:guess-free}; also the annotation of $\phi$ with
$\emptyset$, $\emptyset$, $*$ and $|w|+1$ is trivially correct
w.r.t. $w$.  Finally pick $(\Gamma,b)=s$, $v=\epsilon$, $b'=b=*$ and $u=w$ to
see that the stated property indeed implies our proof goal.  Next we
prove the property by induction, using the measure
$(|u|,\kappa_{(\Gamma,b)},\mathsf{u}(\Gamma),|\Gamma|)$, in lexicographic
ordering, where $\kappa_{(\Gamma,b)}=1$ if $\Gamma$ contains some annotated
formula that contains an unguarded $\Box$-formula or formula
of the shape $\negepsilon$ and
$\kappa_{(\Gamma,b)}=0$ otherwise, $\mathsf{u}(\Gamma)$ denotes the sum of
the numbers of unguarded fixpoint operators in annotated formulae in
$\Gamma$, and $|\Gamma|$ denotes the sum of the sizes of annotated
formulae in $\Gamma$ (measured as just the size of the corresponding
unannoted formulae).  We proceed by case distinction on
whether $(\Gamma,b)$ is in $\states$, in $\quasistates$, or in
$\prestates$. We show ($\dagger$) by picking $b'=b$ in all cases below, except 
for the case  where the reserve letter is used
(that is, the case where the next letter in $u$ is some free letter $a$,
$(\Gamma,b)$ is a quasimodal state, and $\Gamma$ does neither contain an
annotated $\Diamond$-formulae nor an annotated $\Box_a$-formula;
in this case we will pick $b'=a$).

If $(\Gamma,b)\in\states$, then we distinguish the following subcases
(if $|u|=0$, then we have $u=\epsilon$ and every annotated formula in
$\Gamma$ is satisfied by $\epsilon$ by assumption, that is, we are
either in case (1) or (2) below, showing that the inductive
construction of an accepting run, which we effectively perform in the
proof, eventually terminates).
\begin{enumerate}
\item If $\Gamma=\emptyset$, then we have $f(\Gamma,b)=\top$ so that we are done.
\item If $\Gamma\neq\emptyset$ and every annotated formula in $\Gamma$ is of the
shape $\epsilon$, then we have 
$\BN(v),u\models\epsilon$ by assumption. Hence
$u=\epsilon$. We are done
since $f(\Gamma,b)=1$ by definition of $f$.
\item If $\Gamma$ 
contains some annotated formula $\rho\cdot(\Diamond_{a'}\xi)^B_C(b)$, then we have
\begin{align*}
\BN(v),u\models(\rho\cdot \Diamond_{a'}\xi)^{\rho\cdot B}_{\rho\cdot C}(\rho\cdot b)_{|u|+1}
\end{align*}
by assumption.
Thus $(\rho\cdot \Diamond_{a'}\xi)^{\rho\cdot B}_{\rho\cdot C}(\rho\cdot b)_{|u|+1}\neq\bot$, 
hence $a:=\rho(a')\in C$. If $a\in B$,
then
$\BN(v),u\models\Diamond_{a}((\rho\cdot \xi)^{\rho\cdot B}_{\rho\cdot C}(\rho\cdot b)_{|u|})$.
If $a\notin B$,
then
$\BN(v),u\models\Diamond_{a}((\rho\cdot \xi)^{\emptyset}_{\FN(\psi)}(*)_{|u|})$.
In both cases, $u=au'$ and for all annotated formulae $\pi\cdot(\psi^B_C(b))\in\Gamma$,
we have $\psi=\Diamond_{a'}\theta$ where $\pi(a')=a$,
since we have $\BN(v),u\models((\pi\cdot\psi)^{\pi\cdot B}_{\pi\cdot C}(\pi\cdot b)_{|u|+1})$ for each $\pi\cdot(\psi^B_C(b))\in\Gamma$ 
by assumption (that is, $\Gamma$ does not contain an annotated formula
of the shape $\epsilon$ or $\Diamond_{\scriptnew b}\xi$). 
Let $\Gamma=\{\pi_1\cdot((\Diamond_{a'_1}\xi_1)^{B_1}_{C_1}(b_1)),\ldots,
\pi_o\cdot((\Diamond_{a'_o}\xi_o)^{B_o}_{C_o}(b_o))\}$ such that
$\pi_i(a'_i)=a$ for $1\leq i\leq o$
and put 
\begin{align*}
\Gamma'=\{\pi_1\cdot(\xi_1)^{B'_1}_{C'_1}(b'_1)),\ldots,\pi_o\cdot((\xi_o)^{B'_o}_{C'_o}(b'_o))\},
\end{align*}
where, for $1\leq i\leq o$, we put $B'_i=B_i$, $C'_i=C_i$  and $b'_i=b_i$ if $a\in\pi_i\cdot B_i$,
and $B'_i=\emptyset$, $C'_i=\FN(\xi_i)$  and $b'_i=*$ if $a\notin\pi_i\cdot B_i$.
Towards applying the inductive hypothesis, we note that 
by the semantics of annotated $\Diamond_a$-modalities and since $\BN(va)=\BN(v)$, we have
$\BN(va),u'\models(\pi_i\cdot\xi_i)^{\pi_i\cdot B'_i}_{\pi_i\cdot C'_i}(\pi_i\cdot b'_i))_{|u'|+1}$ for each $\pi_i\cdot((\Diamond_{a'_i}\xi_i)^{B_i}_{C_i}(b_i))\in\Gamma'$;
furthermore, the annotation of $\pi_i\cdot\xi_i$ with 
${\pi_i\cdot B'_i}$, ${\pi_i\cdot C'_i}$, $\pi_i\cdot b'_i$ and ${|u'|+1}$ is correct w.r.t.
$u'$ since
\begin{enumerate}
\item 
$\pi_i\cdot B'_i\subseteq {\pi_i\cdot C'_i}$, $b_i\in
(({\pi_i\cdot C'_i}\setminus{\pi_i\cdot B'_i})\cap\FN(\pi_i\cdot\xi_i))\cup\{*\}$
and $|u'|+1>|u'|$,
\item if $\pi_i\cdot b'_i\neq *$, then 
$a\in \pi_i\cdot B_i$, $b'_i=b_i$ and $B'_i=B_i$.
By assumption,
the annotation of $\pi_i\cdot\Diamond_{a_i}\xi_i$ with 
${\pi_i\cdot B_i}$, ${\pi_i\cdot C_i}$, $\pi_i\cdot b_i$ and ${|u|+1}$ is correct w.r.t.
$u=au'$ so that $\pi_i\cdot b_i$ is the first
freely occuring letter in $au'$ that is not in $\pi_i\cdot B_i$ 
(but in $\FN(\pi_i\cdot \Diamond_{a_i}\xi_i)$).
Hence $a\neq \pi_i\cdot b_i$ so that $\pi\cdot b_i$ also is the first
freely occuring letter in $u'$ that is not in $\pi_i\cdot B'_i$ (but in 
$\FN(\pi_i\cdot\xi_i)$),
%Furthermore, each letter that has a free occurrence in $u=au'$ before the first
%free occurrence of $b_i$ in $u=au'$ is contained in $\pi_i\cdot B_i$ by assumption.
%Since $a\neq b_i$, each letter that has a free occurrence in $u'$ before the first
%free occurrence of $b_i=b'_i$ in $u'$ is contained in $\pi_i\cdot B_i=\pi_i\cdot B'_i$.
\item if $\pi_i\cdot b'_i=*$, then let $d$ 
be the first freely occuring letter in $u'$ that is not in $\pi_i\cdot B'_i$,
and let this letter exist. Also let $d\in \FN(\pi_i\cdot \xi_i)$. 
It remains to show that $d\in\pi_i\cdot C'_i$.
If $a\notin \pi_i\cdot B_i$,
then $C'_i=\FN(\xi_i)$ so that trivially $d\in\pi_i\cdot C'_i=\pi_i\cdot\FN(\xi_i)=
\FN(\pi_i\cdot \xi_i)$. If $a\in \pi_i\cdot B_i$, then~$d$ also is the first freely occuring letter in $u=au'$ that is not in $\pi_i\cdot B_i$.
Also $d\in \FN(\pi_i\cdot \Diamond_{a_i}\xi_i)$ 
since $d\in \FN(\pi_i\cdot \xi_i)$ so that $d\in \pi_i\cdot C_i=\pi_i\cdot C'_i$ by assumption.
\end{enumerate}
Hence we have shown that $\Gamma'$ is correctly annotated w.r.t.~$u'$
and satisfied by $u'$.  By~\autoref{lemm:rest-annot}, we have some
$\Delta\in\mathsf{rest}(\{\Gamma'\})$ that is correctly annotated
w.r.t. $u'$ and satisfied by $u'$.  By assumption, we have
$s\stackrel{v}{\to}(\Gamma,b)$ and by definition of $A(\phi)$, we have
$(\Gamma,b)\stackrel{a} {\to}(\Delta,b)$, hence
$s\stackrel{va}{\to}(\Delta,b)$.  As we have $|u'|< |u|$, there is --
by the inductive hypothesis -- a state $q'$ such that either
$(\Delta,b)\stackrel{u'}{\to}q'$ and $f(q')=1$, or
$q\stackrel{u''}{\to}q'$ and $f(q')=\top$ where $u''$ is a prefix of
$u'$.  Since $(\Gamma,b)\stackrel{a}{\rightarrow}(\Delta,b)$, we have
either $(\Gamma,b)\stackrel{au'}{\to}q'$ and $f(q')=1$, or
$(\Gamma,b)\stackrel{au''}{\to}q'$ and $f(q')=\top$ where $u''$ is a
prefix of $u'$, as required.
\item If $\Gamma$ contains some annotated formula 
$\rho\cdot(\Diamond_{\scriptnew a}\xi^B_C(b))$,
then we have 
\begin{align*}
\BN(v),u\models(\rho\cdot\Diamond_{\scriptnew a}\xi)^{\rho\cdot B}_{ \rho\cdot C}(\rho\cdot b)_{|u|+1} 
\end{align*} 
by assumption.
Then $\BN(v),u\models(\Diamond_{\scriptnew a}\rho\cdot\xi)^{\rho\cdot B}_{ \rho\cdot C}(\rho\cdot b)_{|u|+1}$ since $\Gamma$ is satisfied by $u$
 so that
$u=\newletter cu'$ for some $c\in\names$
by the semantics of $\Diamond_{\scriptnew a}$-modalities.
Then all formulae in $\Gamma$ are of this shape since
$\BN(v),u\models(\pi\cdot\psi)^{\pi\cdot B}_{\pi\cdot C}(\pi\cdot b)$ 
for each $\pi\cdot(\psi^B_C(b))\in\Gamma$. 
Let 
\begin{align*}
\Gamma'=\{\pi_1\cdot((\Diamond_{\scriptnew b_1}\xi_1)^{B_1}_{C_1}(a_1)),\ldots,
\pi_o\cdot((\Diamond_{\scriptnew b_o}\xi_o)^{B_o}_{C_o}(a_o))\}.
\end{align*}
Towards applying the inductive hypothesis, we note that
by the semantics of $\Diamond_{\scriptnew{}}$-modalities
and since $\BN(v\newletter c)=\BN(v)\cup\{c\}$,
we have some $d\in \names$ and $u''\in\barstrings$ such that $u\equiv_\alpha \newletter d u''$ and, for each $\pi_i\cdot((\Diamond_{\scriptnew b_i}\xi_i)^{B_i}_{C_i}(a_i))\in\Gamma$, 
some $\theta_i$, $B'_i$, $C'_i$, $a'_i$ and some $\rho_i$ such that 
\begin{align*}
\langle b_i\rangle((\pi_i\cdot\xi_i)^{(\pi_i\cdot B_i)\cup\{b_i\}}_{(\pi_i\cdot C_i)\cup\{b_i\}}(\pi_i\cdot a_i)_{|u|})=
\langle d\rangle((\rho_i\cdot\theta_i)^{(\rho_i\cdot B'_i)\cup\{d\}}_{(\rho_i\cdot C'_i)\cup\{d\}}(\rho_i\cdot a'_i)_{|u|})
\end{align*}
 and
$\BN(v\newletter d),u''\models(\rho_i\cdot\theta_i)^{(\rho_i\cdot B'_i)\cup\{d\}}_{(\rho_i\cdot C'_i)\cup\{d\}}(\rho_i\cdot a'_i)_{|u''|+1}$, where there is some $d'$ such that
$\phi_i\cdot d_i'=d$, noting $|u|=|u''|+1$.
Furthermore, the annotation of $\rho_i\cdot\theta_i$ with 
${\rho_i\cdot (B'_i\cup\{d'_i\})}$, ${\rho_i\cdot (C'_i\cup\{d'_i\})}$, $\rho_i\cdot a'_i$ and ${|u''|+1}$ is correct w.r.t.
$u''$ since
\begin{enumerate}
\item ${\rho_i\cdot (B'_i\cup\{d'_i\})}\subseteq {\rho_i\cdot (C'_i\cup\{d'_i\})}$
since $\pi_i\cdot B_i\subseteq \pi_i\cdot C_i$ which implies
$\rho_i\cdot B'_i\subseteq \rho_i\cdot C'_i$; also
$\rho_i\cdot a_i'\in
(({\rho_i\cdot (C'_i\cup\{d'_i\})}\setminus {\rho_i\cdot (B'_i\cup\{d'_i\})})\cap
\FN(\rho_i\cdot\theta_i))
\cup \{*\}$ since
$\pi_i\cdot a_i\in
(((\pi_i\cdot C_i)\setminus (\pi_i\cdot B_i))\cap
\FN(\pi_i\cdot\xi_i))
\cup \{*\}$ and
$|u''|+1>|u''|$.
\item if $(\rho_i\cdot a_i')\neq *$, then we have to show that
$\rho_i\cdot a'_i$ has a free occurrence in $u''$, and is the first freely occuring
letter in $u''$ that is not contained in ${\rho_i\cdot (B'_i\cup\{d'_i\})}$.
Then $(\pi_i\cdot a_i)\neq *$ so that 
$\pi_i\cdot a_i$ has a free occurrence in $u=\newletter c u'$, and is the first freely occuring
letter in $u$ that is not contained in ${\pi_i\cdot B_i}$.
In particular, $(\pi_i\cdot a_i)\neq b_i$.
Since \begin{align*}
\langle b_i\rangle((\pi_i\cdot\xi_i)^{(\pi_i\cdot B_i)\cup\{b_i\}}_{(\pi_i\cdot C_i)\cup\{b_i\}}(\pi_i\cdot a_i))=
\langle d\rangle((\rho_i\cdot\theta_i)^{(\rho_i\cdot B'_i)\cup\{d\}}_{(\rho_i\cdot C'_i)\cup\{d\}}(\rho_i\cdot a'_i)),
\end{align*}
we have $(\pi_i\cdot a_i)=(\rho_i\cdot a'_i)$.
Since $u\equiv_\alpha\newletter d u''$, 
$\rho_i\cdot a'_i$ has a free occurrence in $u''$ and for
any letter $e$ that has a free occurrence in $u''$ before the first free
occurrence of $\rho_i\cdot a'_i$, we either have $e\in\pi_i\cdot B_i$,
hence $e\in\rho_i\cdot B'_i$, or $e=d=\rho_i\cdot d'_i$; in both cases,
we have $e\in {\rho_i\cdot (B'_i\cup\{d'_i\})}$, as required.
\item if $\rho_i\cdot a_i'=*$ and there is some freely occurring letter in $u''$
that is not contained in ${\rho_i\cdot (B'_i\cup\{d'_i\})}$, then let $e$ be the
first such letter. Also let $e\in\FN(\rho_i\cdot \theta)$
so that there is some $e'\in\FN(\theta)$ such that $\rho_i\cdot e'=e$. It remains
to show that $e\in{\rho_i\cdot (C'_i\cup\{d'_i\})}$.
Then $e$ also is the first freely occuring letter in $\newletter d u''$.
Let $f$ be the first freely occuring letter in $u$
that is not contained in ${\pi_i\cdot B_i}$ but in
$\FN(\pi_i\cdot \Diamond_{\scriptnew b_i}\xi_i)$ so that
$f\in (\pi_i\cdot C_i)$ by assumption.
Then there is $f'\in\FN(\Diamond_{\scriptnew b_i}\xi_i)$ such that
$\pi_i\cdot f'=f$.
If $f'=e'$, then we are done since $\pi_i\cdot f'=\rho_i\cdot e'$ 
since 
\begin{align*}
\langle b_i\rangle((\pi_i\cdot\xi_i)^{(\pi_i\cdot B_i)\cup\{b_i\}}_{(\pi_i\cdot C_i)\cup\{b_i\}}(\pi_i\cdot a_i))=
\langle d\rangle((\rho_i\cdot\theta_i)^{(\rho_i\cdot B'_i)\cup\{d\}}_{(\rho_i\cdot C'_i)\cup\{d\}}(\rho_i\cdot a'_i)),
\end{align*}
so that
$\pi_i\cdot f'\in (\pi_i\cdot C_i)$ implies $\rho_i\cdot e'\in
(\rho_i\cdot C'_i)$.
If $f'\neq e'$, then $e=d$ so that
$e\in{\rho_i\cdot (C'_i\cup\{d'_i\})}$ since $\rho_i\cdot d'_i=d$.
\end{enumerate}
Put
\begin{align*}
\Gamma'=\{\rho_1\cdot((\theta_1)^{B'_1\cup\{d'_i\}}_{C'_1\cup\{d'_i\}}(a'_1)),\ldots,\rho_o\cdot((\theta_o)^{B'_o\cup\{d'_i\}}_{C'_o\cup\{d'_i\}}(a'_o))\}.
\end{align*}
Hence we have shown that $\Gamma'$ is correctly annotated w.r.t. $u''$
and satisfied by $u''$.  By \autoref{lemm:rest-annot}, we have some
$\Delta\in\mathsf{rest}(\{\Gamma'\})$ that is correctly annotated
w.r.t.~$u''$ and satisfied by $u''$.  By assumption, we have
$s\stackrel{v}{\to}(\Gamma,b)$ and by definition of $A(\phi)$, we have
$(\Gamma,b)\stackrel{\scriptnew d} {\to}(\Delta,b)$, hence
$s\stackrel {v\scriptnew d}{\rightarrow}(\Delta,b)$.  
Also we have $w\equiv_\alpha v\newletter d u''$.
As $|u''|< |u|$,
there is -- by the inductive hypothesis -- a state $q'$ such that
either $(\Delta,b)\stackrel{u''}{\to}q'$ and $f(q')=1$, or
$(\Delta,b)\stackrel{u'''}{\to}q'$ and $f(q')=\top$ where $u'''$ is a
prefix of $u''$.  Since
$(\Gamma,b)\stackrel {\scriptnew d}{\rightarrow}(\Delta,b)$, we have
either $(\Gamma,b)\stackrel{\scriptnew d u''}{\to}q'$ and $f(q')=1$,
or $(\Gamma,b)\stackrel{\scriptnew d u'''}{\to}q'$ and $f(q')=\top$
where $u'''$ is a prefix of $u''$, as required.
\end{enumerate}
If $(\Gamma,b)\in\quasistates$, then we distinguish the following cases.
\begin{enumerate}
\item If $\Gamma$ contains some annotated formula of the shape
  $\rho\cdot((\Diamond_\sigma\xi)^B_C(b))$, then recall that
  $\BN(v),u\models(\rho\cdot\Diamond_\sigma\xi)^{\rho\cdot
    B}_{\rho\cdot C}(\rho\cdot b)_{|u|+1}$ by assumption so that
  either there is some plain name $a\in\names$ such that all
  $\Diamond$-formulae in $\Gamma$ are of the form
  $\pi\cdot((\Diamond_c\xi)^B_C(b)$ such that $\pi(c)=a$, or all
  $\Diamond$-formulae are of the general form
  $\pi\cdot((\Diamond_{\scriptnew c}\xi)^B_C(b)$. We distinguish
  cases accordingly.
\begin{enumerate}
\item Suppose that there is some $a\in\names$ such that all annotated
  $\Diamond$-formulae in $\Gamma$ are of the form
  $\pi\cdot((\Diamond_c\xi)^B_C(b)$ where $\pi(c)=a$. Pick one such
  annotated formula $\pi\cdot((\Diamond_c\xi)^B_C(b)$.  By
  assumption, we have
  $\BN(v),u\models (\pi\cdot \Diamond_c\xi)^{\pi\cdot B}_{\pi\cdot
    C}(\pi\cdot b)_{|u|+1}$, which is the case if and only if
  $\BN(v),u\models \Diamond_a((\pi\cdot\xi)^{\pi\cdot B}_{\pi\cdot C}(\pi\cdot
  b)_{|u|})$ so that we have $u=au'$.  Let
\begin{align*}
\{\pi_1\cdot((\Box_{a_1}\psi_1)^{B_1}_{C_1}(b_1)),\ldots,\pi_o\cdot((\Box_{a_o}\psi_o)^{B_o}_{C_o}(b_o))\}\subseteq\Gamma
\end{align*}
be the set of all  annotated $\Box$-formulae in $\Gamma$ such that
$\pi_i(a_i)=a$.
Towards applying the inductive hypothesis, we note that
for each $\pi_i\cdot((\Box_{a_i}\psi_i)^{B_i}_{C_i}(b_i))$ from this set, we have
either $a\in\pi_i\cdot C_i$, or not.
If $a\notin\pi_i\cdot C_i$ and $a=*$,
then proceed as in item $\mathsf{i}.$ below. If
$a\notin\pi_i\cdot C_i$ and $a\neq *$,
then proceed as in item $\mathsf{iii}.$ below, but
with the first letter in $u$ being from the set $\pi_i\cdot B_i
\cup\{\pi_i\cdot b_i\}$.
Otherwise, we have $a\in\pi_i\cdot C_i$ and
proceed as follows: 
Either we have $a\in \pi_i\cdot B_i$ and then
$\BN(v),u\models (\pi_i\cdot \Box_{a_i}\psi_i)^{\pi_i\cdot B_i}_{\pi_i\cdot C_i}(\pi_i\cdot b_i)_{|u|+1}$ 
by assumption so that
$\BN(va),u'\models 
(\pi_i\cdot \psi_i)^{\pi_i\cdot B_i}_{\pi_i\cdot C_i}(\pi_i\cdot b_i)_{|u|}$, which in turn implies
$\BN(v),u\models (\pi_i\cdot \Diamond_{a_i}\psi_i)^{\pi_i\cdot B_i}_{\pi_i\cdot C_i}(\pi_i\cdot b_i)_{|u|+1}$;
put $\theta_i=\pi_i\cdot((\Diamond_{a_i}\psi_i)^{B_i}_{C_i}(b_i))$.
Or we have $a\notin \pi_i\cdot B_i$ 
in which case we have $\BN(v),u\models \chi(a)_{\pi_i\cdot C_i}^{\pi_i\cdot B_i}(\pi_i\cdot b_i)_{|u|}$ and distinguish cases as follows.
\begin{enumerate}
\item If $(\pi_i\cdot b_i)=*$, then
  $\BN(v),u\models \Box_a((\pi_i\cdot \psi_i)^{\pi_i\cdot
    \emptyset}_{\pi_i\cdot \FN(\psi)}(\pi_i\cdot *)_{|u|})$ by the
  definition of~$\chi(a)$ (\autoref{def:restriction}), so that
  $\BN(va),u'\models (\pi_i\cdot \psi_i)^{\pi_i\cdot
    \emptyset}_{\pi_i\cdot \FN(\psi)}(\pi_i\cdot *)_{|u|}$ which in
  turn implies
  $\BN(v),u\models (\pi_i\cdot \Diamond_{a_i}\psi_i)^{\pi_i\cdot
    B_i}_{\pi_i\cdot C_i}(\pi_i\cdot b_i)_{|u|+1}$ by
  \autoref{def:restriction}; put
  $\theta_i=\pi_i\cdot((\Diamond_{a_i}\psi_i)^{B_i}_{C_i}(b_i))$.
\item If $(\pi_i\cdot b_i)=a$, then 
\begin{align*}
\chi(a)_{\pi_i\cdot C_i}^{\pi_i\cdot B_i}(\pi_i\cdot b_i)_{|u|}=
\epsilon\vee\Diamond_{\scriptnew c}\top\vee\bigvee_{d\in \pi_i\cdot B_i}\Diamond_d\top\vee
\Diamond_{b_i}((\pi_i\cdot \psi_i)^\epsilon_{\FN(\pi_I\cdot \psi)}(*)_{|u|}).
\end{align*}
Since $u=au'$ and $\pi_i\cdot b_i=a$, this means that
$\BN(v),u\models (\pi_i\cdot \Diamond_{a_i}\psi_i)^{\pi_i\cdot B_i}_{\pi_i\cdot C_i}(\pi_i\cdot b_i)_{|u|+1}$;
put $\theta_i=\pi_i\cdot((\Diamond_{a_i}\psi_i)^{B_i}_{C_i}(b_i))$.
\item If $*\neq (\pi_i\cdot b_i)\neq a$, then
\begin{align*}
\chi(a)_{\pi_i\cdot C_i}^{\pi_i\cdot B_i}(\pi_i\cdot b_i)_{|u|}=
\epsilon\vee\Diamond_{\scriptnew c}\top\vee\bigvee_{d\in (\pi_i\cdot B_i)\cup\{\pi_i\cdot b_i\}}\Diamond_d\top.
\end{align*}
Since $u=au'$ and $\pi_i\cdot b_i=a$, this means that
$\BN(v),u\models (\pi_i\cdot \Diamond_{a_i}\top)^{\pi_i\cdot B_i}_{\pi_i\cdot C_i}(\pi_i\cdot b_i)_{|u|+1}$;
put $\theta_i=\pi_i\cdot((\Diamond_{a_i}\top)^{B_i}_{C_i}(b_i))$.

\end{enumerate}

Let $s_\Gamma$ denote the set of formulae from $\Gamma$
that are not of the shape $\Box_\sigma\psi$ or $\negepsilon$.
Put 
\begin{align*}
\Gamma'=s_\Gamma\cup\{\theta_1,\ldots,\theta_o\}.
\end{align*}
\item All annotated $\Diamond$-formulae in $\Gamma$ are of the
form $\pi\cdot((\Diamond_{\scriptnew a_i}\psi_i)^{B_i}_{C_i}(b_i))$.
Pick one such annotated formula $\pi\cdot((\Diamond_{\scriptnew a_i}\psi_i)^{B_i}_{C_i}(b_i))$.
Since $\BN(v),u\models 
(\pi\cdot \Diamond_{\scriptnew a_i}\psi_i)^{\pi\cdot B_i}_{\pi\cdot C_i}(\pi\cdot b_i)_{|u|+1}$
by assumption, we have
$\BN(v),u\models \Diamond_{\scriptnew a_i}
((\pi\cdot\psi_i)^{\pi\cdot B_i}_{\pi\cdot C_i}(\pi\cdot b_i)_{|u|+1})$ and hence $u=|cu'$
for some $c\in\names$.
Let 
\begin{align*}
\{\pi_1\cdot((\Box_{\scriptnew a_1}\psi_1)^{B_1}_{C_1}(b_1)),\ldots,
\pi_o\cdot((\Box_{\scriptnew a_o}\psi_o)^{B_o}_{C_o}(b_o))\}\subseteq\Gamma
\end{align*} 
be the set of all annotated $\Box_\scriptnew$-formulae in $\Gamma$.
Towards applying the inductive hypothesis, we note that
for each $\pi_i\cdot((\Box_{\scriptnew a_i}\psi_i)^{B_i}_{C_i}(b_i))$ from this set, 
we have $\BN(v),u\models 
(\pi_i\cdot \Box_{\scriptnew a_i}\psi_i)^{\pi_i\cdot B_i}_{\pi_i\cdot C_i}(\pi_i\cdot b_i)_{|u|+1}$ 
by assumption, which implies
$\BN(v)\cup\{c\},u'\models 
(\rho_i\cdot \theta_i)^{\rho_i\cdot B'_i}_{\rho_i\cdot C'_i}(\rho_i\cdot b'_i)_{|u|}$ 
for some $\theta_i$, $\rho_i$, $B'_i$, $C'_i$ and $b'_i$, such that
$\langle a_i\rangle(\pi_i\cdot((\psi_i)^{B_i}_{C_i}(b_i)))=
\langle c\rangle(\rho_i\cdot((\theta_i)^{B'_i}_{C'_i}(b'_i)))$. This in turn implies
$\BN(v),u\models (\pi_i\cdot \Diamond_{\scriptnew a_i}\psi_i)^{\pi_i\cdot B_i}_{\pi_i\cdot C_i}(\pi_i\cdot b_i)_{|u|+1}$.
Let $s_\Gamma$ denote the set of formulae from $\Gamma$
that are not of the shape $\Box_\sigma\psi$ or $\negepsilon$. Put 
\begin{align*}
\Gamma'=S_\Gamma\cup\{\pi_1\cdot((\Diamond_{\scriptnew a_1}\psi_1)^{B_1}_{C_1}(b_1)),\ldots,
\pi_o\cdot((\Diamond_{\scriptnew a_o}\psi_o)^{B_o}_{C_o}(b_o))\}.
\end{align*}
\end{enumerate}
In both cases, $\Gamma'$ is satisfied by $u$,
and in both cases, $\Gamma'$ is obtained from $\Gamma$ by replacing annotated
box formulae with identically annotated diamond formulae;
since $\Gamma$ is correctly annotated w.r.t. $u$,
$\Gamma'$ is correctly annotated w.r.t.~$u$ as well.
By~\autoref{lemm:rest-annot}, we have some
$\Delta\in\mathsf{rest}(\{\Gamma'\})$ that is correctly annotated w.r.t. $u$
and satisfied by $u$ by~\autoref{lemm:rest-annot}. 
By assumption, we have $s\stackrel{v}{\to}(\Gamma,b)$ and
by definition of $A(\phi)$, we have $(\Gamma,b)\stackrel{\epsilon}
{\to}(\Delta,b)$, hence $s\stackrel {v\epsilon}{\rightarrow}(\Delta,b)$.
We have $\kappa_{(\Gamma,b)}=1$ and $\kappa_{(\Delta,b)}=0$ so that,
by the inductive hypothesis, there is a state $q'$ such that either
$(\Delta,b)\stackrel{u}{\to}q'$ and $f(q')=1$, or
$(\Delta,b)\stackrel{u'}{\to}q'$ for some prefix $u'$ of $u$ and $f(q')=\top$.
Then we have $(\Gamma,b)\stackrel{\epsilon u}{\to}q'$
and $f(q')=1$, or
$(\Gamma,b)\stackrel{\epsilon u'}{\to}q'$ for some prefix $u'$ of $u$ and $f(q')=\top$,
as required.

\item If $\Gamma$ contains no annotated formula of the form $\pi\cdot((\Diamond_\sigma\psi)^B_C(a))$,
then we distinguish
the following cases, noting that $\Gamma$ does not contain
both an annotated formula of the shape~$\epsilon$ and
an annotated formula of the shape $\negepsilon$, by assumption:
\begin{enumerate}
\item If $u=\epsilon$, then $\Gamma$ does not contain an annotated formula 
of the form $\pi\cdot(\negepsilon^B_C(a))$
since $\BN(v),\epsilon\not\models(\pi\cdot\negepsilon)^{\pi\cdot B}_{\pi\cdot C}(\pi\cdot a)_{|u|+1}$. Hence there
is a transition $(\Gamma,b)\stackrel{\epsilon}{\to}(\{\mathsf{id}\cdot(\epsilon^\emptyset_\emptyset(*))\},b)$
by definition of $A(\phi)$. We have $\kappa_{(\Gamma,b)}=1$ and 
$\kappa_{(\{\mathsf{id}\cdot(\epsilon^\emptyset_\emptyset(*))\},b)}=0$ 
so that we are done by the inductive hypothesis.
\item If $u=au'$ for some $a\in\names$, then
let 
\begin{equation}\label{eq:free-boxes}
\{\pi_1\cdot((\Box_{a_1}\psi_1)^{B_1}_{C_1}(b_1)),\ldots,\pi_o\cdot((\Box_{a_o}\psi_o)^{B_o}_{C_o}(b_o))\}\subseteq\Gamma
\end{equation}
be the set of all annotated $\Box_a$-formulae in $\Gamma$ 
for which $a\in\pi_i\cdot C_i$
(that is, let this set contain all annotated box formulae 
$\rho\cdot((\Box_{c}\xi)^{B'}_{C'}(b))\in\Gamma$ such that $\pi(c)=a$). 
We first deal with the case that this set is non-empty. We note 
towards applying the inductive hypothesis that
for each $\pi_i\cdot((\Box_{a_i}\psi_i)^{B_i}_{C_i}(b_i))$ from this set, we have
either $a\in\pi_i\cdot C_i$, or not.
If $a\notin\pi_i\cdot C_i$ and $a=*$,
then proceed as in item $\mathsf{i}.$ below. If
$a\notin\pi_i\cdot C_i$ and $a\neq *$,
then proceed as in item $\mathsf{iii}.$ below, but
with the first letter in $u$ being from the set $\pi_i\cdot B_i
\cup\{\pi_i\cdot b_i\}$.
Otherwise, we have $a\in\pi_i\cdot C_i$ and
proceed as follows: 
We have that
either $a\in \pi_i\cdot B_i$ and then
$\BN(v),u\models (\pi_i\cdot \Box_{a_i}\psi_i)^{\pi_i\cdot B_i}_{\pi_i\cdot C_i}(\pi_i\cdot b_i)_{|u|+1}$ 
by assumption so that
$\BN(va),u'\models 
(\pi_i\cdot \psi_i)^{\pi_i\cdot B_i}_{\pi_i\cdot C_i}(\pi_i\cdot b_i)_{|u|}$ which in turn implies
$\BN(v),u\models (\pi_i\cdot \Diamond_{a_i}\psi_i)^{\pi_i\cdot B_i}_{\pi_i\cdot C_i}(\pi_i\cdot b_i)_{|u|+1}$;
then put $\theta_i=\pi_i\cdot((\Diamond_{a_i}\psi_i)^{B_i}_{C_i}(b_i))$.
Or we have $a\notin \pi_i\cdot B_i$ 
in which case we have $\BN(v),u\models \chi(a)_{\pi_i\cdot C_i}^{\pi_i\cdot B_i}(\pi_i\cdot b_i)_{|u|}$
and distinguish cases as follows.
\begin{enumerate}
\item If $\pi_i\cdot b_i=*$, then 
$\BN(v),u\models 
\Box_a((\pi_i\cdot \psi_i)^{\pi_i\cdot \emptyset}_{\pi_i\cdot \FN(\psi)}(\pi_i\cdot *)_{|u|})$ by the
  definition of~$\chi(a)$ (\autoref{def:restriction}), 
so that
$\BN(va),u'\models 
(\pi_i\cdot \psi_i)^{\pi_i\cdot \emptyset}_{\pi_i\cdot \FN(\psi)}(\pi_i\cdot *)_{|u|}$ which in turn implies
$\BN(v),u\models (\pi_i\cdot \Diamond_{a_i}\psi_i)^{\pi_i\cdot B_i}_{\pi_i\cdot C_i}(\pi_i\cdot b_i)_{|u|+1}$ by
  \autoref{def:restriction};
put $\theta_i=\pi_i\cdot((\Diamond_{a_i}\psi_i)^{B_i}_{C_i}(b_i))$.
\item If $\pi_i\cdot b_i=a$, then 
\begin{align*}
\chi(a)_{\pi_i\cdot C_i}^{\pi_i\cdot B_i}(\pi_i\cdot b_i)_{|u|}=
\epsilon\vee\Diamond_{\scriptnew c}\top\vee\bigvee_{d\in \pi_i\cdot B_i}\Diamond_d\top\vee
\Diamond_{b_i}((\pi_i\cdot \psi_i)^\epsilon_{\FN(\pi_I\cdot \psi)}(*)_{|u|}).
\end{align*}
Since $u=au'$ and $\pi_i\cdot b_i=a$, this means that
$\BN(v),u\models (\pi_i\cdot \Diamond_{a_i}\psi_i)^{\pi_i\cdot B_i}_{\pi_i\cdot C_i}(\pi_i\cdot b_i)_{|u|+1}$;
put $\theta_i=\pi_i\cdot((\Diamond_{a_i}\psi_i)^{B_i}_{C_i}(b_i))$.
\item If $*\neq (\pi_i\cdot b_i)\neq a$, then
\begin{align*}
\chi(a)_{\pi_i\cdot C_i}^{\pi_i\cdot B_i}(\pi_i\cdot b_i)_{|u|}=
\epsilon\vee\Diamond_{\scriptnew c}\top\vee\bigvee_{d\in (\pi_i\cdot B_i)\cup\{\pi_i\cdot b_i\}}\Diamond_d\top.
\end{align*}
Since $u=au'$ and $\pi_i\cdot b_i=a$, this means that
$\BN(v),u\models (\pi_i\cdot \Diamond_{a_i}\top)^{\pi_i\cdot B_i}_{\pi_i\cdot C_i}(\pi_i\cdot b_i)_{|u|+1}$;
put $\theta_i=\pi_i\cdot((\Diamond_{a_i}\top)^{B_i}_{C_i}(b_i))$.
\end{enumerate}
Put 
\begin{align*}
\Gamma'=\{\theta_1,\ldots,\theta_o\}.
\end{align*}
Then $\Gamma'$ is satisfied by $u$. 
Furthermore, $\Gamma'$ is obtained from $\Gamma$ by replacing
annotated box formulae with identically annotated diamond formulae;
since $\Gamma$ is correctly annotated w.r.t.~$u$, $\Gamma'$ is correctly
annotated w.r.t.~$u$ as well.  By~\autoref{lemm:rest-annot}, we have some
$\Delta\in\mathsf{rest}(\{\Gamma'\})$ that is correctly annotated
w.r.t. $u$ and satisfied by $u$.  By assumption, we have
$s\stackrel{v}{\to}(\Gamma,b)$ and by definition of $A(\phi)$, we have
$(\Gamma,b)\stackrel{\epsilon} {\to}(\Delta,b)$, hence
$s\stackrel {v\epsilon}{\rightarrow}(\Delta,b)$.  Again we have
$\kappa_{(\Gamma,b)}=1$ and $\kappa_{(\Delta,b)}=0$ so that we are done by inductive
hypothesis.

It remains to deal with the case that the set defined
in~\eqref{eq:free-boxes} is empty, i.e.~$\Gamma$ contains no annotated
$\Box_a$-formula. In this case, we show that all annotated $\Box$-modalities 
(and all $\epsilon$-formulae) in~$\Gamma$ can be evaded 
by accordingly guessing the reserve name component. 
Since $(\Gamma,b)\in\quasistates$, $\Gamma$ must either contain some other modality (that is, either some
annotated formula $\rho\cdot(\Box_c\xi^{B'}_{C'}(b))$ such that
$\rho(c)\neq a$ or some annotated formula
$\rho\cdot(\Box_{\scriptnew c}\xi^{B'}_{C'}(b))$) or some formula of
the shape $\negepsilon$. Note that $a\in\BN(v)$ (since $w$ is
closed and $a\in\FN(u)$). The assumptions of our generalized property~($\dagger$) hold, by
\autoref{lem:guessfn} below, for the state $(\Gamma,a)$ in which the reserve name
component stores the letter $a$ that comes next in the word $u$ (recall $u=au''$);
in particular, $s\stackrel{v}{\to} (\Gamma,a)$.  Put
\begin{align*}
q=(\{\mathsf{id}\cdot((\Diamond_{a}\top)^\emptyset_{\{a\}}(*))\},a)
\end{align*}
so that there is a transition $(\Gamma,a)\stackrel{\epsilon}{\to} q$
by definition of $A(\phi)$. We have
$\BN(v),u\models
(\mathsf{id}\cdot\Diamond_a\top)^{\mathsf{id}\cdot\emptyset}_{\mathsf{id}\cdot{\{a\}}}(\mathsf{id}\cdot
*)_{|u|+1}$ since 
$(\mathsf{id}\cdot\Diamond_a\top)^{\mathsf{id}\cdot\emptyset}_{\mathsf{id}\cdot{\{a\}}}(\mathsf{id}\cdot
*)_{|u|+1}
=\Diamond_a((\mathsf{id}\cdot\top)^{\mathsf{id}\cdot \emptyset}_
{\mathsf{id}\cdot \{a\}}(\mathsf{id}\cdot *)_{|u|})=\Diamond_a\top$ and
$u=au'$. 
Furthermore, $\kappa_{(\Gamma,b)}=1$, $\kappa_q=0$
and the set
$\{\mathsf{id}\cdot((\Diamond_{a}\top)^\emptyset_{\{a\}}(*))\}$ is
satisfied by $u$ and trivially correctly annotated w.r.t.~$u$  so that,
picking $b'=a$,
the inductive hypothesis finishes the case.
\item In the case that $u=\newletter cu'$ for some $c\in\names$, let 
\begin{align*}
\{\pi_1\cdot((\Box_{\scriptnew a_1}\psi_1)^{B_1}_{C_1}(b_1)),\ldots,\pi_o\cdot((\Box_{\scriptnew a_1}\psi_o)^{B_o}_{C_o}(b_o))\}\subseteq\Gamma
\end{align*}
be the set of all annotated $\Box_\scriptnew$-formulae in $\Gamma$. If this set is non-empty,
then put 
\begin{align*}
\Gamma'=\{\pi_1\cdot((\Diamond_{\scriptnew a_1}\psi_1)^{B_1}_{C_1}(b_1)),\ldots,
\pi_o\cdot((\Diamond_{\scriptnew a_1}\psi_o)^{B_o}_{C_o}(b_o))\},
\end{align*}
having $\BN(v),u\models 
(\pi_i\cdot \Box_{\scriptnew a_i}\psi_i)^{\pi_i\cdot B_i}_{\pi_i\cdot C_i}(\pi_i\cdot b_i)_{|u|+1}$
 for $1\leq i\leq o$
by assumption, which implies
$\BN(v)\cup\{d\},u''\models 
(\rho_i\cdot\xi_i)^{\rho_i\cdot B'_i}_{\rho_i\cdot C'_i}(\rho_i\cdot b'_i)_{|u|}$
for suitable $d$, $u''$ and $\rho_i\cdot((\xi_i)^{B'_i}_{C'_i}(b'_i))$.
Then $\BN(v),u\models (\pi_i\cdot 
\Diamond_{\scriptnew a_i}\psi_i)^{\pi_i\cdot B_i}_{\pi_i\cdot C_i}(\pi_i\cdot b_i)_{|u|+1}$.
Hence $\Gamma'$ is satisfied by $u$.
Since $\Gamma$ is correctly annotated w.r.t. $u$ by assumption
and since $\Gamma'$ is obtained from $\Gamma$ by simply replacing boxes with diamonds, 
$\Gamma'$ is correctly annotated w.r.t. $u$ as well.
By~\autoref{lemm:rest-annot}, we have some
$\Delta\in\mathsf{rest}(\{\Gamma'\})$ that is correctly annotated w.r.t. $u$
and satisfied by $u$.
By assumption, we have $s\stackrel{v}{\to}(\Gamma,b)$ and
by definition of $A(\phi)$, we have $(\Gamma,b)\stackrel{\epsilon}
{\to}(\Delta,b)$, hence $s\stackrel {v\epsilon}{\rightarrow}(\Delta,b)$.
Again we have $\kappa_{(\Gamma,b)}=1$ and $\kappa_{(\Delta,b)}=0$ so that
the inductive hypothesis finishes the case.
If $\Gamma$ contains no $\Box_\scriptnew$-formula, then
 $\Gamma$ contains a formula of the shape $\negepsilon$
 since $(\Gamma,b)\in\quasistates$. 
Put $q=(\{\mathsf{id}\cdot((\Diamond_{\scriptnew c}\top)_\emptyset^\emptyset(*))\},b)$
so that there is a transition
$(\Gamma,b)\stackrel{\epsilon}{\to} q$ by definition of $A(\phi)$. Again, the set
$\{\mathsf{id}\cdot((\Diamond_{\scriptnew c}\top)_\emptyset^\emptyset(*))\}$ is
trivially correctly annotated w.r.t. $u$ and satisfied by $u=\newletter c u'$ 
so that the inductive hypothesis finishes the proof.
\end{enumerate}
\end{enumerate}
If $(\Gamma,b)\in\prestates$, then we pick some annotated formula
$\pi\cdot(\psi^B_C(a))\in\Gamma$ such that the top-level operator in $\psi$ is propositional.
\begin{enumerate}
\item If $\psi=\psi_1\vee\psi_2$, then we have
$\BN(v),u\models(\pi\cdot\psi)^{\pi\cdot B}_{\pi\cdot C}(\pi\cdot a)_{|u|+1}$ if and only if 
$\BN(v),u\models (\pi\cdot\psi_1)^{\pi\cdot B}_{\pi\cdot C}(\pi\cdot a)_{|u|+1}$ or 
$\BN(v),u\models (\pi\cdot\psi_2)^{\pi\cdot B}_{\pi\cdot C}(\pi\cdot a)_{|u|+1}$. Pick $i\in\{1,2\}$ such that 
$\BN(v),u\models(\pi\cdot\psi_i)^{\pi\cdot B}_{\pi\cdot C}(\pi\cdot a)_{|u|+1}$.
The annotation of formulae stays correct w.r.t. $u$. Put 
\begin{align*}
\Gamma'=(\Gamma\setminus\{\pi\cdot(\psi^B_C(a))\})\cup\{\pi\cdot((\psi_i)^B_C(a))\}.
\end{align*}
Hence $\Gamma'$ is correctly annotated w.r.t. $u$ and satisfied by $u$.
By~\autoref{lemm:rest-annot}, we have some
$\Delta\in\mathsf{rest}(\{\Gamma'\})$ that is correctly annotated w.r.t. $u''$
and satisfied by $u''$.
By assumption, we have $s\stackrel{v}{\to}(\Gamma,b)$ and
by definition of $A(\phi)$, we have $(\Gamma,b)\stackrel{\epsilon}
{\to}(\Delta,b)$, hence $s\stackrel {v\epsilon}{\rightarrow}(\Delta,b)$.
This implies $s\stackrel{u\epsilon}{\to}q$ since
$s\stackrel{v}{\to}(\Gamma,b)$ by assumption.
We have $\mathsf{u}(\Gamma)=
\mathsf{u}(\Delta)$ and
$|\psi_i|<|\psi|$ so that the inductive hypothesis yields
a state $q'$ such that either $(\Delta,b)\stackrel{u}{\to}q'$ and
$f(q')=1$, or $(\Delta,b)\stackrel{u'}{\to}q'$ and $f(q')=\top$
where $u'$ is a prefix of $u$.
Hence have 
either $(\Gamma,b)\stackrel{\epsilon u}{\to}q'$ and
$f(q')=1$, or $(\Gamma,b)\stackrel{\epsilon u'}{\to}q'$ and $f(q')=\top$
where $u'$ is a prefix of $u$, as required.

\item If $\psi=\psi_1\wedge\psi_2$, then we have
$\BN(v),u\models(\pi\cdot\psi)^{\pi\cdot B}_{\pi\cdot C}(\pi\cdot a)_{|u|+1}$ if and only if 
$\BN(v),u\models (\pi\cdot\psi_1)^{\pi\cdot B}_{\pi\cdot C}(\pi\cdot a)_{|u|+1}$ and 
$\BN(v),u\models (\pi\cdot\psi_2)^{\pi\cdot B}_{\pi\cdot C}(\pi\cdot a)_{|u|+1}$. 
The annotation of formulae stays correct w.r.t. $u$. Put 
\begin{align*}
\Gamma'=(\Gamma\setminus\{\pi\cdot(\psi^B_C(a))\})\cup
\{\pi\cdot((\psi_1)^B_C(a)),\pi\cdot((\psi_2)^B_C(a))\}.
\end{align*}
Hence $\Gamma'$ is correctly annotated w.r.t. $u$ and satisfied by $u$.
By~\autoref{lemm:rest-annot}, we have some 
$\Delta\in\mathsf{rest}(\{\Gamma'\})$ that is correctly annotated w.r.t. $u''$
and satisfied by $u''$.
By assumption, we have $s\stackrel{v}{\to}(\Gamma,b)$ and
by definition of $A(\phi)$, we have $(\Gamma,b)\stackrel{\epsilon}
{\to}(\Delta,b)$, hence $s\stackrel {v\epsilon}{\rightarrow}(\Delta,b)$.
This implies $s\stackrel{u\epsilon}{\to}q$ since
$s\stackrel{v}{\to}(\Gamma,b)$ by assumption.
We have $\mathsf{u}(\Gamma)=
\mathsf{u}(\Delta)$,
$|\psi_1|<|\psi|$ and $|\psi_2|<|\psi|$ so that the inductive hypothesis yields
a state $q'$ such that either $(\Delta,b)\stackrel{u}{\to}q'$ and
$f(q')=1$, or $(\Delta,b)\stackrel{u'}{\to}q'$ and $f(q')=\top$
where $u'$ is a prefix of $u$.
Hence have 
either $(\Gamma,b)\stackrel{\epsilon u}{\to}q'$ and
$f(q')=1$, or $(\Gamma,b)\stackrel{\epsilon u'}{\to}q'$ and $f(q')=\top$
where $u'$ is a prefix of $u$, as required.

\item 
If $\psi=\mu X.\,\psi_1$, then we have
$\BN(v),u\models(\pi\cdot\psi)^{\pi\cdot B}_{\pi\cdot C}(\pi\cdot a)_{|u|+1}$ if and only if 
$\BN(v),u\models(\pi\cdot\psi_1[\mu X.\,\psi_1/X])^{\pi\cdot B}_{\pi\cdot C}(\pi\cdot a)_{|u|+1}$
by the semantics of fixpoint operators. 
The annotation of formulae stays correct w.r.t. $u$.
Put 
\begin{align*}
\Gamma'=(\Gamma\setminus\{\pi\cdot(\psi^B_C(a))\})\cup\{\pi\cdot((\psi_1[\mu X.\,\psi_1/X])^B_C(a))\}.
\end{align*}
Hence $\Gamma'$ is correctly annotated w.r.t. $u$ and satisfied by $u$.
By~\autoref{lemm:rest-annot}, we have some
$\Delta\in\mathsf{rest}(\{\Gamma'\})$ that is correctly annotated w.r.t. $u''$
and satisfied by $u''$.
By assumption, we have $s\stackrel{v}{\to}(\Gamma,b)$ and
by definition of $A(\phi)$, we have $(\Gamma,b)\stackrel{\epsilon}
{\to}(\Delta,b)$, hence $s\stackrel {v\epsilon}{\rightarrow}(\Delta,b)$.
This implies $s\stackrel{u\epsilon}{\to}q$ since
$s\stackrel{v}{\to}(\Gamma,b)$ by assumption.
Fixpoint variables are guarded by modal operators
so that $\mu X.\,\psi_1$ is guarded in $\psi_1[\mu X.\,\psi_1/X]$.
Hence we have
$\mathsf{u}(\Gamma)>
\mathsf{u}(\Delta)$
and the inductive hypothesis yields
a state $q'$ such that either $(\Delta,b)\stackrel{u}{\to}q'$ and
$f(q')=1$, or $(\Delta,b)\stackrel{u'}{\to}q'$ and $f(q')=\top$
where $u'$ is a prefix of $u$.
Hence we have 
either $(\Gamma,b)\stackrel{\epsilon u}{\to}q'$ and
$f(q')=1$, or $(\Gamma,b)\stackrel{\epsilon u'}{\to}q'$ and $f(q')=\top$
where $u'$ is a prefix of $u$, as required.
\end{enumerate}
For the converse inclusion $L_\alpha(A(\phi))\subseteq\sem{\phi}$, let
$[w]_\alpha\in L_\alpha(A(\phi))$ so that either
$(w,1)\in \prelang(A(\phi))$ or $(v,\top)\in \prelang(A(\phi))$ where
$v$ is some prefix of $w$.  The result follows from the more general
property that
\begin{quote}
for all states $(\Gamma,b)\in Q$ and all
$u,v\in\barstrings$ such that $w=uv$ and
$s\stackrel{u}{\to}(\Gamma,b)$,
if there is a state $q$ such that either
$(\Gamma,b)\stackrel{v}{\to}q$ and $f(q)=1$ or
$(\Gamma,b)\stackrel{v'}{\to}q$ and $f(q)=\top$ where $v'$ is a prefix
of $v$, then for all $\pi\cdot(\psi^B_C(a))\in\Gamma$
we have $\BN(u),v\models\pi\cdot\psi$.
\end{quote}
We show this as follows.
Let $\tau=(\Gamma_0,b_0),(\Gamma_1,b_1),\ldots,(\Gamma_n,b_n)$ where
$(\Gamma_0,b_0)=(\Gamma,b)$, $(\Gamma_n,b_n)=q$ and
$f(q)\in\{1,\top\}$ be an accepting path as per the assumption.  The
proof of the result is by induction on the length $|\tau|$ of $\tau$.
If $|\tau|=1$, then $(\Gamma,b)$ is, by assumption an accepting state
or a $\top$-state, that is we have either $f(\Gamma)=\top$ or
$f(\Gamma)=1$.  In the former case, we have $\Gamma=\emptyset$ by
definition of $f$ so that there is nothing to show. 
In the latter case, we have $v=\epsilon$ since $\tau$ is an
accepting path; furthermore we have that every annotated
formula in $\Gamma$ is of the shape $\epsilon$ by definition of $f$.
Then we have $\psi=\epsilon$, $\pi\cdot\epsilon=\epsilon$ and
$\BN(u),v\models\epsilon$. 
If $|\tau|>1$, then there are at least two pairs
$(\Gamma_0,b_0),(\Gamma_1,b_1)$ in $\tau$.  If
$(\Gamma,b)\in\prestates$ and $\pi\cdot(\psi^B_C(a))\in\Gamma_1$, then the inductive
hypothesis finishes the case.  If $(\Gamma,b)\in\prestates$ and
$\Gamma_1\in\mathsf{rest}(\{\Gamma'\})$ where $\Gamma'$ is a set such that
$\pi\cdot(\psi^B_C(a))\notin \Gamma'$, then 
$v=\epsilon v''$ and
the $\epsilon$-transition of the automaton
manipulates $\pi\cdot(\psi^B_C(a))$ which then is of the shape $\psi_1\vee\psi_2$,
$\psi_1\wedge\psi_2$ or $\mu X.\,\psi_1$. Also we have
$\BN(u)=\BN(u\epsilon)$. 
\begin{itemize}
\item If $\psi=\psi_1\vee\psi_2$, then we have $\pi\cdot((\psi_1)^B_C(a))\in\Gamma'$
or $\pi\cdot((\psi_2)^B_C(a))\in\Gamma'$ by definition of $A(\phi)$. Since
$\BN(u),v\models\pi\cdot(\psi_1\vee\psi_2)$ 
if and only if 
$\BN(u\epsilon),v''\models \pi\cdot \psi_1$ or 
$\BN(u\epsilon),v''\models \pi\cdot \psi_2$,
and since $\Gamma_1\in\mathsf{rest}(\{\Gamma'\})$ 
we are by~\autoref{lem:runtosat} done if $\Gamma_1$ is satisfied by $v''$;
this is the case by the inductive hypothesis.
\item If $\psi=\psi_1\wedge\psi_2$, then we have $\pi\cdot((\psi_1)^B_C(a))\in\Gamma'$
and $\pi\cdot((\psi_2)^B_C(a))\in\Gamma'$ by definition of $A(\phi)$. Since 
$\BN(u),v\models\pi\cdot (\psi_1\wedge\psi_2)$
if and only if 
$\BN(u\epsilon),v''\models\pi\cdot \psi_1$ and 
$\BN(u\epsilon),v''\models\pi\cdot \psi_2$, 
and since $\Gamma_1\in\mathsf{rest}(\{\Gamma'\})$ 
we are by~\autoref{lem:runtosat} done if $\Gamma_1$ is satisfied by $v''$;
this is the case by the inductive hypothesis.
\item If $\psi=\mu X.\,\psi_1$, then we have
$\pi\cdot((\psi_1[\mu X.\,\psi_1/X])^B_C(a))\in\Gamma'$ by definition of $A(\phi)$. 
Since $\BN(u),v\models\pi\cdot(\mu X.\,\psi_1)$ 
if and only if 
$\BN(u),v\models
\pi\cdot(\psi_1[\mu X.\,\psi_1/X])$,
and since $\Gamma_1\in\mathsf{rest}(\{\Gamma'\})$ 
we are by~\autoref{lem:runtosat} done if $\Gamma_1$ is satisfied by $v''$;
this is the case by the inductive hypothesis.
\end{itemize}
If $(\Gamma,b)\in\states$, then there is some $\barname\in\barA$ such
that $\Gamma$ has an outgoing $\barname$-transition since $\tau$ does
not end at $(\Gamma,b)$; then $v=\barname v'$.
As $\pi\cdot(\psi^B_C(a))\in\Gamma$,
we have either $\barname=d$, $\psi=\Diamond_c\psi_1$ and
$(\pi\cdot c)=d$ or 
$\sigma=\newletter d$ and $\psi=\Diamond_{\scriptnew c}\psi_1$ for some $c\in\names$
by the definition of $A(\phi)$.  If $\barname=d$, then we have
$(\pi\cdot c)\in (\pi\cdot C)$ since 
the run is accepting; also, $v=dv'$, 
$\pi\cdot((\psi_1)^{B'}_{C'}(a'))\in\Gamma'$ for suitable
$B'$, $C'$ and $a'$, and for some
$\Gamma'$ s.t. $\Gamma_1\in\mathsf{rest}(\{\Gamma'\})$. 
Furthermore, we have $\BN(u)=\BN(ud)$ so that 
$\BN(ud),v'\models \pi\cdot\psi_1$ 
by the inductive hypothesis in combination 
with~\autoref{lem:runtosat}, which 
in turn implies 
$\BN(u),v\models \pi\cdot(\Diamond_d\psi_1)$ 
and hence finishes the case.  If $\barname=\newletter d$,
then we have $v=\newletter d v'$. By definition of
$A(\phi)$, there are $\theta\in\formulae$, a permutation $\rho$,
sets $B',C'$, a name $a'$ and $n>0$ such that
$\langle c\rangle((\pi\cdot\psi_1)^{\pi\cdot B}_{\pi\cdot C}(\pi\cdot a)_n)=
\langle d\rangle ((\rho\cdot\theta)^{\rho\cdot B'}_{\rho\cdot C'}(\rho\cdot a')_n)$ and
$\rho\cdot(\theta^{B'}_{C'}(a'))\in\Gamma'$ for some $\Gamma'$
such that $\Gamma_1\in\mathsf{rest}(\{\Gamma'\})$. 
We have $\BN(u\newletter d)= \BN(u)\cup\{d\}$ so
that $\BN(u)\cup\{d\},v'\models \rho\cdot \theta$ by the inductive
hypothesis in combination with~\autoref{lem:runtosat}, which in turn implies 
$\BN(u),v\models \pi\cdot(\Diamond_{\scriptnew c}\psi_1)$
so that we are done.  

If $(\Gamma,b)\in\quasistates$, then we have
$(\Gamma,b)\stackrel{\epsilon}{\to}(\Gamma_1,b_1)$ and
$(\Gamma_1,b_1)\in\states$, where 
$\Gamma_1\in\mathsf{rest}(\Gamma')$ for suitable $\Gamma'$. 
If $\pi\cdot(\psi^B_C(a))\in\Gamma'$, 
then the inductive hypothesis finishes the
case in combination with~\autoref{lem:runtosat}.
Otherwise, $\psi=\negepsilon$, $\psi=\Box_c\phi$ or 
$\psi=\Box_{\scriptnew c}\phi$ for some $c\in\names$. 
 If $v=bv'$ and $\psi=\Box_{\scriptnew c}\phi$,
then $\BN(u),v\models \pi\cdot (\Box_{\scriptnew c}\phi)$ trivially.  Similarly,
if $v=bv'$ and $\psi=\Box_c\phi$ for $(\pi\cdot c)\neq b$, then
$\BN(u),v\models \pi\cdot(\Box_c\phi)$ trivially,
and if $v\neq\epsilon$ and $\psi=\negepsilon$, 
then $\BN(u),v\models(\pi\cdot\negepsilon)$
trivially. The case $v=\epsilon$ and $\psi=\negepsilon$
does not occur since in this case, the run
ends in $(\Gamma_1,b)$ but we have $f(\Gamma_1,b)=0$ since
$\Gamma_1\neq\emptyset$ which is a contradiction to the run being
accepting.
If $v=bv'$, $\psi=\Box_c\phi$ and $(\pi\cdot c)=b$,
then we distinguish cases:
\begin{enumerate}
\item $b\notin (\pi\cdot C)$, $c=*$: Then 
we have $(\pi\cdot a) = b$ 
so that $(\pi\cdot a)\in\FN(\pi\cdot\Box_c\phi)$ but
$(\pi\cdot a)\notin (\pi\cdot C)$ in which case
the run ends at $(\Gamma_1,b_1)$ -- in contradiction
to the run being accepting, showing that this case does not occur.
\item $b\notin (\pi\cdot C)$, $c\neq *$: Then we have
$\pi\cdot(\Diamond_{c}\phi^{B'}_{C'}(a'))\in\Gamma'$ for suitable
$B'$, $C'$ and $a'$ by
construction of $A(\phi)$.
\item $b\in (\pi\cdot C)$, $b\in(\pi\cdot B)$:  Then we have
$\pi\cdot(\Diamond_{c}\phi^{B'}_{C'}(a'))\in\Gamma'$ for suitable
$B'$, $C'$ and $a'$ by
construction of $A(\phi)$.
\item $b\in (\pi\cdot C)$, $b\notin(\pi\cdot B)$, $c=*$:  Then we have
$\pi\cdot(\Diamond_{c}\phi^{B'}_{C'}(a'))\in\Gamma'$ for suitable
$B'$, $C'$ and $a'$ by
construction of $A(\phi)$.
\item $b\in (\pi\cdot C)$, $b\notin(\pi\cdot B)$, $(\pi\cdot c)=b$: Then we have
$\pi\cdot(\Diamond_{c}\phi^{B'}_{C'}(a'))\in\Gamma'$ for suitable
$B'$, $C'$ and $a'$ by
construction of $A(\phi)$.
\item $b\in (\pi\cdot C)$, $b\notin(\pi\cdot B)$, $*\neq (\pi\cdot c)\neq b$: This case does occur since $(\pi\cdot c)=b$.
\end{enumerate}
In all cases that can occur, we have
$\pi\cdot(\Diamond_{a}\phi^{B'}_{C'}(a'))\in\Gamma'$ 
for suitable $B'$, $C'$ and $a'$ and the inductive hypothesis 
in combination with~\autoref{lem:runtosat} finishes the
respective case. 
If $v=\newletter c v'$ and $\psi=\Box_{a}\phi$ for some
$a\in \names$, then $\BN(u),v\models \pi\cdot (\Box_a\phi)$ trivially.  If
$v=\newletter c v'$ and $\psi=\Box_{\scriptnew a}\phi$ for some
$a\in \names$, then we have $\pi\cdot(\Diamond_{\scriptnew a}\phi^{B'}_{C'}(a'))\in\Gamma'$ for suitable $B'$, $C'$ 
and $a'$ by construction of $A(\phi)$ and the inductive hypothesis 
in combination with~\autoref{lem:runtosat} finishes the
case.
\end{proof}

\begin{lemma}\label{lem:runtosat}
Let $\Gamma_1\in\mathsf{rest}(\Gamma')$ and
let $u,v$ be bar strings such that $\BN(u),v\models\pi\cdot\psi$
for all $\pi\cdot(\psi^B_C(a))\in\Gamma_1$. Then we have
$\BN(u),v\models\pi\cdot\psi$
for all $\pi\cdot(\psi^B_C(a))\in\Gamma'$.
\end{lemma}
\begin{proof}
The set $\Gamma_1$ is obtained from $\Gamma'$ by a sequence of
restriction steps as defined in~\autoref{defn:nameres}, which
change the annotation, but not the shape of formulae.
Hence there are, for each $\pi\cdot(\psi^B_C(a))\in\Gamma'$,
sets $B'$, $C'$ and a name $a'$ 
such that $\pi\cdot(\psi^{B'}_{C'}(a'))
\in\Gamma_1$. Then $\BN(u),v\models\pi\cdot\psi$ by assumption
so that we are done.
\end{proof}

\begin{lemma}\label{lem:guessfn}
For all states $(\Gamma,b)\in A(\phi)$ 
such that $\Gamma\neq\{\mathsf{id}\cdot (\Diamond_b\top)^\emptyset_{\{b\}}(*)\}$,
for all $u\in\barstrings$ such that
$s\stackrel{u}{\to}(\Gamma,b)$ and 
for all $a\in\mathsf{BN(u)}$, we have $s\stackrel{u}{\to}(\Gamma,a)$.
\end{lemma}
\begin{proof}
  There is a path $\tau$ witnessing
  $s\stackrel{u}{\to}(\Gamma,b)$. Traverse the automaton as in $\tau$,
  but keeping~$*$ as the second component of states until the last
  occurrence of $\newletter a$ in $u$ is reached. At this last
  occurrence, take the $\Diamond_{\scriptnew a}$-transition that
  guesses~$a$ in the second component, and afterwards again follow the
  transitions in $\tau$.  The transitions that are required to
  construct this path exist since
  $(\Gamma,b)\stackrel{\sigma}{\to}(\Gamma',b')$ and 
  $\Gamma'\neq \{\mathsf{id}\cdot (\Diamond_{b}\top)^\emptyset_{\{b\}}(*)\}$
  imply
  $(\Gamma,*)\stackrel{\sigma}{\to}(\Gamma',*)$
 (since the only transitions that are enabled only if
    $b\neq *$ are the transitions that use the reserve letter and
    lead to states $\{\mathsf{id}\cdot (\Diamond_b\top)^\emptyset_{\{b\}}(*)\}$)  
  and $(\Gamma,c)\stackrel{a}{\to}(\Gamma',c)$ for
  $c\in\BN(\phi)$, since we use the last occurrence of $\newletter a$
  in $u$ so that taking the transition with $b=*$ and $b'=a$ does not lead to
  blocked transitions. 
\end{proof}

 \subsection*{Proof of \autoref{lem:autformeq}}
 
 \begin{proof}
   Let $A=(Q,\to,s,f)$ be an extended bar NFA. We define the
   \emph{automaton formula} $\mathsf{form}(A)$ for $A$ to be
   $\mathsf{form}(s,\emptyset)$ where $\mathsf{form}(q,Q')$ is defined
   recursively for $q\in Q$ and $Q'\subseteq Q$ by 
   \begin{align*}
     \mathsf{form}(q,Q')=\begin{cases}
       \mu X_q.\,\bigvee_{(q',\barname), q\stackrel{\barname}{\to}q'}
       \Diamond_\barname\mathsf{form}(q',Q'\cup\{q\}) & q\notin Q', f(q)=0\\
       \epsilon & q\notin Q', f(q)=1\\
       \top & q\notin Q', f(q)=\top\\
       X_q & q\in Q'.
     \end{cases}
   \end{align*}
   For one direction of the claimed bar language equality, let
   $[w]_\alpha\in L_\alpha(A)$.  Then there is $v$ such that
   $v\equiv_\alpha w$ and $v\in L_0(A)$.  By~\autoref{lem:alpheq:1}
   of~\autoref{lem:alpheq} it suffices to show
   $\emptyset,v\models \mathsf{form}(A)$.  We show the more general
   property that for all $v',u\in\barstrings$ such that $v=uv'$ and
   all $q$ such that $s\stackrel{u}{\to}q$ and $v'\in L_0(A,q)$, we
   have $\BN(u),v'\models \theta^*(\mathsf{form}(q,Q'))$ where $Q'$ is
   the set of states of some $u$-path from $s$ to $q$.  We proceed by
   induction over $|v'|$. If $|v'|=0$, then $v'=\epsilon$ and since
   $v\in L_0(A)$, $f(q)=1$ or $f(q)=\top$ so that
   $\theta^*(\mathsf{form}(q,Q'))=\mathsf{form}(q,Q')=\epsilon$ and
   $\BN(u),\epsilon\models \epsilon$ trivially, or
   $\theta^*(\mathsf{form}(q,Q'))=\mathsf{form}(q,Q')=\top$ and
   $\BN(u),\epsilon\models \top$ trivially.  If $|v'|>0$, then
   $v'=av''$ or $v'=\newletter av''$ for some $a\in\names$,
   $v''\in\barstrings$. If $f(q)=\top$, then
   $\theta^*(\mathsf{form}(q,Q'))=\mathsf{form}(q,Q')=\top$ and
   $v'\in\free{\BN(u)}$ since $v\in L_0(A)$. Hence
   $\BN(u),v'\models \top$.  If $f(q)\neq\top$ and $v'=av''$, then
   there is some transition $q\stackrel{a}{\to}q'$ such that
   $v''\in L_0(A,q')$. We also have $s\stackrel{ua}{\to}q'$.  In this
   case,
   $\theta^*(\mathsf{form}(q,Q'))=\mu X_{q}.\,
   \bigvee_{(q',\barname)\mid q\stackrel{\barname}{\to}q'}
   \Diamond_\barname\theta^*(\mathsf{form}(q',Q'\cup\{q\}))$.  Hence it
   suffices to show that
   $\BN(u),v'\models
   \Diamond_a\theta^*(\mathsf{form}(q',Q'\cup\{q\}))$ which follows if
   $\BN(u),v''\models\theta^*(\mathsf{form}(q',Q'\cup\{q\}))$.  Since
   $\BN(u)=\BN(ua)$, this follows from the inductive hypothesis.  If
   $f(q)\neq\top$ and $v'=\newletter av''$, then there is some
   transition $q\stackrel{\scriptnew a}{\to}q'$ such that
   $v''\in L_0(A,q')$. We also have
   $s\stackrel{u\scriptnew a}{\to}q'$.  In this case,
   $\theta^*(\mathsf{form}(q,Q'))=\mu X_{q}.\,
   \bigvee_{(q',\barname)\mid q\stackrel{\barname}{\to}q'}
   \Diamond_\barname\theta^*(\mathsf{form}(q',Q'\cup\{q\}))$.  Hence it
   suffices to show that
   $\BN(u),v'\models \Diamond_{\scriptnew
     a}\theta^*(\mathsf{form}(q',Q'\cup\{q\}))$ which follows if
   $\BN(u)\cup\{a\},v''\models\theta^*(\mathsf{form}(q',Q'\cup\{q\}))$.
   Since $\BN(u)\cup\{a\}=\BN(u\newletter a)$, this follows from the
   inductive hypothesis.
  
For the converse direction, let $\emptyset,w\models\mathsf{form}(A)$.
We have to show that there is $v\equiv_\alpha w$ such that
$v\in L_0(A)$.  We show the more general property that for all
$v',u\in\barstrings$ such that $w\equiv_\alpha uv'$, all $q\in Q$, all
$Q'\subseteq Q$ and all permutations $\pi$ such that
$\BN(u),v'\models\pi\cdot\mathsf{form}(q,Q')$, $s\stackrel{u}{\to}q$
and such that $Q'$ is the set of states of a $u$-path from $s$ to $q$,
there is $v''\in\barstrings$ such that $v''\equiv_\alpha v'$ and
$v''\in L_0(A,q)$.  The proof is by induction over $|v'|$. If
$|v'|=0$, then $v'=\epsilon$ and
$\mathsf{form}(q,Q')\in\{\epsilon,\top\}$.  If
$\pi\cdot\mathsf{form}(q,Q')= \mathsf{form}(q,Q')=\epsilon$, then
$f(q)=1$ by definition of $\mathsf{form}(q,Q')$. Hence
$\epsilon\in L_0(A,q)$.  If
$\pi\cdot\mathsf{form}(q,Q')=\mathsf{form}(q,Q')=\top$, then
$f(q)=\top$ by definition of $\mathsf{form}(q,Q')$ so that
$\epsilon\in L_0(A,q)$. If $|v'|>0$, then $v'=av'''$ or
$v''\newletter av'''$ for some $a\in\names$. Hence
$\pi\cdot\mathsf{form}(q,Q')\neq \epsilon$. If
$\pi\cdot\mathsf{form}(q,Q')=\mathsf{form}(q,Q')=\top$, then
$f(q)=\top$.  Since $\BN(u),v'\models\top$, $v'\in\free{\BN(u)}$ so
that $uv'\in\free{\emptyset}$.  Hence $v'\in L_0(A,q)$. If
$\pi\cdot\theta^*(\mathsf{form}(q,Q'))= \pi\cdot(\mu
X_q.\,\bigvee_{(q',\barname)\mid q\stackrel{\barname}{\to}q'}
\Diamond_\barname\theta^*(\mathsf{form}(q',Q'\cup\{q\})))$, then we
distinguish cases: If $v'=av'''$, then
$\BN(u),v'\models\Diamond_a\theta^*(\pi\cdot
\mathsf{form}(q',Q'\cup\{q\}))$ for some $a\in\BN(u)$ and $q'\in Q$
such that $q\stackrel{a}{\to}q'$.  By the inductive hypothesis, there
is $v''''\equiv_\alpha v'''$ such that $v''''\in L_0(A,q')$. Since
$q\stackrel{a}{\to}q'$, we have $av''''\in L_0(A,q)$. Since
$av''''\equiv_\alpha av'''$, we are done.  If $v'=\newletter av'''$,
then
$\BN(u),v'\models\Diamond_{\scriptnew c} \theta^*(\pi'\cdot
\mathsf{form}(q',Q'\cup\{q\}))$ for some $c\in\names$, some $\pi'$ and
some $q$ such that there is is some $d\in\names$ such that
$q\stackrel{\scriptnew d}{\to}q'$.  
We have $\newletter d ((ad)v''')\equiv_\alpha v'$ 
and $\Diamond_{\scriptnew c}\theta^*(\pi'\cdot
\mathsf{form}(q',Q'\cup\{q\}))\equiv_\alpha\Diamond_{\scriptnew d}\theta^*((cd)\pi'\cdot
\mathsf{form}(q',Q'\cup\{q\}))$
so that
by~\autoref{lem:alpheq:1} 
and~\autoref{lem:alpheq:2} of~\autoref{lem:alpheq},
$\BN(u),\newletter d ((ad)v''')\models\Diamond_{\scriptnew d} \theta^*((cd)\pi'\cdot
\mathsf{form}(q',Q'\cup\{q\}))$ which implies
$\BN(u)\cup\{d\},(ad)v'''\models \theta^*((dcd)\pi'\cdot
\mathsf{form}(q',Q'\cup\{q\}))$. 
 Since
$s\stackrel{u}{\to}q$ and $q\stackrel{\scriptnew d}{\to}q'$, we have
$s\stackrel{u\scriptnew d}{\to}q'$. Also
$w\equiv_\alpha u\newletter av'''\equiv_\alpha u\newletter d((ad)v''')$.
By the inductive hypothesis, there
is $v''''\equiv_\alpha (ad)v'''$ such that $v''''\in L_0(A,q')$. Since
$q\stackrel{\scriptnew d}{\to}q'$, we have
$\newletter dv''''\in L_0(A,q)$. Since
$\newletter dv''''\equiv_\alpha\newletter d((ad)v''')\equiv_\alpha \newletter av'''$, 
we are done.
 \end{proof}

\section{Details for \autoref{sec:mc}}

We first recall the definition of nondeterministic
finite bar automata from
\cite{SchroderEA17}.

\begin{defn}
A \emph{nondeterministic finite bar automaton} (bar NFA) over
$\names$ is an NFA $A$ with alphabet $\barA$.
The bar NFA $A$ \emph{literally accepts} the language
$L_0(A)$ that $A$ accepts as an NFA with alphabet $\barA$.
The \emph{bar language} accepted by $A$ is defined by
\[
  L_\alpha(A)=  L_0(A)/\mathord{\equiv_\alpha}.
\]
The degree $\degree(A)$ of $A$ is the number of
names $a\in\names$ that occur in transitions
$q\stackrel{a}{\rightarrow}q'$ or 
$q\stackrel{\scriptnew a}{\rightarrow}q'$ in $A$.
\end{defn}
%Next, we restate Theorem 5.13 from
%\cite{SchroderEA17}.
%
\begin{notheorembrackets}
\begin{theorem}[{\cite[Thm.~5.13]{SchroderEA17}}]\label{thm:RNNAbarNFA}
Given an RNNA $A_1$, there is a bar NFA $A_2$ such that
$L_\alpha(A_1)=L_\alpha(A_2)$ and
the number of states in $A_2$ is linear in the number of orbits of
$A_1$ and exponential in $\degree(A)$. Furthermore,
$\degree(A_2)\leq \degree(A_1)+1$.
\end{theorem}
\end{notheorembrackets}

%\section{Omitted lemmas and proofs}

\subsection*{Proof of \autoref{cor:mc}}
\begin{proof}
  \begin{enumerate}[wide,labelindent=0pt]
  \item The bound for model checking is immediate from
    \autoref{thm:mubar-ernna} and \autoref{thm:incl}. Under local
    freshness, validity reduces to model checking over an RNNA
    accepting the bar language $(\newletter a)^*$, whose local
    freshness semantics is the universal language.
  \item We have $\sem{\phi}\neq\emptyset$ if and only if
    $L_0(\mathsf{nd}(A(\phi)))\neq\emptyset$.  Since
    $L_0(\mathsf{nd}(A(\phi)))$ is closed under $\alpha$-equivalence,
    non-emptiness checking for this automaton amounts to checking
    whether some accepting state is reachable from its initial
    state. By Savitch's theorem, it suffices to exhibit an
    $\NExpSpace$ algorithm that performs this task.  The algorithm
    guesses a bar string $w$ and an accepting $w$-path in
    $\mathsf{nd}(A(\phi))$, letter by letter, and at the same time
    ensures that the path indeed is accepting. To ensure that the path
    is accepting, it suffices to check that each taken transition
    indeed is a valid transition in the automaton, and that the last
    state in the guessed path is an accepting state. The algorithm
    requires a singly exponential amount of memory since it suffices
    two keep two states in memory at all times when traversing the
    path, and a single state can be encoded in polynomial space.
    Hence the satisfiability problem for $\muBar$ is in
    $\ExpSpace$. The more fine-grained analysis relies on a
    correspondingly more fine-grained estimate of the representation
    size of the automaton states

    By closure under complement, validity checking under bar
    language semantics / global freshness reduces to satisfiability
    checking in the usual way.
    
    Under bar language semantics / global freshness, refinement
    reduces to validity using closure under complement and
    intersection. \qedhere
  \end{enumerate}
\end{proof}

\end{document}